\documentclass[prd,aps,floats,preprint,preprintnumbers,amssymb,longbibliography,nofootinbib]{revtex4-1}
\usepackage{graphicx}
\usepackage{microtype,graphicx}
\usepackage[justification=justified,singlelinecheck=false]{caption}
\usepackage[english]{babel}
\let\oldfootnote\footnote
\renewcommand{\footnote}[1]{%
    \begingroup%
    \linespread{1}%    % <- linespread for footnote: 1, 1.1, 1.2 etc
    \oldfootnote{#1}%
    \endgroup%
}
\usepackage{amsthm}
\usepackage{enumitem}
\usepackage{latexsym}
\usepackage{amsfonts}
\usepackage{amssymb}
\usepackage{caption}
\usepackage[dvipsnames]{xcolor}
\usepackage{CJK}
\usepackage[export]{adjustbox}
\usepackage{amsmath}
\usepackage{slashed}
\usepackage{dcolumn}
\usepackage{verbatim}
\usepackage{appendix}
\usepackage{float}
\usepackage{multirow}
\usepackage{soul}
\usepackage[normalem]{ulem}
\usepackage{ulem}
\usepackage{hyperref}
\usepackage{adjustbox}
\usepackage{natbib}
\usepackage{epsfig}
\usepackage{nicefrac}

\newtheorem{thm}{\protect\theoremname}
  \theoremstyle{plain}
\newtheorem{prop}[thm]{\protect\propositionname}
\newtheorem{defn}[thm]{\protect\definitionname}
\newtheorem{cor}[thm]{\protect\corollaryname}
\newtheorem{conj}[thm]{\protect\conjecturename}
\newtheorem{lemma}[thm]{\protect\lemmaname}
\DeclareMathOperator\arctanh{arctanh}
  
\usepackage[english]{babel}
  \providecommand{\definitionname}{Definition}
  \providecommand{\propositionname}{Proposition}
  \providecommand{\theoremname}{Theorem}
  \providecommand{\corollaryname}{Corollary}
  \providecommand{\conjecturename}{Conjecture}
  \providecommand{\lemmaname}{Lemma}

\newtheorem{definition}{Definition}

\DeclareMathOperator{\sech}{sech}

\def\ie{{\em i.e.\/}}  
\def\eg{{\em e.g.\/}}

\def\beq{\begin{equation}}
\def\eeq{\end{equation}}
\def\be{\begin{equation}}
\def\ee{\end{equation}}
\def\bea{\begin{eqnarray}}
\def\eea{\end{eqnarray}}

\renewcommand{\theequation}{\arabic{section}.\arabic{equation}}
\newcommand{\eq}[1]{Eq.~\ref{#1}}
\newcommand{\eqs}[2]{Eqs.~\ref{#1} and \ref{#2}}

\newcommand{\fig}[1]{Fig.~\ref{#1}}

\newcommand{\thref}[1]{Thm.~\ref{#1}}

\def\a{\alpha}

\begin{document}
\count\footins = 1000 
%=====TITLE AND ABSTRACT===============================================================

%\title{\Huge $\infty$ \rm}
\title{Eternal Universes}
%\title{Possibility of an eternal universe}
\author{Damien A. Easson\footnote{easson@asu.edu}}
\affiliation{
Department of Physics \& Beyond Center for Fundamental Concepts in Science,  
Arizona State University, Tempe, AZ 85287-1504, USA}
\author{Joseph E. Lesnefsky\footnote{joseph.lesnefsky@pm.me}}
\affiliation{ 
Beyond Center for Fundamental Concepts in Science,  
Arizona State University, Tempe, AZ 85287-1504, USA}

\date{\today}

\begin{abstract}
We consider the possibility of a past and future eternal universe, constructing geodesically complete inflating, loitering, and bouncing spacetimes. We identify the constraints energy conditions in General Relativity place on the building of eternal cosmological models. Inflationary and bouncing behavior are shown to be essential ingredients in all significant examples. Non-trivial complete spacetimes are shown to violate the null energy condition (NEC) for at least some amount of time, as well as the average null energy condition unless curvature is included.
Ignoring the intractable subtleties introduced by quantum considerations, such as rare tunneling events and Boltzmann brains, 
we demonstrate that these universes need not have had a beginning or an end. 
\end{abstract}
\maketitle

\tableofcontents
%\newpage
\section{Introduction}\label{intro}
Contemplating the vastness of the universe is an innate aspect of human curiosity. We are not only drawn to ponder the physical expanse of space and whether the universe stretches infinitely but also the expanse of the dimension of time. Could the universe have endured for an eternity into the past and what will be its ultimate fate? In this paper we explore these questions from a conservative viewpoint using simple arguments from General Relativity (GR) and a field theoretic treatment of matter and energy. Our goal is to be as rigorous as possible using our known, tested laws of physics. We operate within a framework of metrical geometry and apply the results to classical Einstein gravity. Within this context we provide examples of geodesically complete past and future eternal cosmological solutions. 

Whether or not the universe had a beginning is a question that naturally arises within our most successful theory of the early universe, the theory of inflation~\cite{Guth:1980zm,Linde:1981mu,Albrecht:1982wi}. Soon after the theory's inception it became clear that once inflation started, it would never stop. While early-time inflation has ceased in our observable universe, allowing galaxies and structures to form, it continues indefinitely in other parts of the universe, and hence, inflation is eternal into the future. Specifically, quantum fluctuations in the field responsible for the accelerated expansion can cause the field to remain in its high-energy state, sustaining inflation in those regions indefinitely~\cite{Vilenkin:1983xq,Linde:1986fd,Guth:2007ng}. In a Letter by Borde, Guth and Vilenkin (BGV) \cite{Borde:2001nh}, it was subsequently argued that inflation cannot likewise be eternal into the past; although, we have countered this belief \cite{Easson:2024uxe,Lesnefsky:2022fen} and present further evidence that inflation can, in fact, be eternal into the past in the current work. Acknowledging the possibility of a past eternal universe opens entirely new avenues of research.  We show there are multiple ways that the universe can sustain itself within classical General Relativity for an eternity, albeit at the expense of violation of certain energy conditions.

Three broad categories of an eternal universe are suggested. Our examples are admittedly toy models; however, given the wide variety we expect our findings to be quite general and applicable to many solutions of interest. 
%yet some stand in stark contrast to proposed no-go theorems in the literature. 
In particular, we will construct eternal inflationary, loitering, and bouncing cosmologies. In each case we will provide a rigorous proof of full geodesic completeness. Existing proposed no-go theorems are circumvented or obviated as discussed in \cite{Easson:2024uxe}. 
% We prove, that within the confines of classical GR, there exist loitering, bouncing and inflationary models that are eternal into the past and future. 
We show that all (non-trivial) models that are eternal geodesically complete necessarily involve a phase of accelerated expansion (or inflationary phase) for at least some amount of time and generically include a bouncing phase in accordance with our reported definitions.

The models we present violate the traditional energy conditions, a property we conjecture is necessary in GR for geodesic completeness.~\footnote{A contained discussion of energy conditions is offered in Appendix~\ref{appa}.} In some instances violations of the null energy condition (NEC) can be limited to an arbitrarily short amount of time $\Delta t$, and confined to the primordial past;~\footnote{In its most provocative form, the Heisenberg uncertainty principle is written $\Delta E \Delta t \ge \hbar/2$. 
%In general violations of the NEC can in some circumstances be limited to arbitrarily short times $\Delta t$. 
If the time interval is made very small the uncertainty principle may result in a very large energy fluctuation which can back-react on the spacetime.}
Although, the models presented here ultimately violate the average null energy condition (ANEC) in a flat universe.

In several known circumstances, energy condition violations can be relatively benign. For example, a simple time-independent cosmological constant $\Lambda$ violates the SEC if $\Lambda>0$ and the WEC and DEC for $\Lambda<0$. 
The NEC is the only energy condition which remains valid in the presence of $\Lambda$. Non-minimally coupled scalar fields are capable of violating all of the traditional energy conditions \cite{Barcelo:2002bv,Chatterjee:2012zh}. Indeed, all classical energy conditions are manifestly violated by quantum effects both experimentally and theoretically \cite{Epstein:1965zza}, (e.g.,~the Casimir Effect \cite{Casimir:1948dh}). Even the NEC is violated ubiquitously in particle physics, for example, during the Hawking radiation process local energy density can become negative from the perspective of certain observers and interpreted as violation of NEC. In theories outside of pure GR, which can mix gravitational and matter degrees of freedom in a non-trivial way through non-minimal and higher-derivative terms or couplings, it becomes difficult to adequately define analogues of energy conditions in the first place (see \eg, \cite{Chatterjee:2012zh}).

As is well known, violation of energy conditions are necessary in order to avoid singularities. Within classical GR non-trivial flat cosmological solutions which avoid singularities \it must \rm violate energy conditions to circumvent the singularity theorems~\cite{Penrose:1964wq,Hawking:1966sx,Hawking:1966jv,Hawking:1967ju,Hawking:1970zqf,Hawking:1973uf}. On occasion, the NEC can be violated in a stable way~\cite{Dubovsky:2005xd, Nicolis:2009qm, Kobayashi:2010cm,Deffayet:2010qz,Sawicki:2012pz,Rubakov:2014jja,Cai:2016thi,Creminelli:2016zwa,Cai:2017tku,Easson:2018qgr,Alexandre:2021imu,Alexandre:2023iig}, sometimes allowing for well-behaved cosmological bouncing solutions~\cite{Easson:2011zy,Cai:2012va, Easson:2016klq,Ijjas:2016tpn, Ijjas:2016vtq,Cai:2017dyi,Alexandre:2023pkk}.

%ANEC

Currently known consistent quantum field theories in flat spacetime obey the \emph{averaged} null energy condition (ANEC), for a complete achronal null geodesic (those that do not contain any points connected by a timelike curve)~\cite{Wald:1991xn,Graham:2007va, Hartman:2016lgu}. Effectively, the ANEC condition quantifies the degree of violation of the null energy condition. The condition is expected to rule out closed timelike curves, time machines and wormholes connecting different asymptotically flat regions \cite{Morris:1988tu,Friedman:1993ty}. Matter violating the ANEC can be used to violate the second law of thermodynamics \cite{Wall:2009wi}.

The ANEC may be stated 
\begin{equation}\label{anec}
    \int_\gamma T_{\mu\nu} k^\mu k^\nu d\lambda \ge 0 \,,
\end{equation}
where the integral is over a null geodesic $\gamma$ with tangent vector $k^\mu = dx^\mu/d\lambda$, and $\lambda$ is an affine parameter with respect to which the tangent vector to the geodesic is defined. In GR, it is frequently convenient to replace $T_{\mu \nu}$ in \eq{anec} by the Ricci tensor $R_{\mu\nu}$, simply using the Einstein field equations and the definition of a null vector. The ANEC has been shown to hold even in the case of arbitrary Casimir systems as long as the geodesic does not intersect or asymptotically approach the plates \cite{Graham:2005cq,Fewster:2006uf}. If the null geodesic is not assumed to be achronal, the ANEC may be violated by a quantum scalar field in a spacetime compactified in one spatial dimension or in the spacetime around a Schwarzschild black hole \cite{Visser:1996iv}. 
%As we shall see some of the eternal models we construct are capable of satisfying the ANEC in a closed universe.

Given the many nuances involved with respect to energy condition violation and stability, it can be difficult to pass definitive judgment on the validity of solutions which experience such violations. In this work we refrain from making such criticisms as more rigorous analysis of specific situations is needed; however, it is clear that at least some violations do not lead to pathologies and ultimately singularity formation is arguably a worst-case malady plaguing a given physical solution. Yet many physically relevant and useful solutions in GR contain such singularities, even at the level of the vacuum solutions, for example, the singular Schwarzschild and Kerr black hole solutions.

After assembling a diverse grouping of possible eternal universe models, our analysis demonstrates that there are no known obstructions indicating the universe must have had a ``first moment". In the context of the inflationary universe paradigm, this translates into the statement that inflation need not have had a beginning.

%This is the companion paper to our Letter \cite{Easson:2024uxe}. 
A few of the Theorems and Corollaries we discuss were presented without proof in our earlier work \cite{Lesnefsky:2022fen,Easson:2024uxe}. The proofs are revealed presently.

An outline of this paper is as follows. In section \ref{geodcomp}, we prove a theorem which definitively provides the criteria for geodesic completeness of generalized Friedmann-Robertson-Walker (GFRW) spacetimes. In section \ref{inflate}, we discuss the definition of inflation and bouncing cosmologies and provide an example of an eternally inflating universe which is eternal and geodesically complete both into the past and future. In section \ref{bounce}, we provide examples of eternal bouncing cosmologies which are geodesically complete both into the past and future. In section \ref{loiter}, we provide examples of loitering (and quasi-loitering) models which are eternal and geodesically complete both into the past and future.
%In section \ref{nogo}, we discuss how the various no-go theorems in the literature are circumvented in all of the eternal universe models we produce. 
Finally, in section \ref{conclude}, we present a discussion and our conclusions. We provide appendices to discuss the energy conditions in General Relativity with respect to cosmology, to provide proof details, to define Generalized Friedmann-Robertson-Walker (GFRW) spacetimes, and to present a detailed discussion of geodesics for several models.

\section{Geodesic Completeness}\label{geodcomp}
We begin with a rigorous discussion of what it means to be geodesically complete within the context of GFRW spacetimes (which include the familiar FRW spacetimes).  Stated formally, a GFRW spacetime is the warped product - see \cite{Bishop1969,ONeill1983} - given by $\mathbb{R}^1_1 \times_a \Sigma$ where $\left( \Sigma , g_\Sigma \right)$ is a complete Riemannian manifold and scale factor $a \in \mathcal{C}^\infty \left( \mathbb{R}^1_1 \right)$ which is smooth and $a > 0$.  To fully explore the geodesic completeness of a spacetime manifold $\mathcal{M}$, we must explore \emph{maximal} geodesics: consider a maximal geodesic \( \gamma: I \to \mathcal{M} \), where \( I \) is an \emph{open} interval of \( \mathbb{R} \), and, because geodesics are usually parameterized with constant speed, it is uniquely defined up to transversality.   A manifold is \emph{geodesically complete} if every geodesic \( \gamma \), \( I = (-\infty, \infty) \), is defined for all time (see e.g. \cite{ONeill1983}, Prop. 7.38).

We now construct a theorem quantifying the criteria for geodesic completeness of a GFRW spacetime having scale factor $a$. A geodesic $ \gamma \left( \lambda \right) = \left( t \left( \lambda \right) , \beta \left( \lambda \right) \right) $ in a warped product spacetime of\footnote{Here we utilize the hyperquadratic notation of \cite{ONeill1983} where $\mathbb{R}^n_k$ is the topological space $\mathbb{R}^n$ with semi-Riemannian metric of signature $\{ -1 , \overset{k-2}{\cdots} , -1 , +1 , \overset{n-k-2}{\cdots} , +1 \}$.} $\mathbb{R}^1_1 \times_a \Sigma$ obeys
\begin{equation}
   \nabla_{\frac{\partial}{\partial t}} \frac{\partial}{\partial t} = g_\Sigma \left( \beta ' , \beta ' \right) \left( a \left( t \left( \lambda \right) \right) \right)\left( \mathrm{grad}_{-dt^2} a \right) \label{eq:gfrw geodesic equation1} \,,
\end{equation} 
in $\mathbb{R}^1_1$ and 
\begin{equation}
  \nabla_{\beta '} \beta ' = -\frac{2}{a \left( t \left( \lambda \right) \right)} \left( \beta ' \right) \frac{d}{d \lambda} \left[ a \left( t \left( \lambda \right) \right) \right] \label{eq:gfrw geodesic eq2} \,,
\end{equation}
in Riemannian manifold spacelike leaf $\left( \Sigma , g_\Sigma \right)$ where the geodesic equation $\nabla_{\gamma '} \gamma ' = 0$ has been projected down into the respective spaces $\mathbb{R}^1_1$ and $\left( \Sigma , g_\Sigma \right)$.  The coupled ordinary differential equations (ODE) can be solved exactly yielding an integral solution as a functional of $a$.   
With appropriately selected initial conditions, one finds
\begin{equation}
    \lambda \left( t \right) = \int^{t} \frac{a \left( \zeta \right)}{\sqrt{a^2 \left( \zeta \right) + 1}} d \zeta \label{eq:gfrw soltn timelike integral} \,,
\end{equation}
as an integral solution to \eqs{eq:gfrw geodesic equation1}{eq:gfrw geodesic eq2} for timelike initial conditions and 
\begin{equation}
    \lambda \left( t \right) = \int^{t} a \left( \zeta \right) d \zeta \label{eq:gfrw soltn null integral} \,,
\end{equation}
for null initial conditions.  To solve the ODE we invert $\lambda \left( t \right)$, instead of using $t \left( \lambda \right)$; which is auspicious, because if \eqs{eq:gfrw soltn timelike integral}{eq:gfrw soltn null integral} diverge then the domain of $\gamma$ is infinite and the geodesic ray has domain $\left[ 0 , \infty \right)$ and is complete.

Eq. \ref{eq:gfrw geodesic eq2} defines a pregeodesic and thus has a geodesic reparameterization.  Appropriately reparameterizing, there is a geodesic in $\Sigma$ which has $\beta$ as its image: thus, if $\left( \Sigma , g_\Sigma \right)$ is complete as a Riemannian manifold, then the limiting factor to geodesic completeness is Eq.~\ref{eq:gfrw geodesic equation1} with solutions of \eqs{eq:gfrw geodesic equation1}{eq:gfrw geodesic eq2}.  The assumption of ``generality'' in a GFRW is the least restrictive logical assumption needed to have \eqs{eq:gfrw geodesic equation1}{eq:gfrw geodesic eq2}  characterize geodesic completeness of GFRWs.
Hence, the criteria for geodesic completeness of a GFRW spacetime is given by Thm.~2 of \cite{Lesnefsky:2022fen} (which builds on previous work \cite{Sanchez1998}). The theorem states:
 \begin{thm}{}
        Let $\mathcal{M} = \mathbb{R}^1_1 \times_a \Sigma$ be a GFRW spacetime.
        \begin{enumerate}
            \item The spacetime $\mathcal{M}$ is future timelike complete iff $ \int_{t_0}^\infty \frac{a \left( \zeta \right) d\zeta}{\sqrt{\left( a \left( \zeta \right) \right)^2 + 1}}$ diverges for all $t_0 \in \mathbb{R}$. 
            \item The spacetime $\mathcal{M}$ is future null complete iff $ \int_{t_0}^\infty a \left( \zeta \right) d\zeta$ diverges for all $t_0 \in \mathbb{R}$. 
            \item The spacetime $\mathcal{M}$ is future spacelike complete iff $\mathcal{M}$ is future null complete and $a < \infty$. 
            \item The GFRW is past timelike / null / spacelike complete if, for items 1-3 above, upon reversing the limits of integration from $\int_{t_0}^\infty$ to $\int_{-\infty}^{t_0}$ the word ``future'' is replaced by ``past''.
            \item The spacetime $\mathcal{M}$ is \emph{geodesically complete} iff it is both future and past timelike, null, and spacelike geodesically complete.
        \end{enumerate}
        \label{ledthm}
    \end{thm}

    \begin{proof}
    The definition of a geodesically complete manifold \(\mathcal{M}\) is a manifold where all geodesics are defined for all \(\mathbb{R}\). Given the structure of \(\mathbb{R}^1_1 \times_a \Sigma\) the geodesic equation projections of Eqs. \ref{eq:gfrw geodesic equation1},\ref{eq:gfrw geodesic eq2} reduce to an integral over \(\mathbb{R}^1_1\). One must discriminate the space \(\mathbb{R}^1_1 \times_{\exp t} \Sigma\), appropriate for the cosmological flat slicing of de Sitter space,  which is known to be past geodesically incomplete (see \eg \cite{ONeill1983}, Example 7.41). The integration \(\int_{\mathbb{R}^1_1} \frac{d\zeta a}{\sqrt{a^2 + 1}}\) diverges, yet the model is past incomplete, in particular, the mass density from the (complete) future divergence obfuscates the mass density of the (incomplete) past convergence.

Using future and / or past directed maximal geodesic rays, for example \(\gamma : \left[ 0 , b \right) \rightarrow \mathcal{M}\), completeness requires for all initial points \(p \in \mathcal{M}\), the endpoint \(b \rightarrow \infty\). In the case of a geodesic \(\gamma \left( \lambda  \right) = t \left( \lambda \right) \oplus x^k \left( \lambda \right)\) of \(\mathbb{R}^1_1 \times_a \Sigma\), integral divergence is required for all future directed geodesic rays 
\begin{equation}\label{geodrayf}
    \int_{t_0}^\infty \frac{d\zeta a \left( \zeta \right)}{\sqrt{w_0 a^2 \left( \zeta \right) + 1}} \,,
\end{equation}
and past directed geodesic rays
    \begin{equation}\label{geodrayp}
    \int_{-\infty}^{t_0} \frac{d\zeta a \left( \zeta \right)}{\sqrt{w_0 a^2 \left( \zeta \right) + 1}} \,,
\end{equation}
because in both of these cases the integral computes the range of the affine parameter \(\lambda \left( t \right)\) and if \(\lambda \left( t \right)\) diverges, no finite \(b\) exists for domain \(\left[ 0 , b \right)\). Independent of causal character this proves item (4). The integral constant \(w_0\) encapsulates causal character with \(w_0 \in \left\{ +1 , 0 , -1 \right\}\) for timelike, null, and spacelike, respectively. Items (1), (2) trivially follow from evaluation of \(w_0 = +1\) and \(w_0 = 0\). 

The proof of Item (3) pertaining to spacelike completeness is less straightforward and not relevant for our primary discussion and hence, relegated to Appendix~\ref{proofof3}. 

\end{proof}

\subsection{Singularities}
In our effort to build geodesically complete spacetimes, we pause to further consider the catastrophes which can prevent completeness. Perhaps the most common object leading to incomplete geodesics is the curvature singularity, where a curvature invariant such as the Kretschmann scalar $K= R_{\mu\nu\kappa\sigma}R^{\mu\nu\kappa\sigma}$, built from the Riemann tensor, diverges at a particular place, or time. Infamous examples of curvature singularities lie at the center of the Schwarzschild black hole or at the cosmological Big-Bang. Generally, curvature singularities occur in regions of spacetime with high curvature, where quantum gravity effects are expected to be significant. In some of our models, it is possible to arrange parameter values so that the energy and curvature scales remain significantly lower than the Planck regime for the entire cosmological evolution, thereby avoiding conditions where quantum gravity is necessary. Consequently, quantum gravity might not play a critical role in the development of an eternal universe.

In order to avoid cosmological singularities it is necessary to violate classical energy conditions in GR. These violations are anticipated to manifest close to the would-be singularity formation points and, occasionally, can be confined to an arbitrarily brief period. Arguably these energy condition violations are a lesser inconvenience when weighed against the formation of singularities, at which point all established physical laws cease to apply.

Of course, it is possible for a spacetime to be geodesically incomplete even without encountering a curvature singularity. The most renowned example is the incompleteness encountered in the flat cosmological slicing of de Sitter space, where an observer following a geodesic into the past encounters a boundary in finite proper time (for a recent discussion see \cite{Kinney:2023urn,Geshnizjani:2023hyd}). 

Yet another, arguably less common malady is a \emph{big rip}, \cite{Caldwell:2003vq}, a catastrophic ``phantom" force that eventually tears the universe apart. This type of cataclysm can occur in models which are dominated by NEC violating matter for a sufficiently long period of time. As a demonstration we construct the following example. We adopt  a perfect-fluid description of matter with equation of state (EOS)
\begin{equation}
    p = w \rho \,,
\end{equation}
where $w$ is the equation of state parameter relating the pressure $p$ to the energy density $\rho$ of the fluid. We now force this fluid to violate the NEC for all time, thus also violating the ANEC. From the NEC condition $\rho + p = \rho (1+w) \ge 0$ (see Append.~\ref{appa}), the fluid will have EOS parameter $w<-1$ in order to violate NEC. We take a flat $k=0$ FRW metric, and since $\dot H \propto - \rho(1+w)$, we find NEC violation implies $\dot H >0$ for $\rho>0$.  We now take $\dot H = \ddot a/a - (\dot a/a)^2 =  2/\beta^2$, for a positive constant $\beta$, and integrate to get 
\begin{equation}
   a(t) = a_0 \, \exp{ \left[ \frac{t}{\alpha} + \left(\frac{t}{\beta} \right)^2 \right] } \,.
\end{equation}
This spacetime is eternally super-inflating with $\dot H>0$ and an energy density which grows without bound from $t=0$ to $t \rightarrow \infty$ and violates the NEC for all time. Such a universe cannot sustain itself as it will inevitably encounter a big rip. The Ricci scalar, for example, is given by 
\begin{equation}
   R = 12\left(\frac{1}{\beta^2} + \frac{(\beta^2 + 2 \alpha t)^2}{\alpha^2 \beta^4}\right) \,,
\end{equation}
which diverges in a singular fashion as $t \rightarrow \infty$. The Kretschmann scalar $K= R_{\mu\nu\kappa\sigma}R^{\mu\nu\kappa\sigma}\propto t^4$, likewise diverges at $t \rightarrow \infty$.  Geometrically, the geodesically connected domains\footnote{For an FRW spacetime, the observable universe (for a given observer) is the geodesically connected domain of said observer at that point.  The geodesically connected domain of a point $p \in \mathcal{M}$ consists of all points which can be reached by a geodesic.  In contrast to a geodesically complete Riemannian manifold, where any two points are guaranteed to be connected by a geodesic, many geodesically complete Lorentzian manifolds are \emph{not} geodesically connected.} of any point become arbitrarily small as  the scale factor and curvatures become arbitrarily large.  Eventually, the geodesically connected domain becomes smaller than a Planck volume, although the affine parameter of a geodesic remains well defined.  Hence, in general, it is difficult to construct a physically reasonable eternal cosmological model with such extended period of NEC violation.

To summarize, a geodesic \( \gamma: I \to \mathcal{M} \) may fail be be complete -- where $I \ne \mathbb{R}$ -- for any of the above reasons, including a curve encountering a curvature singularity at a (topological) limit point, a causal curve encountering a spacelike boundary $\partial \mathcal{M}$ at a limit point, or light cones on a geodesic ``tipping over'' where said curve changes causal character.~\footnote{For a more complete discussion of singularities and geodesic incompleteness see \cite{Hawking:1973uf}.}

\subsection{Inflationary and bouncing cosmologies}
In preparation for our eternal universe constructions and to further hone our discussion, we offer the following definitions. 
\begin{definition}
Let $\left( \mathcal{M}, g \right)$ be an $n$ dimensional spacetime which admits a connected open neighborhood $U \subset \mathcal{M}$  which is isometric to $ \left( b,c \right) \times_a V$ as a warped product open submanifold, where $\left( b,c \right)$ is a timelike codimension $n-1$ embedded submanifold, $V$ is a connected spacelike codimension 1 embedded submanifold and no assumptions are made concerning $a(t)$ except that it is a well defined function between sets.  The spacetime $\left( \mathcal{M}, g \right)$ is an \emph{inflationary spacetime} if there exists some $t_0 \in \left( b,c \right)$ such that (assuming that it exists and is well defined), $\ddot{a} \left( t_0 \right) > 0$ where derivatives of the scale factor are taken with respect to the timelike $\left( b,c \right)$ coordinate and the sign of $\ddot{a}$ is given with respect to the time orientation of $\mathcal{M}$.  \label{LEDinflationdefn}
\end{definition}
Note, in the above definition of inflation we make no assumptions about the continuity of the scale factor $a$ or its derivatives. Furthermore, this definition includes arbitrarily short periods of accelerated expansion as inflationary. We are aware some readers will find this objectionable; however, if one wishes to distinguish the terms ``short accelerated expansion" from ``inflation", one is obliged to set an arbitrary scale for the duration of the acceleration. While GUT-scale inflation naturally requires at least 60 e-foldings of accelerated expansion to solve the flatness and horizon problems, low-energy (\eg,~TeV-scale) inflation does the job with far fewer. It is only a few e-foldings of accelerated expansion that are actually probed in the cosmic microwave background radiation (CMB). 
In the present context of an eternal universe the difference between a few e-folding, and 60 e-foldings is entirely insignificant.
For these reasons we find it unnecessary to reserve the term inflation to designate significant prolonged expansion and use the term interchangeably with accelerated expansion as is appropriate considering
Defn.~\ref{LEDinflationdefn}. 
This is the definition of ``inflation" from \cite{Lesnefsky:2022fen}, which we have reproduced here to align with our present notation. 

We further suggest the following definition for a ``bouncing" spacetime:
\begin{definition}
Let $\left( \mathcal{M}, g \right)$ be an $n$ dimensional spacetime which admits a connected open neighborhood $U \subset \mathcal{M}$  which is isometric to $ \left( b,c \right) \times_a V$ as a warped product open submanifold, where $\left( b,c \right)$ is a timelike codimension $n-1$ embedded submanifold, $V$ is a connected spacelike codimension 1 embedded submanifold and no assumptions are made concerning $a(t)$ except that it is a well defined non-constant function between sets.  The spacetime $\left( \mathcal{M}, g \right)$ is a \emph{bouncing spacetime} if there exists some $t_0 \in \left(  b,c \right)$ such that (assuming that it exists and is well defined), $\dot{a}(t_0) = 0$, where derivatives of the scale factor are taken with respect to the timelike $\left( b,c \right)$ coordinate and the sign of $\dot{a}$ is given with respect to the time orientation of $\mathcal{M}$.  \label{ELbouncedefn}
\end{definition}

We would like to highlight the fact that, in the the usual consideration of a universe shrinking to a turning point and then enlarging, Defn. \ref{ELbouncedefn} also considers an expanding universe reaching a turning point and then contracting a bounce as well.

As we shall discover, in some cases it is advantageous to discuss a ``bounce at infinity''.  Considering that the point at infinity is strictly not an element\footnote{One solution to this is to compactify.  There are many compactification prescriptions, but typically one utilizes the Hausdorff single point compactification of $\mathbb{R} \hookrightarrow S^1$.  Note, however, that a Hausdorff single point compactification does not possess the universal property of preserving continuity upon function extension.} of $\mathbb{R}$, we suggest:

\begin{definition} \label{defn: bounce at infinity}
    A neighborhood $\left( b , c \right) \times_a V \subset \left( \mathcal{M} , g \right)$ is an open submanifold which is isometric to a GFRW $\left( b , c \right) \times_a \Sigma \subset \mathcal{M}$ is said to admit a \emph{bounce at positive infinity} $+ \infty$ if, upon extension if necessary, $c \rightarrow + \infty$ is unbounded above and there exists a divergent sequence $\left\{ t_\ell \right\}_{\ell = 1}^\infty \rightarrow + \infty$ with $t_k > t_\ell$ for $k > \ell$ where $\left\{ \dot{a} \left( t_\ell \right) \right\}_{\ell = 1}^\infty \rightarrow 0$ constitutes a convergent Cauchy sequence for non-constant $a \left( t \right)$.

    A neighborhood $\left( b , c \right) \times_a V \subset \left( \mathcal{M} , g \right) \subset \mathcal{M}$ is an open submanifold which is isometric to a GFRW which admits a \emph{bounce at negative infinity} $- \infty$ if, upon extension if necessary, $b \rightarrow - \infty$ is unbounded below and there exists a divergent sequence $\left\{ t_\ell \right\}_{\ell = 1}^\infty \rightarrow - \infty$ with $t_k < t_\ell$ for $k > \ell$ where $\left\{ \dot{a} \left( t_\ell \right) \right\}_{\ell = 1}^\infty \rightarrow 0$ constitutes a convergent Cauchy sequence.

    A neighborhood $\left( b , c \right) \times_a V \subset \left( \mathcal{M} , g \right) \subset \mathcal{M}$ is an open submanifold which is isometric to a GFRW which admits a \emph{bounce at infinity} if it admits a bounce at positive infinity and / or a bounce at negative infinity.
\end{definition}

In this definition we require the Hubble parameter $H = \dot{a}/a$ be non-zero for some time on the interval $\left( b,c \right)$ to exclude referring to pure Minkowski as a bouncing spacetime. (One may likewise exclude Minkowski spacetime by requiring there be no timelike killing vector for $\left( \mathcal{M}, g \right)$.) In the following we shall casually refer to a spacetime as a bouncing spacetime if it incorporates either a bounce at finite $t$, or a bounce at infinity in accordance with the above definitions.

We are now in a position to build diverse examples of geodescially complete spacetimes. 
\section{Eternal inflating universe}\label{inflate}

Our first example of an eternal universe was introduced in \cite{Lesnefsky:2022fen}, and analyzed in detail in \cite{Easson:2024uxe}:
\begin{eqnarray}\label{aplusc}
    a(t) = a_0 \, e^\frac{2 t}{\alpha} + c \,,
\end{eqnarray}
for real constants $a_0 > 0$, $\alpha$, and $c > 0$. 

Application of Thm.~\ref{ledthm}, finds that the scale factor (\ref{aplusc}) with $c>0$, defines a geodesically complete spacetime for all (real) values of the parameter $\alpha$. The theorem integrals are explicitly calculated in
\cite{Easson:2024uxe}.
% By setting $a_0 = c = 1/2$ it is  easy to show (\ref{aplusc}) may equivalently be expressed as
% \begin{equation}
%     a(t) = \exp[t/\alpha] \cosh[t/\alpha] \,.
% \end{equation}
The Hubble parameter $H = \dot a/a$ is given by:
\begin{equation}
    H = \frac{2 a_0 e^{\frac{2t}{\alpha}}}{(c + a_0 e^{\frac{2t}{\alpha}}) \alpha} \,.
\end{equation}
This is an example of an eternally inflating spacetime in accordance with Defn.~1:
\begin{equation}
  \frac{\ddot a}{a} =  4 \alpha^{-2} \left(1 - \frac{c}{c + a_0 e^{\frac{2t}{\alpha}}}\right) \,.
\end{equation}
The acceleration is positive for all $t \in (-\infty, \infty)$ for all non-zero values of $\alpha$ with $c>0$. For $\alpha>0$, the universe is expanding and inflating for all $t$.
For $\alpha<0$ the universe is contracting yet inflating ($\ddot a>0$) for all time, approaching zero acceleration as $t \rightarrow \infty$.

% For the scale factor (\ref{aplusc}), it is possible to explicitly calculate the integrals of Theorem 1. We find for the indefinite integrals:
% \begin{align}
% \int \frac{a(t)}{\sqrt{(a(t))^2 + 1}} \, dt &= \nonumber \\
% &= \frac{\alpha \arctan\left(\frac{-\sqrt{a_0^2} e^{\frac{2t}{\alpha}} + \sqrt{1 + (c + a_0 e^{\frac{2t}{\alpha}})^2}}{\sqrt{-1 - c^2}}\right)}{\sqrt{-1 - c^2}} \nonumber \\
% &\quad - \frac{1}{4a_0} \alpha \left(a_0 \log\left(-c - \sqrt{a_0^2} e^{\frac{2t}{\alpha}} + \sqrt{1 + (c + a_0 e^{\frac{2t}{\alpha}})^2}\right) \right. \nonumber \\
% &\quad \left. + \sqrt{a_0^2} \log\left(a_0 (-c - \sqrt{a_0^2} e^{\frac{2t}{\alpha}} + \sqrt{1 + (c + a_0 e^{\frac{2t}{\alpha}})^2})\right) \right. \nonumber \\
% &\quad \left. + (-a_0 + \sqrt{a_0^2}) \log\left(c - \sqrt{a_0^2} e^{\frac{2t}{\alpha}} + \sqrt{1 + (c + a_0 e^{\frac{2t}{\alpha}})^2}\right)\right)
% \end{align}
% and
% \begin{equation}
% \int a \left( t \right) dt= ct + a_0 \frac{\alpha}{2} e^{\frac{2t}{\alpha}} \,. 
% \end{equation}

% It is easy to show the above integrals diverge over the full set of conditions discussed in Theorem 1 for all (non-zero) values of $\alpha$; hence, the spacetime with scale factor (\ref{aplusc}) is geodesically complete. 
The cosmological evolution is depicted in \fig{ceplusc}. Shown are the scale factor $a(t)$, Hubble parameter $H$ (amplified by an order of magnitude) and co-moving Hubble radius $1/a H$. When the co-moving Hubble radius is decreasing the spacetime is inflating--in this case it is decreasing for all $t$.

\begin{figure}[H]
\centering
%\begin{minipage}[t]{.6\textwidth}
 % \centering
  \includegraphics[width=0.8\linewidth]{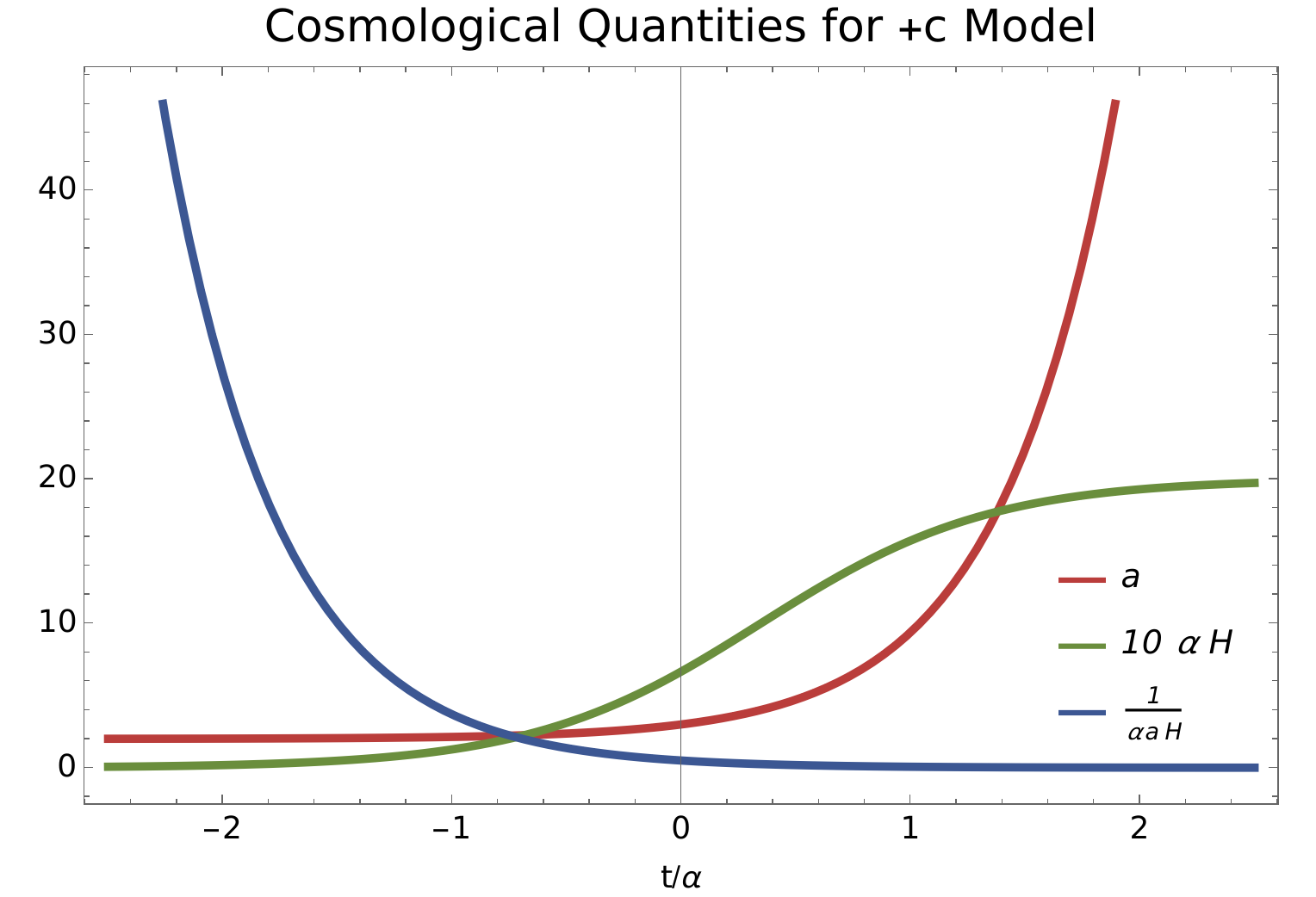}
  \captionof{figure}{Evolutionary behavior for (\ref{aplusc}). The scale factor is plotted in red. The Hubble parameter (green), scaled by $10 \times$ for readability. The co-moving Hubble radius (blue). Model parameters: $c = 2$, $\a_0=\alpha=1$.}
  \label{ceplusc}
%\end{minipage}%
%\begin{minipage}[t]{.5\textwidth}
%  \centering
%  \includegraphics[width=.9\linewidth]{mTilldefor111.jpg}
%  \captionof{figure}{$\tilde{m}_{4}^{2}$ for the case $b=d=\Lambda=1.$}
%  \label{mtilde2}
%\end{minipage}
\end{figure}

The spacetime is geodesically complete, eternal and inflating, and provides a simple counterexample to the statement that inflationary models must be past incomplete, regardless of energy condition considerations (see, \cite{Borde:2001nh}). It is only one example of an infinite number of models with such behavior as we shall see in the proceeding discussion.~\footnote{The above model \eq{aplusc} is not special. For example, it shares many features of the scale factor used in the \emph{Emergent Universe} scenario~\cite{Ellis:2002we}.}
For $\alpha>0$, the universe \emph{approaches} Minkowski spacetime as $t \rightarrow - \infty$; however, an observer travelling along a past-directed geodesic will experience an infinite proper time to reach the non-accelerating Minkowski space at the point $t=-\infty$ and hence, the model is an eternal inflating spacetime. As we shall see, another eternal inflating example having $\ddot a>0$ $\forall \, t \in \left(-\infty,\infty\right)$, is given by \eq{apoly}.

One may identify the absence of curvature singularities from the curvature invariants, as shown in \fig{csplusc}. We plot the Ricci scalar $R$ and the Kretchmann scalar $K = R_{\mu\nu\rho\sigma}R^{\mu\nu\rho\sigma}$ built from the Riemann tensor. Both are observed to be finite for all $t$:
\begin{equation}
    R = \frac{24 a_0 e^{\frac{2t}{\alpha}} (c + 2 a_0 e^{\frac{2t}{\alpha}})}{(c + a_0 e^{\frac{2t}{\alpha}})^2 \alpha^2}
\,, \qquad K =\frac{192 a_0^2 e^{\frac{4t}{\alpha}} (c^2 + 2 a_0 c e^{\frac{2t}{\alpha}} + 2 a_0^2 e^{\frac{4t}{\alpha}})}{(c + a_0 e^{\frac{2t}{\alpha}})^4 \alpha^4}
 \,.
\end{equation}
\begin{figure}[H]
\centering
%\begin{minipage}[t]{.5\textwidth}
 % \centering
  \includegraphics[width=0.8\linewidth]{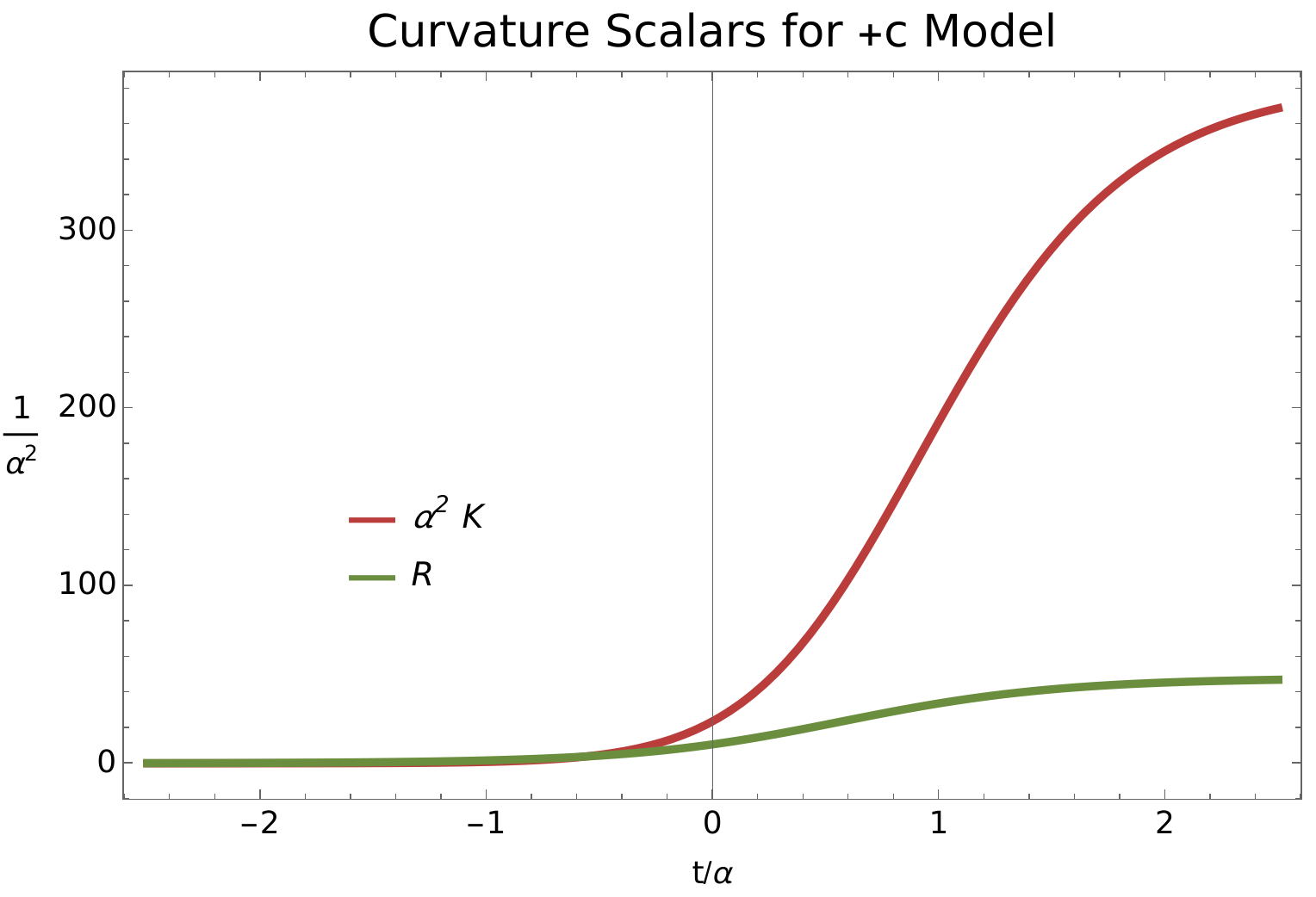}
  \captionof{figure}{Curvature invariants computed from (\ref{aplusc}). The Ricci scalar is plotted in green and the Kretchmann scalar is plotted in red. Model parameters: $c = 2$, $\a_0=\alpha=1$.}
  \label{csplusc}
%\end{minipage}%
%\begin{minipage}[t]{.5\textwidth}
%  \centering
%  \includegraphics[width=.9\linewidth]{mTilldefor111.jpg}
%  \captionof{figure}{$\tilde{m}_{4}^{2}$ for the case $b=d=\Lambda=1.$}
%  \label{mtilde2}
%\end{minipage}
\end{figure}
%\subsection{Energy conditions}
We now examine the energy conditions for the model given by (\ref{aplusc}). Calculation of the Einstein tensor yields non-vanishing components:
\begin{eqnarray}\label{etplusc}
   G_{tt} &=& \frac{12 a_0^2 e^{\frac{4t}{\alpha}}}{(c + a_0 e^{\frac{2t}{\alpha}})^2 \alpha^2} \,, \nonumber \\
   G_{ii} &=& -\frac{4 a_0 e^{\frac{2t}{\alpha}} (2c + 3a_0 e^{\frac{2t}{\alpha}})}{\alpha^2} \,.
\end{eqnarray}
The energy density is given by $\rho = - G^t{}_t$ and the pressure is $p = G^i{}_i$. A plot depicting the energy conditions is given in \fig{ecplusc}
\begin{figure}[H]
\centering
%\begin{minipage}[t]{.5\textwidth}
 % \centering
  \includegraphics[width=0.8\linewidth]{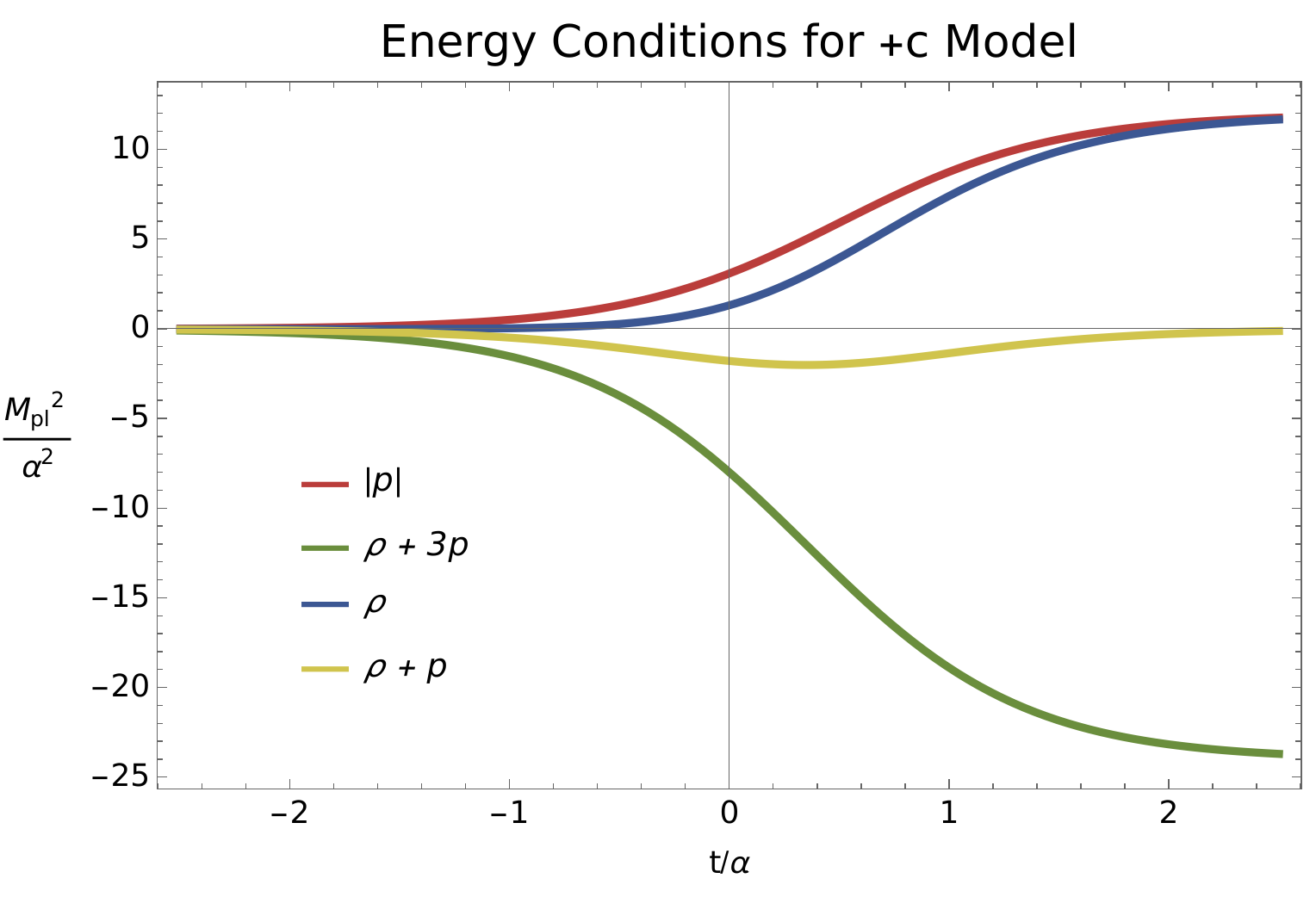}
  \captionof{figure}{Energy conditions from (\ref{ecplusc}). Plot of energy density $\rho$ (blue), $\rho + p$  (yellow), $|p|$ (red) and $\rho + 3p$ (green). Model parameters: $c = 2$, $\a_0=\alpha=1$.}
  \label{ecplusc}
%\end{minipage}%
%\begin{minipage}[t]{.5\textwidth}
%  \centering
%  \includegraphics[width=.9\linewidth]{mTilldefor111.jpg}
%  \captionof{figure}{$\tilde{m}_{4}^{2}$ for the case $b=d=\Lambda=1.$}
%  \label{mtilde2}
%\end{minipage}
\end{figure}
All of the energy conditions are violated. NEC violation is identified by the dipping of the yellow curve ($\rho + p$) below the vertical axis, both here, and in all of the proceeding energy condition plots. The NEC asymptotically approaches saturation in the early and late universe since $\rho + p = -2 \sech^2{t} \rightarrow 0$ as $t \rightarrow \pm \infty$. It is easy to see the ANEC (\eq{anec}) is also violated in this scenario. We will see that the ANEC violation is a generic feature of flat eternal models. In the late universe, the SEC remains violated allowing continued acceleration, as indicated by the negativity of the green curve.
\section{Eternal bouncing universe}\label{bounce}
Next we consider two bouncing models. Both models are shown to be geodesically complete, and past and future eternal, satisfying the completeness requirements of \thref{ledthm}. The first is a transcendental function bounce and the second is a polynomial bounce. While we refer to these models as bouncing cosmologies, both models have $\ddot a>0$ and thus, in accordance with Defn.~\ref{LEDinflationdefn}, the models are \emph{also} inflationary models. We will comment further on this fact in our concluding remarks.

\subsection{Transcendental bounce}
The first bouncing model we consider is given by a transcendental function scale factor:
\begin{equation}\label{acosh}
    a(t)= a_0 \cosh \left( \frac{t}{\alpha} \right) \,.
\end{equation}
The above scale factor describes a universe that contracts, bounces and then expands (similar to the familiar case of closed de Sitter).~\footnote{A solution with this form was found in a non-local higher-derivative gravity in \cite{Biswas:2005qr,Biswas:2012bp}, which was argued to be free of pathologies, which if true would be quite remarkable given the level of energy condition violation (see \fig{eccosh}).} The Hubble parameter is
\begin{equation}
    H = \alpha^{-1} \tanh\left(\frac{t}{\alpha}\right)\,.
\end{equation}
Remarkably, the universe is inflating for the entire evolution since:
\begin{equation}\label{eq:acccosh}
  \frac{\ddot a}{a} =  \frac{1}{\alpha^{2}} \,,
\end{equation}
and hence, this bouncing model is \emph{also} an eternal inflationary model.
As with our preceding example \eq{aplusc}, it is possible to explicitly calculate the integrals of \thref{ledthm}. We find for the indefinite integrals:
\begin{equation}
\int^t \frac{a \left( \zeta \right) d\zeta}{\sqrt{\left( a \left(  \zeta \right) \right)^2 + 1}}=\alpha \operatorname{arctanh}\left(\frac{a_0 \sinh\left(\frac{t}{\alpha}\right)}{\sqrt{1 + a_0^2 + a_0^2 \sinh\left(\frac{t}{\alpha}\right)^2}}\right)
 \,,
\end{equation}
and
\begin{equation}
\int^t a \left( \zeta \right) d\zeta=a_0 \alpha \sinh\left(\frac{t}{\alpha}\right) \,. 
\end{equation}

The above integrals diverge over the full set of conditions discussed in Thm.~1 for all (non-zero) values of $\alpha$; hence, the spacetime with scale factor \eq{acosh} is geodesically complete. The cosmological evolution is depicted in \fig{cecosh}. Shown are the scale factor $a(t)$, Hubble parameter $H$ and co-moving Hubble radius $1/a H$. As is inherent in bouncing models the co-moving Hubble radius diverges at the bounce point. When the co-moving Hubble radius is decreasing the spacetime is inflating--in this case it is decreasing for all $t$.
\begin{figure}[H]
\centering
%\begin{minipage}[t]{.5\textwidth}
 % \centering
  \includegraphics[width=0.8\linewidth]{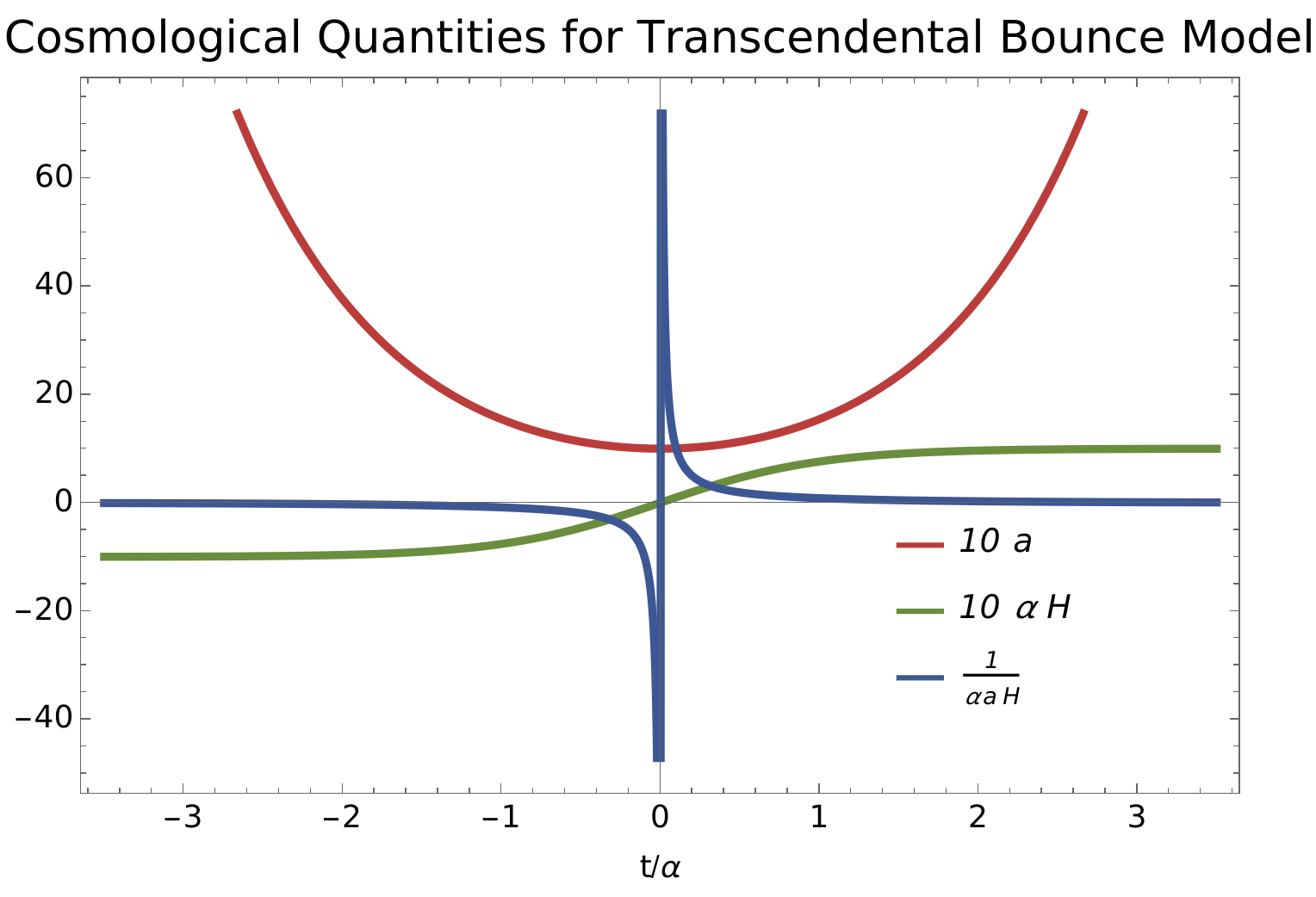}
  \captionof{figure}{Evolutionary behavior for \eq{acosh}. The scale factor amplified by an order of magnitude (red). The Hubble parameter amplified by an order of magnitude (green). The co-moving Hubble radius (blue). Parameters:~$a_0=\alpha=1$.}
  \label{cecosh}
%\end{minipage}%
%\begin{minipage}[t]{.5\textwidth}
%  \centering
%  \includegraphics[width=.9\linewidth]{mTilldefor111.jpg}
%  \captionof{figure}{$\tilde{m}_{4}^{2}$ for the case $b=d=\Lambda=1.$}
%  \label{mtilde2}
%\end{minipage}
\end{figure}
The Ricci scalar  and Kretchsmann scalar are given by:
\begin{equation}\label{cscosheq}
    R = \frac{6 \left(1 + \tanh\left(\frac{t}{\alpha}\right)^2\right)}{\alpha^2}\,, \qquad K = \frac{12 \left(1 + \tanh\left(\frac{t}{\alpha}\right)^4\right)}{\alpha^4} \,.
\end{equation}
We find there are no curvature singularities as inferred from the finite curvature invariants, plotted in \fig{cscosh}
\begin{figure}[H]
\centering
%\begin{minipage}[t]{.5\textwidth}
 % \centering
  \includegraphics[width=0.8\linewidth]{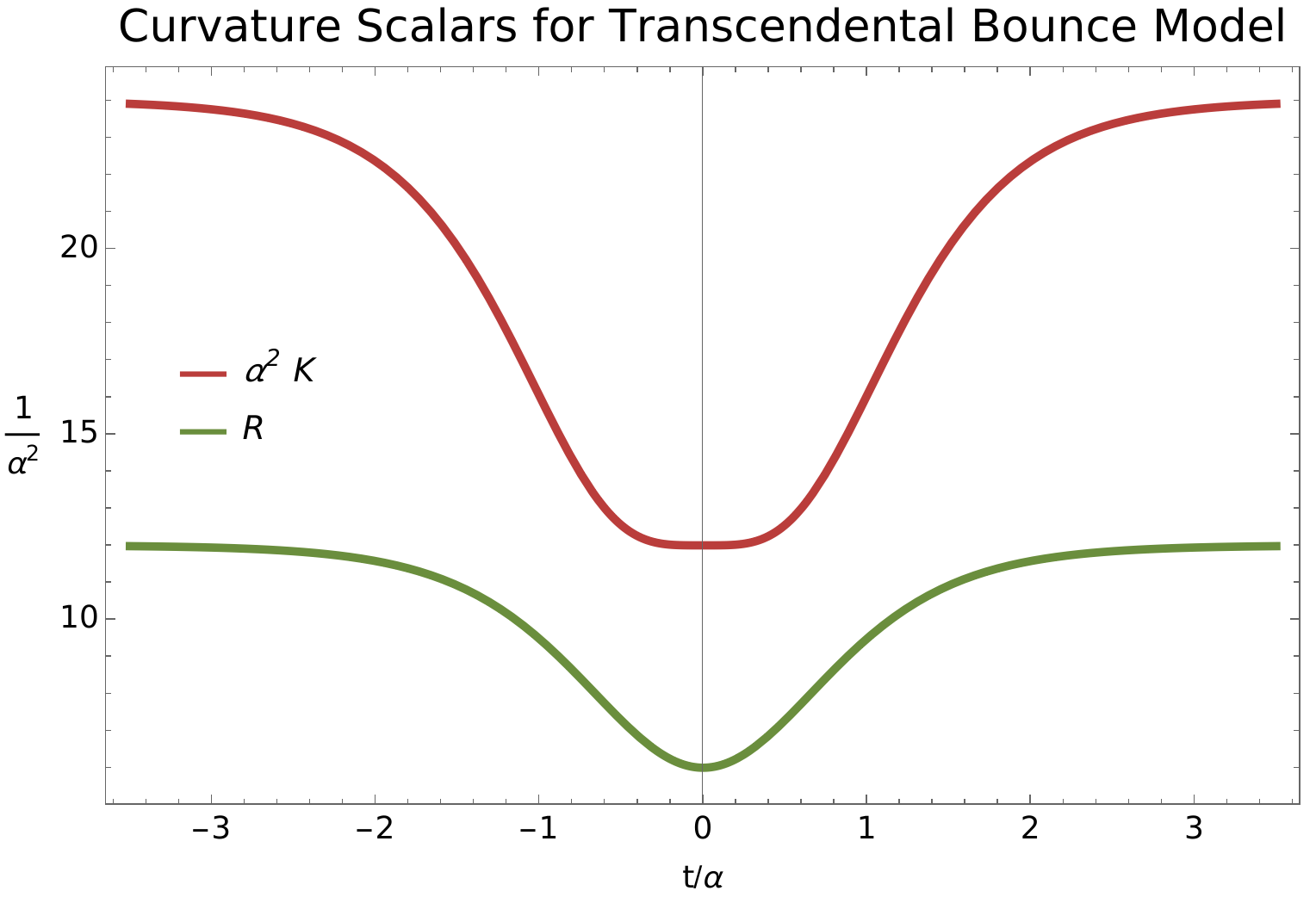}
  \captionof{figure}{Curvature invariants given by \eq{cscosheq}. The Ricci scalar (green) and the Kretchmann scalar (red). Model parameters: $a_0=\alpha=1$.}
  \label{cscosh}
%\end{minipage}%
%\begin{minipage}[t]{.5\textwidth}
%  \centering
%  \includegraphics[width=.9\linewidth]{mTilldefor111.jpg}
%  \captionof{figure}{$\tilde{m}_{4}^{2}$ for the case $b=d=\Lambda=1.$}
%  \label{mtilde2}
%\end{minipage}
\end{figure}
Calculation of the Einstein tensor yields non-vanishing components:
\begin{eqnarray}\label{eccosheq}
   G_{tt} &=& \frac{3 \tanh\left(\frac{t}{\alpha}\right)^2}{\alpha^2} \,, \nonumber \\
   G_{ii} &=& -\frac{a_0^2 (1 + 3 \cosh\left(\frac{2t}{\alpha}\right))}{2 \alpha^2} \,.
\end{eqnarray}
The energy density is given by $\rho = - G^t{}_t$ and the pressure is $p = G^i{}_i$. A plot depicting the energy conditions is given in \fig{eccosh}
\begin{figure}[H]
\centering
%\begin{minipage}[t]{.5\textwidth}
 % \centering
  \includegraphics[width=0.8\linewidth]{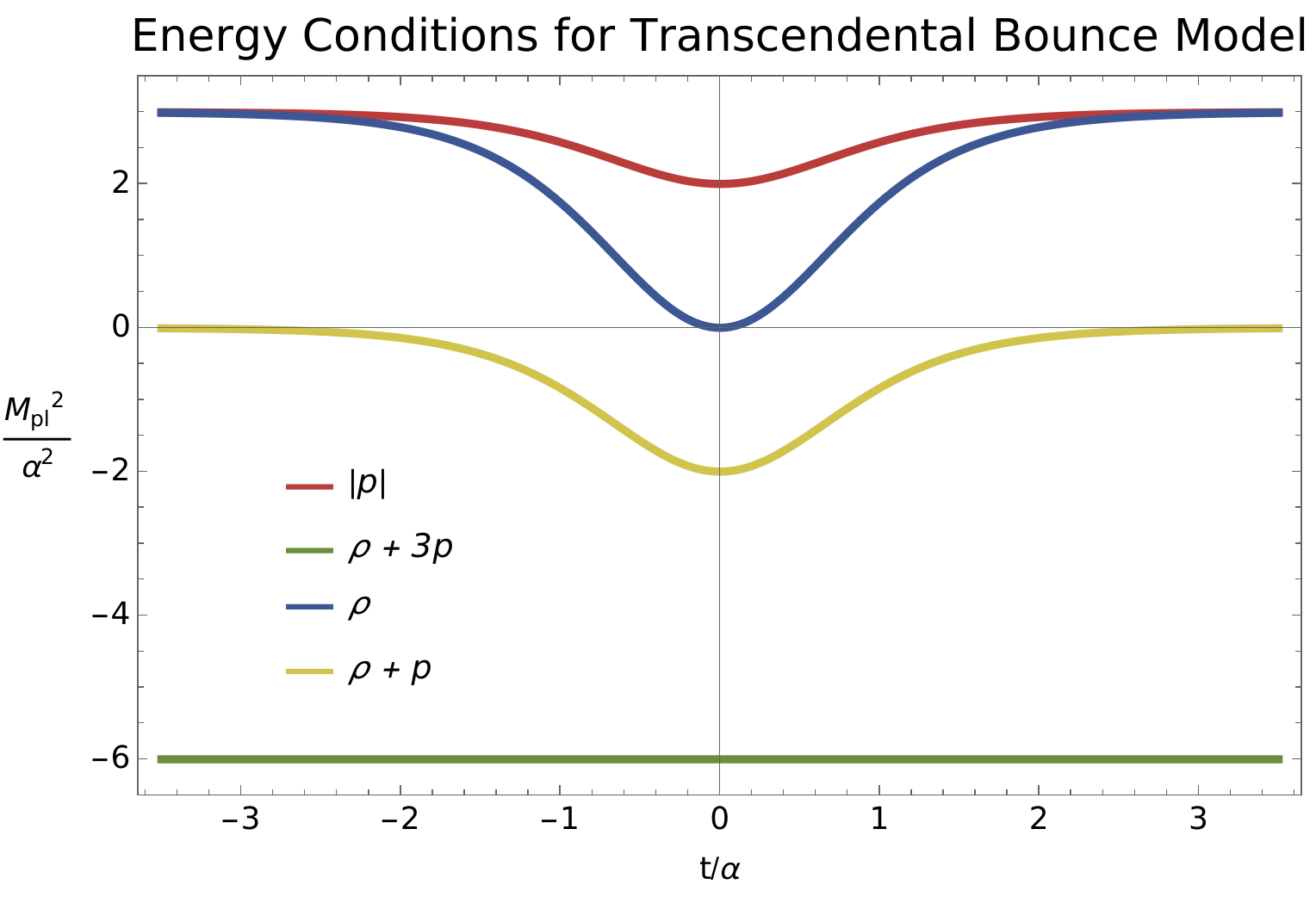}
  \captionof{figure}{Energy conditions from \eq{eccosheq}. Plot of energy density $\rho$ (blue), $\rho + p$  (yellow), $|p|$ (red) and $\rho + 3p$ (green). Model parameters: $a_0=\alpha=1$.}
  \label{eccosh}
%\end{minipage}%
%\begin{minipage}[t]{.5\textwidth}
%  \centering
%  \includegraphics[width=.9\linewidth]{mTilldefor111.jpg}
%  \captionof{figure}{$\tilde{m}_{4}^{2}$ for the case $b=d=\Lambda=1.$}
%  \label{mtilde2}
%\end{minipage}
\end{figure}
As with our previous example, all of the energy conditions including the ANEC, are violated. The SEC is violated for the entire range of the time coordinate as the model is eternally inflating, which can be seen from the eternal negativity of the green curve, and \eq{eq:acccosh}. 

\subsection{Polynomial bounce}\label{sec:poly}
An important class of nonsingular models is given by the general polynomial function:
\begin{equation}\label{gapoly}
    a(t) = a_n t^n + a_{n-1} t^{n-1} + \ldots + a_2 t^2 + a_1 t + a_0 \,.
\end{equation}
In order to avoid having $a$ become negative somewhere on the interval $t \in (-\infty,\infty)$, our polynomial must have even degree and $a_n > 0$.  With respect to Thm.~1, we derive the following with respect to a scale factor $p(t)$:

\begin{cor} \label{cor: polynomial completeness}
Let $p\left( t \right) \in \mathbb{R} \left[ t \right]$ be a polynomial over the field of real numbers over a single indeterminate $t$.  Additionally, let $p \left( t \right) >0$ and hence has no real roots.  The GFRW $\mathbb{R}^1_1 \times_p \Sigma$, with $\left( \Sigma , g_\Sigma \right)$ a complete Riemannian manifold, is complete. 
\end{cor}
\begin{proof}
  Apropos Thm. 1, any scale factor which has strictly positive infimum is geodesically complete.  A degree $n$ polynomial which is strictly positive satisfies this requirement.  First, if $p \left( t \right) > 0$ then it must be of even degree with $a_n > 0$.  Next, it is a well-known fact that an even degree polynomial with positive diverging range realizes its infimum, which by assumption is strictly positive.  Thus, a sufficient condition for geodesic completeness under Thm. 1 is satisfied, hence $\mathbb{R}^1_1 \times_p \Sigma$ is complete.
\end{proof}
\begin{defn} \label{defn: appropriate polynomials}
Let $p \left( t \right) \in \mathbb{R} \left[ t \right]$ be a polynomial which is complete under Cor.~\ref{cor: polynomial completeness} - namely an even degree polynomial with $a_n > 0$ and no real roots.  Any aforementioned polynomial is known as \emph{appropriate} for the remainder of this exposition.
\end{defn}

Because any even-degree polynomial scale factor $a(t) = a_n t^n + a_0$ with $n \ge 2$ and $a_n , a_0 > 0$ is inflationary and complete, all models of the aforementioned type will accelerate for at least some time. This suggests the following Corollary to Thm. \ref{ledthm}:

\begin{prop} \label{cor: complete are inflationary}
Every smooth, non-constant scale factor $a > 0$ of a geodesically complete GFRW spacetime must inflate for at least some period of time.
\end{prop}

Before the proof is stated, a lemma from analysis:
\begin{lemma} \label{lemma: no globally concave pos func exist}
Let $f : \mathbb{R} \rightarrow \mathbb{R}^+$ be twice differentiable with $f'' \leq 0$ over all $\mathbb{R}$.  Then $f$ must be a constant function: thus no globally concave positive non-constant real functions exist.
\end{lemma}
\begin{proof}
First, assume that $f$ is not a constant function.  There are now two cases to examine for $f'$: for some $t_0 \in \mathbb{R}$ either $f' \left( t_0 \right) > 0$ or $f' \left( t_0 \right) < 0$.

Let $f'\left( t_0 \right) > 0$.  By $f'' \leq 0$ one reaps $f'$ is non-increasing.  Starting from the Fundamental Theorem of Calculus one can calculate
\begin{equation}
    f \left( t \right) = f \left( t_0 \right) - \int^{t_0}_t f' \left( s \right) d s < f \left( t_0 \right) - f' \left( t_0 \right)\left(t_0 - t \right)
\end{equation}
In the limit
\begin{equation}
    \lim_{t \rightarrow - \infty} f\left( t \right) < \lim_{t\rightarrow - \infty} f \left( t_0 \right) - f' \left( t_0 \right)\left(t_0 -t \right) = - \infty
\end{equation}
thus $f$ cannot have range $\mathbb{R}^+$.

Likewise assume $f' \left( t_0 \right) < 0.$  In a similar vein
\begin{equation}
    f \left( t \right) = f \left( t_0 \right) + \int^{t}_{t_0} f' \left( s \right) d s < f \left( t_0 \right) + f' \left( t_0 \right)\left(t - t_0 \right)    \,
\end{equation}
again taking
\begin{equation}
    \lim_{t \rightarrow  \infty} f\left( t \right) < \lim_{t\rightarrow  \infty} f \left( t_0 \right) + f' \left( t_0 \right)\left(t - t_0 \right) = - \infty    
\end{equation}
Thus $f' = 0$ over all $\mathbb{R}$: hence $f$ is a constant function.
\end{proof}

Now the proof of Cor. \ref{cor: complete are inflationary} can be stated:
\begin{proof}
    Assume that $a'' \leq 0$ over all $\mathbb{R}$; by Lemma \ref{lemma: no globally concave pos func exist} this implies $a$ is a constant function: a contradiction.  Thus there is at least one point where $a'' > 0$; hence the spacetime is inflationary.
\end{proof}

Considering Corollary \ref{cor: complete are inflationary} and Defn.~\ref{LEDinflationdefn} we find 
% it is not that inflationary spacetimes are incomplete, but rather, 
non-trivial complete cosmologies \emph{must} be inflationary \cite{Easson:2024uxe}.\footnote{We pause here to acknowledge the celebrated work of Alexei Starobinsky who was motivated to find non-singular cosmological solutions by having an early-time accelerated expansion \cite{STAROBINSKY198099}, opening the window for the theory of the inflationary universe, and aligning with Cor.\ref{cor: complete are inflationary}.}

\subsubsection{Polynomial model}
Curvature and GR tensors can be computed in generality for a degree $n$ ``appropriate'' scale factor model per Defn. \ref{defn: appropriate polynomials}.  To evaluate the aforementioned class of models said polynomial constants $\left\{ a_\ell \right\}_{\ell = 0}^n$ must be evaluated.  The space of polynomial constants $\left\{ a_\ell \right\}^n_{\ell = 0}$ are not unconstrained–at minimum the requirements that the scale factor has no real roots yields a cursory set of limitations-namely that $a_n > 0$ and $n$ even, proemially.  It is well known that discriminants derived from radical methods can be utilized to reap polynomials with no real roots up to quartic degree, but more advanced methods including Galois theory, must be employed for quintic degree and higher - see \cite{Dummit2004} for a standard graduate reference.  The equations below are simple evaluations of well known curvature and GR tensors of GFRW models applied to an ``appropriate'' polynomial scale factor.  All equations will be left in ``summation form'', seeing as how there is no canonical way to expand said equations without further evaluation.

Recapitulating Eq. \ref{gapoly} of an ``appropriate'' polynomial
\begin{equation} \label{eq: poly scale factor general}
    \mathbb{R} \left[ t \right] \ni p \left( t \right) = \sum_{\ell = 0}^n a_\ell t^\ell
\end{equation}
in FRW\footnote{A propos Thm. \ref{ledthm}, the geodesic completeness of a GFRW is determined by its scale factor.  For purposes of simplicity, a flat ($k=0$) FRW is utilized for computation.} $\mathbb{R}^1_1 \times_p \mathbb{R}^{n-1}$ where $\mathbb{R} \left[ t \right]$ is the ring of polynomials over indeterminate $t$ taken to be the privileged time coordinate foliated in $\mathbb{R}^1_1$ and $a_\ell \in \mathbb{R}$ are the real coefficients of said polynomial.

Applying $p \left( t \right)$ to the formalism, one can calculate Kretchhmann and Ricci scalars to be

\begin{align}\label{eq: polynomial K}
    K =& \frac{12 \left(\left(\sum_{k=0}^n k a_k t^{k-1}\right)^4+\left(\sum_{k=0}^n (k-1) k a_k t^{k-2}\right)^2 \left(\sum_{k=0}^n a_k t^k\right)^2\right)}{\left(\sum_{k=0}^n a_k t^k\right)^4}\, , \\ 
    R =& \frac{6 \left(\left(\sum_{k=0}^n k a_k t^{k-1}\right)^2+\left(\sum_{k=0}^n (k-1) k a_k t^{k-2}\right) \sum_{k=0}^n a_k t^k\right)}{\left(\sum_{k=0}^n a_k t^k\right)^2} \, . 
\end{align}
The non-zero components of the Einstein tensor can be calculated
\begin{align}
    G_{tt} =& \frac{3 \left(\sum_{k=0}^n k a_k t^{k-1}\right)^2}{\left(\sum_{k=0}^n a_k t^k\right)^2} \, , \\
    G_{ii} =& -\left(\sum_{k=0}^n k a_k t^{k-1}\right)^2-2 \left(\sum_{k=0}^n (k-1) k a_k t^{k-2}\right) \sum_{k=0}^n a_k t^k \, . \label{eq: polynomial G}
\end{align}
The values of $\rho = - G^t_t$, $p = G^i_i$, and all subsequent energy conditions are trivially read from the previous equation.  No plots can be generated until $\left\{ a_\ell \right\}^n_{\ell = 0}$ are specified.

For the sake of completeness, the Hubble parameter is given
\begin{equation} \label{eq: polynomial H}
    H = \frac{\sum_{k=0}^n k a_k t^{k-1}}{\sum_{k=0}^n a_k t^k}\, .
\end{equation}

\subsubsection{Quadratic model} \label{sec: quadratic model}
We now study a particular subclass of the above polynomial models. To avoid having zeros of $a$ and the associated singularities, we require that our polynomial must have no real roots. In this example, the linear term is included to skew the symmetry of the pure quadratic for purposes of generality:
\begin{equation}\label{apoly}
    a(t) = a_0 t^2 + bt + c \,,
\end{equation}
and for numerics we will take $a_0 = b= 1$, $c= 2$, adhering to our requirement that \eq{apoly} have no real roots.

The Hubble parameter is
\begin{equation}
    H = \frac{b+2 a_0 t}{c + bt + a_0 t^2}\,.
\end{equation}

As with our previous bouncing example, the universe is inflating for the entire evolution $t \in (-\infty, \infty)$:
\begin{equation}
  \frac{\ddot a}{a} = \frac{2 a_0}{c+ bt  + a_0 t^2}\,,
\end{equation}
and $\ddot a = 2 a_0 \, \forall \, t$. 
For the scale factor \eq{apoly}, the first integral of Thm.~1 may be evaluated as an elliptic function. Further we find for the indefinite integral:
%\begin{equation}
%\int\frac{a \left( t \right) dt}{\sqrt{\left( a \left( t \right) \right)^2 + 1}}=\alpha \operatorname{arctanh}\left(\frac{a_0 \sinh\left(\frac{t}%{\alpha}\right)}{\sqrt{1 + a_0^2 + a_0^2 \sinh\left(\frac{t}{\alpha}\right)^2}}\right)
 %\,,
%\end{equation}
%and(3 (\!\(

\begin{equation}
\int^t a \left( \zeta \right) d\zeta= c t  + b t^2 + \frac{a_0 t^3}{3} \,. 
\end{equation}

For the parameter values we consider in this example ($a_0 = b= 1$, $c= 2$), the integrals diverge over the full set of conditions discussed in Thm~1; hence, the spacetime is geodesically complete. The cosmological evolution is depicted in \fig{cepoly}. Shown are the scale factor $a(t)$, Hubble parameter $H$ and co-moving Hubble radius $1/a H$. As is typical in bouncing models the co-moving Hubble radius diverges at the bounce point. When the co-moving Hubble radius is decreasing the spacetime is inflating--once again we see that in this case it is decreasing for all $t$.
\begin{figure}[H]
\centering
%\begin{minipage}[t]{.5\textwidth}
 % \centering
  \includegraphics[width=0.8\linewidth]{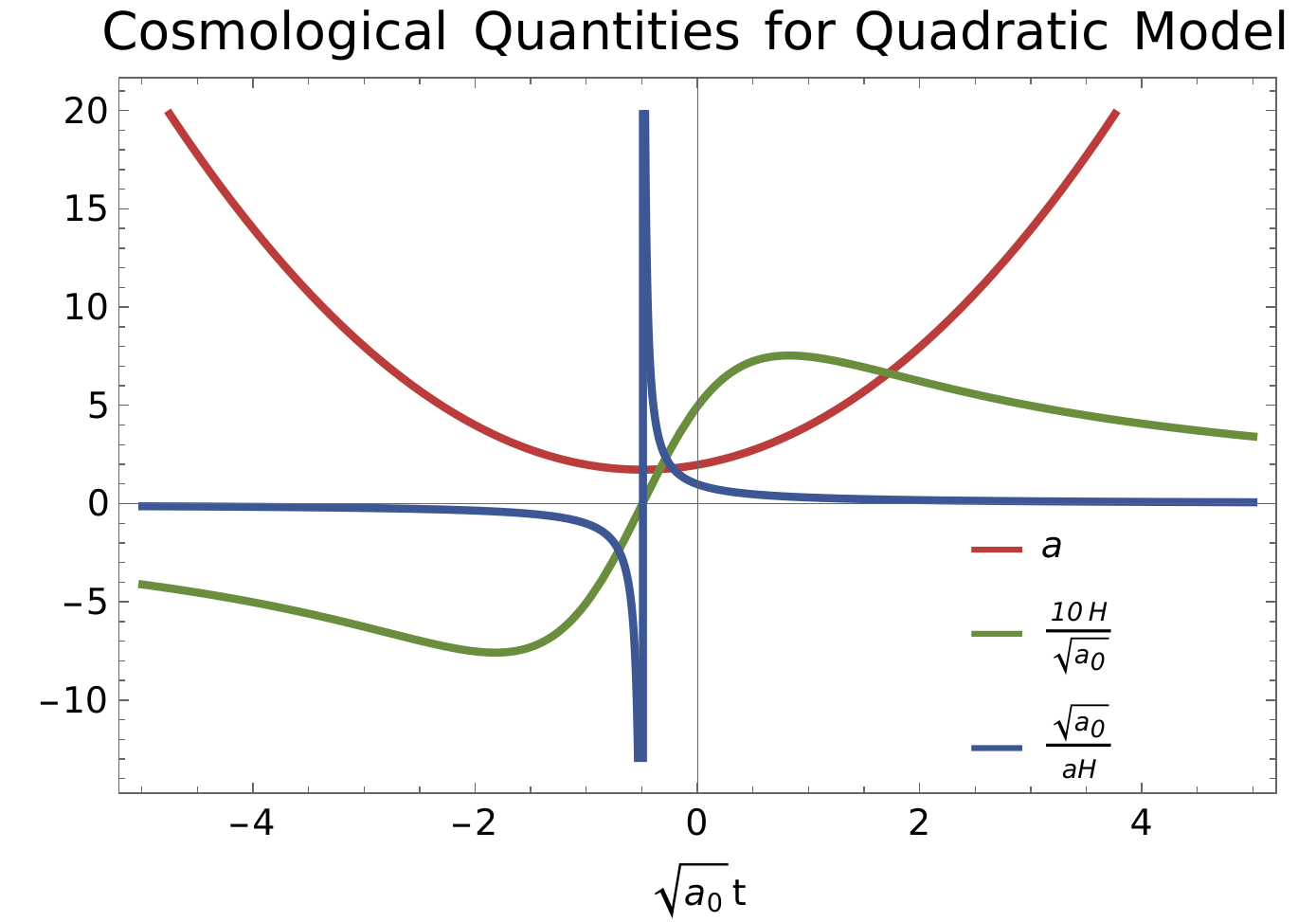}
  \captionof{figure}{Evolutionary behavior for \eq{apoly}. The scale factor (red). The Hubble parameter amplified by an order of magnitude (green). The co-moving Hubble radius (blue). Model parameters:~$a_0=b=1$,~$c=2$.}
  \label{cepoly}
%\end{minipage}%
%\begin{minipage}[t]{.5\textwidth}
%  \centering
%  \includegraphics[width=.9\linewidth]{mTilldefor111.jpg}
%  \captionof{figure}{$\tilde{m}_{4}^{2}$ for the case $b=d=\Lambda=1.$}
%  \label{mtilde2}
%\end{minipage}
\end{figure}
The Ricci scalar  and Kretchsmann scalar are given by:
\begin{equation}\label{cspolyeq}
    R = \frac{6 \left(b^2 + 6 a_0 b t + 2 a_0 (c + 3 a_0 t^2)\right)}{(c + t (b + a_0 t))^2}\,, \qquad K =\frac{12 \left((b + 2 a_0 t)^4 + 4 a_0^2 (c + t (b + a_0 t))^2\right)}{(c + t (b + a_0 t))^4}
 \,.
\end{equation}
The absence of curvature singularities is shown from the curvature invariants plotted in \fig{cspoly}.  Additionally, the geodesics can be computed in closed form, and the existence of said solution for all time can be shown - see Appendix \ref{appendix: quadratic geodesics}.
\begin{figure}[H]
\centering
%\begin{minipage}[t]{.5\textwidth}
 % \centering
  \includegraphics[width=0.8\linewidth]{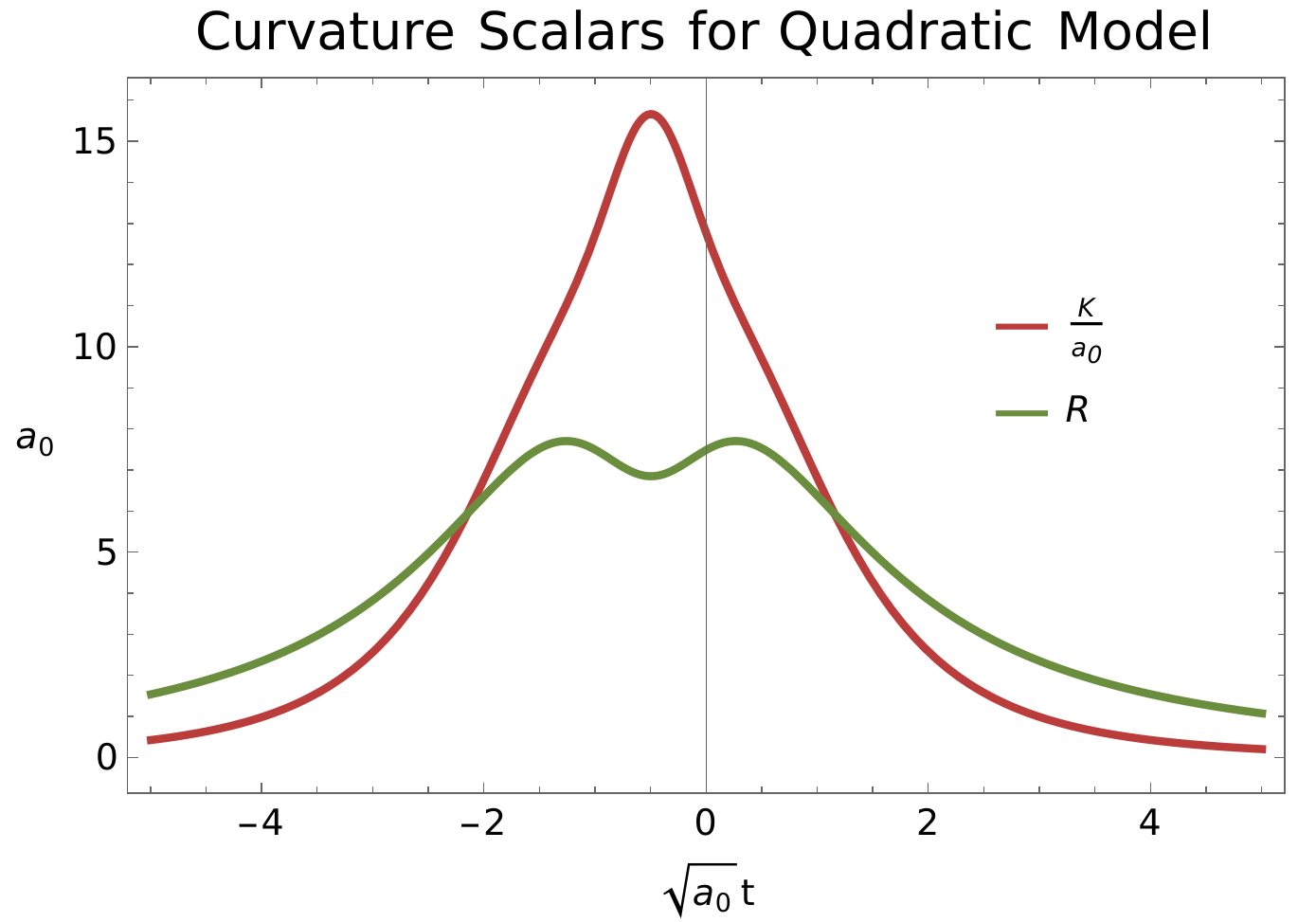}
  \captionof{figure}{Curvature invariants computed in \eq{cspolyeq}. The Ricci scalar (green) and the Kretchmann scalar (red). Model parameters: $a_0=b=1$, $c=2$.}
  \label{cspoly}
%\end{minipage}%
%\begin{minipage}[t]{.5\textwidth}
%  \centering
%  \includegraphics[width=.9\linewidth]{mTilldefor111.jpg}
%  \captionof{figure}{$\tilde{m}_{4}^{2}$ for the case $b=d=\Lambda=1.$}
%  \label{mtilde2}
%\end{minipage}
\end{figure}
Calculation of the Einstein tensor yields non-vanishing components:
\begin{eqnarray}\label{ecpolyeq}
   G_{tt} &=& \frac{3 (b + 2 a_0 t)^2}{(c + t (b + a_0 t))^2} \,, \nonumber \\
   G_{ii} &=& -b^2 - 8 a_0 b t - 4 a_0 \left(c + 2 a_0 t^2\right) \,.
\end{eqnarray}
The energy density is given by $\rho = - G^t{}_t$ and the pressure is $p = G^i{}_i$. A plot elucidating the energy conditions is given in \fig{ecpoly}
\begin{figure}[H]
\centering
%\begin{minipage}[t]{.5\textwidth}
 % \centering
  \includegraphics[width=0.8\linewidth]{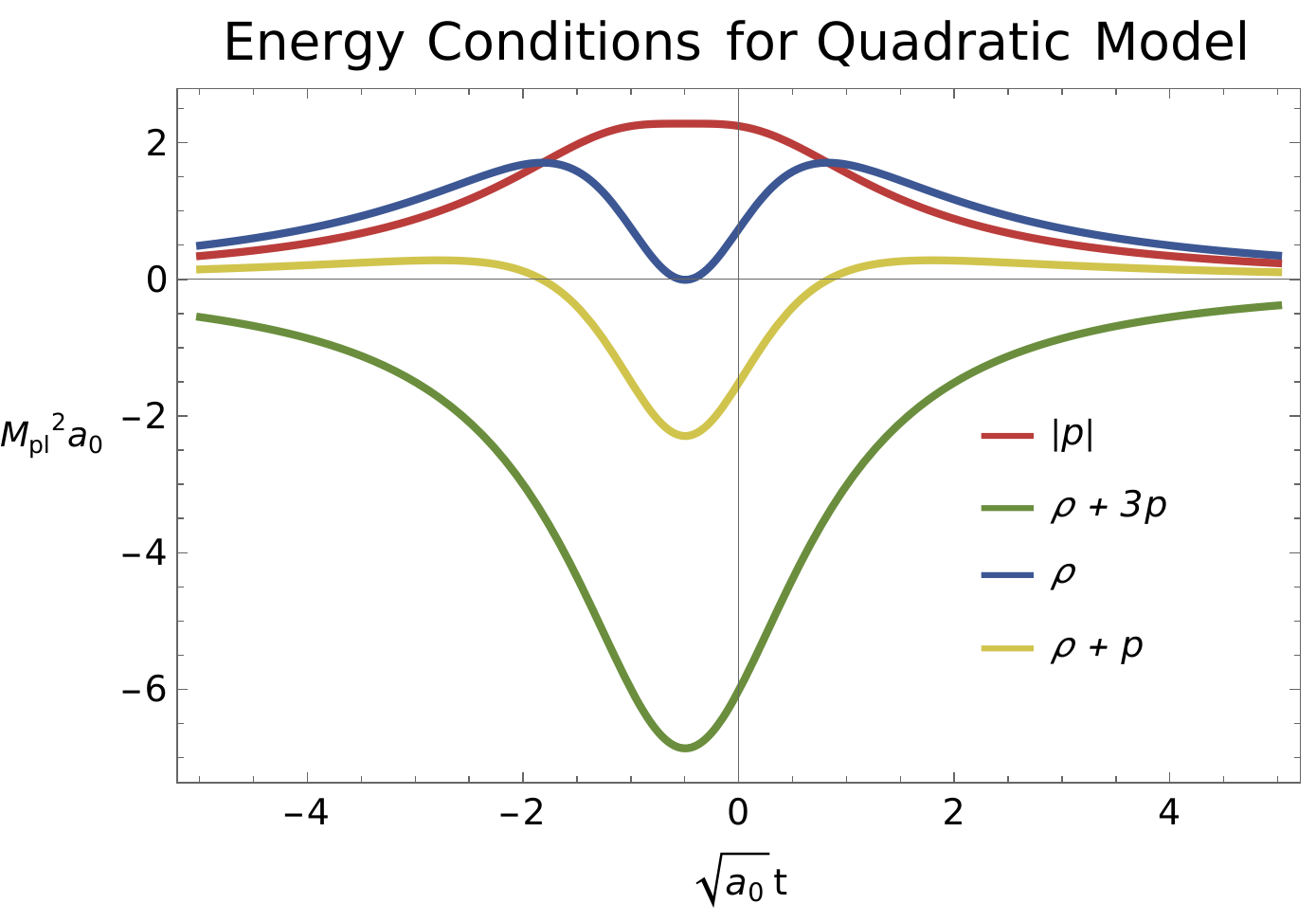}
  \captionof{figure}{Energy conditions from \eq{ecpolyeq}. Plot of energy density $\rho$ (blue), $\rho + p$  (yellow), $|p|$ (red) and $\rho + 3p$ (green).  Model parameters: $a_0=b=1$, $c=2$.}
  \label{ecpoly}
%\end{minipage}%
%\begin{minipage}[t]{.5\textwidth}
%  \centering
%  \includegraphics[width=.9\linewidth]{mTilldefor111.jpg}
%  \captionof{figure}{$\tilde{m}_{4}^{2}$ for the case $b=d=\Lambda=1.$}
%  \label{mtilde2}
%\end{minipage}
\end{figure}
While all of the classic energy conditions are violated, violation of the NEC can be restricted to an arbitrarily short time interval near the bounce point. Because of this, the asymptotic regions may be broken up into those obeying the NEC and the region violating the NEC.

In this model, and others we consider, the NEC is violated during a finite time interval, $-t_- < t< t_+$. One may break up the ANEC integral into the asymptotic regions $t<-t_-$, $t> t_+$ where other energy conditions ( \eg~NEC, WEC and DEC) are satisfied, and the NEC violating region. In this case the ANEC condition \eq{anec} may be recast as \cite{Giovannini:2017ndw}:
\begin{equation}\label{anecarea}
    \left| \int_{t_-}^{t_+} (\rho + p) d\lambda \right| \leq \int_{-\infty}^{-t_-} (\rho + p) d\lambda + \int_{t_+}^{+\infty} (\rho + p) d\lambda \,,
\end{equation}
where in the above we assume the spacetime curvature invariants are all regular and that the non-spacelike (causal) geodesics are complete and extendible to arbitrary values of their affine parameter. This is the case in all eternal universe solutions we present in this paper. In the energy condition plot, any violation of the ANEC may be ascertained via the inequality \eq{anecarea}, each term being conveniently associated to an area beneath the curves defined by $\rho + p$ in the NEC violating, and NEC satisfying regions.

In this case the ANEC integral gives:
% \begin{eqnarray}
% %\begin{equation}
%      \int_{- \infty}^\infty (\rho + p ) d\lambda &=& \\
% &&\left(\frac{2 a_0 \left(2 a_0 t+b\right)}{\left(b^2-4 a_0 c\right) \left(a_0 t^2+b t+c\right)} - \frac{8 a_0^2 \tan ^{-1}\left(\frac{2 a_0 t+b}{\sqrt{4 a_0 c-b^2}}\right)}{\left(4 a_0 c-b^2\right){}^{3/2}}-\frac{2 a_0 t+b}{\left(a_0 t^2+b t+c\right){}^2} \right) \Big|_{-\infty}^{\infty}
% \leq 0 \,,
% %\end{equation}
% \end{eqnarray}
\begin{align}
\int_{- \infty}^\infty (\rho + p ) \, d\lambda 
&= \left(
\frac{2 a_0 (2 a_0 t + b)}{(b^2 - 4 a_0 c)(a_0 t^2 + b t + c)}
- \frac{8 a_0^2 \tan^{-1}\left(\frac{2 a_0 t + b}{\sqrt{4 a_0 c - b^2}}\right)}{(4 a_0 c - b^2)^{3/2}} \right. \nonumber \\
&\quad \left. - \frac{2 a_0 t + b}{(a_0 t^2 + b t + c)^2}
\right) \Big|_{-\infty}^{\infty} \leq 0 \,,
\end{align}
which is negative for generic parameter values indication the ANEC is violated.

% \begin{equation}
%      \int_{- \infty}^\infty (\rho + p ) d\lambda= -\frac{2(b + 2a_0t)}{c + t(b + a_0t)} \Big|_{-\infty}^{\infty}
% = 0 \,.
% \end{equation}

In the late universe, only the SEC is violated allowing continued acceleration, as indicated by the negativity of the green curve. This too is an eternally inflating, bouncing cosmology similar to the behavior of \eq{acosh}, but it is somewhat better-behaved with respect to energy conditions as the NEC violation may be restricted to a moment in the very early universe.~\footnote{A modern example of such a bounce with momentary NEC violation is the S-brane model~\cite{Brandenberger:2013zea}.}

\subsubsection{Models from limiting curvature}\label{lchbounce}
One method for building such non-singular bounces stems from the limiting curvature hypothesis (LCH) which stipulates the existence of minimal length (or maximal curvature) scale in nature, presumably somewhere near the Planck length \cite{markov1982, markov1987}, although one may introduce new physics at lower energy scales using this method. By construction, curvature invariants are limited in an action resulting in non-singular solutions. The LCH has been used to successfully construct many non-singular cosmological solutions \cite{Mukhanov1992,Brandenberger1993,Moessner1995,Brandenberger1995,Brandenberger1998,Easson1999, Easson2003ia,Easson:2006jd,Chamseddine:2016uef} in a variety of settings, as well as non-singular black hole solutions \cite{Frolov1989, Frolov1990, Morgan1991,Trodden1993,Easson2003,Easson:2017pfe,Chamseddine2017ktu,Brandenberger:2021jqs, Frolov:2021afd}.

A particular example following from the LCH is given by the scale factor \cite{Chamseddine:2016uef}:
\begin{equation}
a(t) = a_0 \left(c + b t^2 \right)^{1/3} \,,
\end{equation}
for constants $c>0, b>0$, which is capable of smoothly connecting two matter dominated phases having $a \propto t^{2/3}$ via a NEC violating bounce at $t=0$ where $\ddot a>0$. The schematic behavior of the cosmological quantities and the energy conditions are very similar to the polynomial model analyzed above; however, in the present case, even the SEC is obeyed in both the early and late (non-accelerating) universe. The ANEC integral \eqs{anec}{anecarea} may be calculated in terms of elliptic functions:
\begin{align}
\int_{- \infty}^\infty (\rho + p ) \, d\lambda 
&= 
\frac{2 a_0 (2 a_0 t + b)}{(b^2 - 4 a_0 c)(a_0 t^2 + b t + c)}
- \frac{8 a_0^2 \tan^{-1}\left(\frac{2 a_0 t + b}{\sqrt{4 a_0 c - b^2}}\right)}{(4 a_0 c - b^2)^{3/2}} \nonumber \\
&\quad 
- \frac{2 a_0 t + b}{(a_0 t^2 + b t + c)^2}
\Big|_{-\infty}^{\infty} \leq 0 \,,
\end{align}
%
% \begin{equation}
%  \int_{- \infty}^\infty (\rho + p ) d\lambda=
% \frac{2 a_0 \left(2 a_0 t+b\right)}{\left(b^2-4 a_0 c\right) \left(a_0 t^2+b t+c\right)}-\frac{8 a_0^2 \tan ^{-1}\left(\frac{2 a_0 t+b}{\sqrt{4 a_0 c-b^2}}\right)}{\left(4 a_0 c-b^2\right){}^{3/2}}-\frac{2 a_0 t+b}{\left(a_0 t^2+b t+c\right){}^2}  \, \Big|_{-\infty}^{\infty}
% \leq 0\,,
% \end{equation}
% The ANEC \eqs{anec}{anecarea} is satisfied:
% \begin{equation}
%      \int_{- \infty}^\infty (\rho + p ) d\lambda=
% -\frac{4bt}{3(c + bt^2)}\Big|_{-\infty}^{\infty}= 0 \,.
% \end{equation}
which is negative for generic parameter values.

\section{An eternal loitering universe}\label{loiter}
We now discuss the possibility of an eternal loitering universe \cite{Sahni:1991ks,Alexander:2000xv,Brandenberger:2001kj,Melcher:2023kpd}. We break this section into two model categories, those which actually have an extended past Minkowski phase with $H=0$, and those which are always expanding, but have an extremely slow initial expansion asymptotically approaching Minkowski in the past. The latter we refer to as ``quasi-loitering" models.

\subsection{Loitering eternal universe}\label{subsecloit}
It is possible to construct a loitering eternal universe by smoothly gluing together a long Minkowski past phase to an inflating phase via a smooth bump function. This blueprint was first utilized in our previous work \cite{Lesnefsky:2022fen} to build a loitering model which is complete. 

For this example we take the scale factor to be constant for $t<0$, so that $H=0 \, \forall t<0$, and the spacetime is Minkowski. For $t>1$, the scale factor is exponential, $H$ is constant and the spacetime is quasi-de Sitter. A smooth bump function is used to piece the two regimes together in the regime $t \in [0,1]$.
\begin{figure}[H]
\centering
%\begin{minipage}[t]{.5\textwidth}
 % \centering
  \includegraphics[width=1\linewidth]{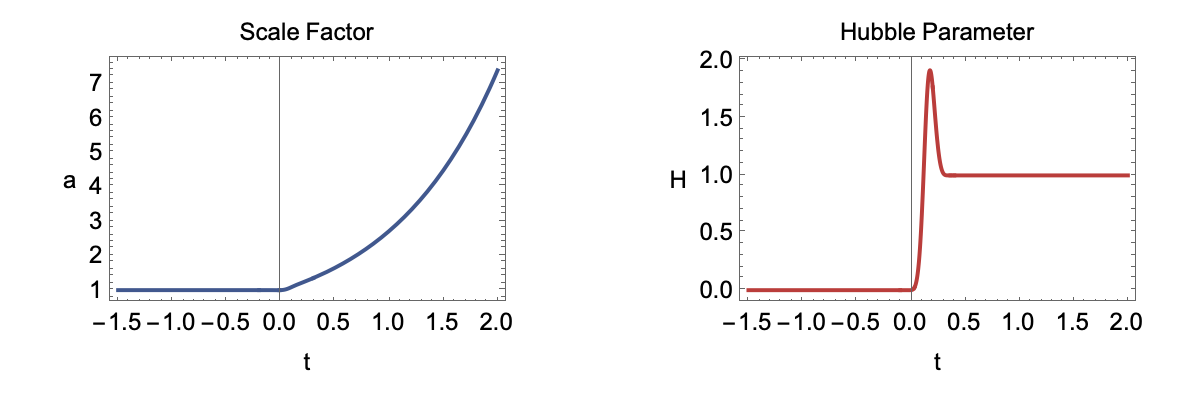}
  \captionof{figure}{Evolutionary behavior for the loitering eternal model of \S \ref{subsecloit}.}
  \label{loitf}
%\end{minipage}%
%\begin{minipage}[t]{.5\textwidth}
%  \centering
%  \includegraphics[width=.9\linewidth]{mTilldefor111.jpg}
%  \captionof{figure}{$\tilde{m}_{4}^{2}$ for the case $b=d=\Lambda=1.$}
%  \label{mtilde2}
%\end{minipage}
\end{figure}

 Clearly the model is eternal and complete. As with previous examples, the NEC is violated when $\dot H >0$, and inflationary by Defn. \ref{LEDinflationdefn} with $\ddot a>0$ for $t>0$. The Minkowski period with $H=0$ satisfies the bouncing conditions with $\dot a=0$ and with non-trivial scale factor overall, of Defn. \ref{ELbouncedefn}. The model has an accelerating de Sitter phase at late times with $H=const$. The ``phantom-like'' physics required to push the universe from the static 
$H = 0$ plateau into the de~Sitter phase imprints an overall negative ANEC integral.

Using this technique it is simple to build an infinite number of interesting eternal and complete cosmological spacetimes. A category of geodesically complete spacetimes which we do not focus on here are cyclic models which are defined by multiple repeating bounces (see \eg~\cite{Penrose:2014vok,Steinhardt:2001st,Steinhardt:2002ih, Khoury:2003rt,Ijjas:2019pyf,Lehners:2013cka}). Such models are easy to construct using the above composite method. 
We dedicate further study of such composite spacetimes to future work \cite{ELfrankenstein}. 

\subsection{Quasi-loitering eternal universe}
For our final example of an eternal universe we construct what we refer to as a quasi-loitering universe–in that the scale factor is always expanding, although initially at an exponentially small (yet accelerated) rate before entering an increased accelerated inflationary phase. The situation at early time is reminiscent of the so-called Galileon Genesis model \cite{Creminelli:2010ba}, in which a universe that is asymptotically Minkowski in the past expands with $\dot H > H^2$, violating the NEC with an increasing energy density $\rho$; although, in the current model this phase naturally comes to an end.  Eventually the inflationary phase transitions, the energy density decreases and the universe enters another period of extremely slow (negative accelerated) non-inflationary expansion. The quasi-loitering scale factor is:
\begin{equation}\label{atanh}
    a(t)= a_0 \tanh\left(\frac{t}{\alpha}\right) + c
 \,.
\end{equation}
In order to achieve geodesic completeness we require sufficient lifting of the tanh function, requiring $c>a_0$. The Hubble parameter is:
\begin{equation}
    H = \frac{a_0 \operatorname{sech}^2\left(\frac{t}{\alpha}\right)}{\alpha (c + a_0 \tanh\left(\frac{t}{\alpha}\right))}\,.
\end{equation}

At early times asymptotic into the past, as $t \rightarrow - \infty$, the universe is very slowly inflating.  This inflation begins to increase
with a sudden growing energy density (the continuity equation in GR gives $\dot \rho = - 3 H (\rho +p)$) and a strong NEC violating phase with $\dot H >0$, $\rho + p<0$. The acceleration equation gives:
\begin{equation}
\frac{\ddot a}{a} = -\frac{2 a_0 \sech^2\left(\frac{t}{\alpha}\right) \tanh\left(\frac{t}{\alpha}\right)}{\alpha^2 (c + a_0 \tanh\left(\frac{t}{\alpha}\right))} \,.
\end{equation}

As $t \rightarrow \pm \infty$, the Hubble parameter approaches $H  \rightarrow 0$. These zeros of $H$ are examples of bounce points at infinity as identified by Defn.~\ref{ELbouncedefn}, and hence, this is both an inflationary and a bouncing spacetime.

For the scale factor \eq{atanh}, it is possible to explicitly calculate the integrals of Thm~1. We find for the indefinite integrals:
\begin{equation}
\begin{aligned}
& \int^t \frac{a(\zeta) \, d\zeta}{\sqrt{\left(a(\zeta)\right)^2 + 1}} \\
&= \frac{1}{2\alpha} \left(\frac{(a_0 - c) \operatorname{arctanh}\left(\frac{1 - a_0 c + c^2 + a_0 (-a_0 + c) \tanh\left(\frac{t}{\alpha}\right)}{\sqrt{1 + (a_0 - c)^2} \sqrt{1 + c^2 + a_0 \tanh\left(\frac{t}{\alpha}\right)(2c + a_0 \tanh\left(\frac{t}{\alpha}\right))}}\right)}{\sqrt{1 + (a_0 - c)^2}} \right. \\
&\left. \quad + \frac{(a_0 + c) \operatorname{arctanh}\left(\frac{1 + c (a_0 + c) + a_0 (a_0 + c) \tanh\left(\frac{t}{\alpha}\right)}{\sqrt{1 + (a_0 + c)^2} \sqrt{1 + c^2 + a_0 \tanh\left(\frac{t}{\alpha}\right)(2c + a_0 \tanh\left(\frac{t}{\alpha}\right))}}\right)}{\sqrt{1 + (a_0 + c)^2}} \right)
\end{aligned}
\end{equation}
and
\begin{equation}
\int^t a \left( \zeta \right) d\zeta=c t + a_0 \alpha \log\left(\cosh\left(\frac{t}{\alpha}\right)\right) \,. 
\end{equation}

As with our previous examples, the above integrals diverge over the full set of conditions discussed in Thm.~1 for all (non-zero) values of $\alpha$; hence, the spacetime is geodesically complete. The cosmological evolution is depicted in \fig{cetanh}. Shown are the scale factor $a(t)$, Hubble parameter $H$ and co-moving Hubble radius $1/a H$. 
% When the co-moving Hubble radius is decreasing the spacetime is inflating. 
The inflation occurs at early times and eventually ends leading to a decelerating expanding phase with $\ddot a/a <0$. The future is a very slowly expanding (quasi-future-loitering) phase during which the energy conditions become well-behaved except the DEC is violated at late times and asymptotically approaches saturation, see Fig.~\ref{ectanh}. 

\begin{figure}[H]
\centering
%\begin{minipage}[t]{.5\textwidth}
 % \centering
  \includegraphics[width=0.8\linewidth]{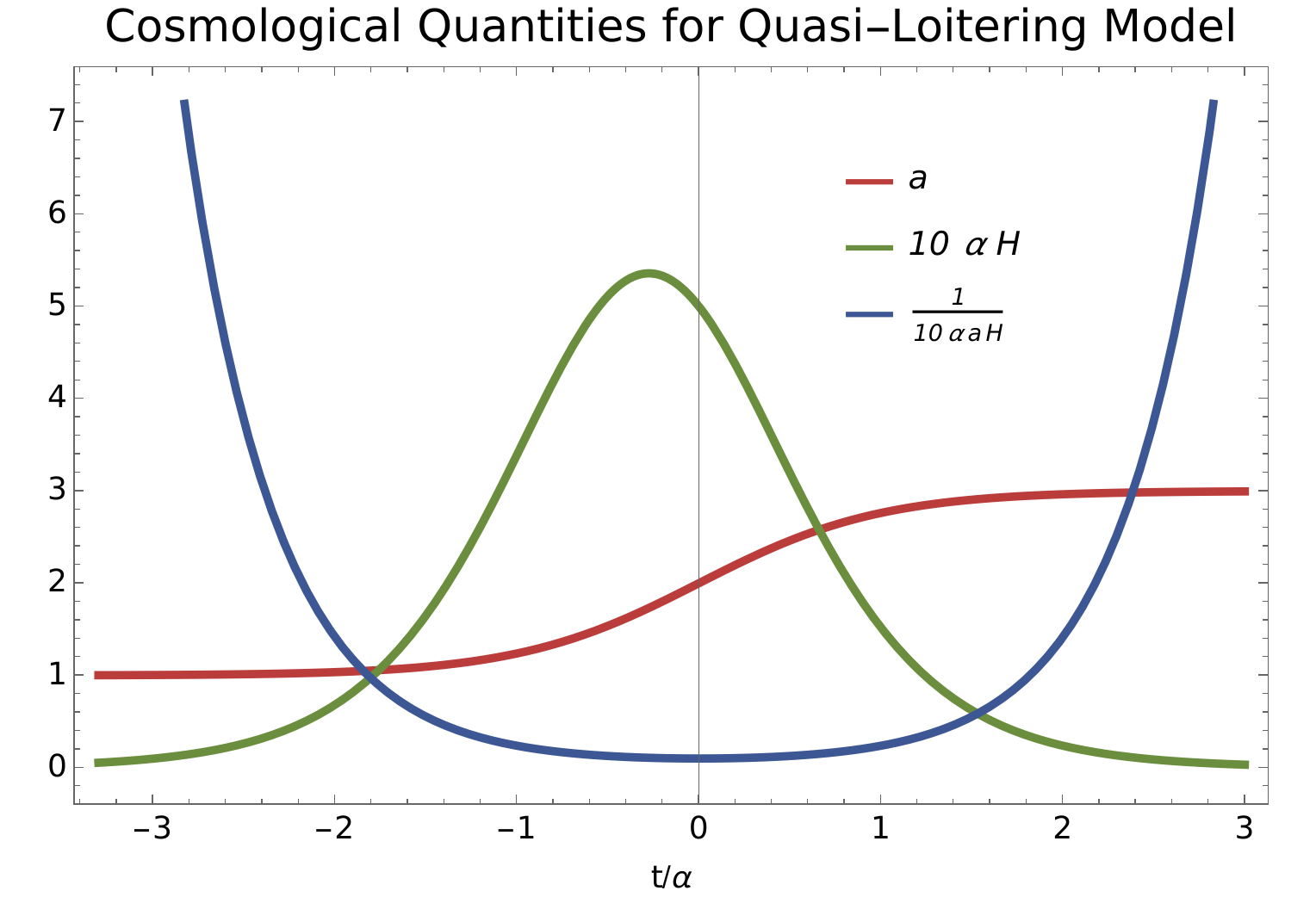}
  \captionof{figure}{Evolutionary behavior for \eq{atanh}. The scale factor (red). The Hubble parameter (green) magnified by an order of magnitude. The co-moving Hubble radius suppressed by an order of magnitude (blue). Model parameters: $a_0=\alpha=1$,~$c=2$.}
  \label{cetanh}
%\end{minipage}%
%\begin{minipage}[t]{.5\textwidth}
%  \centering
%  \includegraphics[width=.9\linewidth]{mTilldefor111.jpg}
%  \captionof{figure}{$\tilde{m}_{4}^{2}$ for the case $b=d=\Lambda=1.$}
%  \label{mtilde2}
%\end{minipage}
\end{figure}

The Ricci scalar  and Kretchsmann scalar are given by:
\begin{eqnarray}
    R &=&-\frac{6 a_0 \operatorname{sech}\left(\frac{t}{\alpha}\right)^4 \left(a_0 (-2 + \cosh\left(\frac{2t}{\alpha}\right)) + c \sinh\left(\frac{2t}{\alpha}\right)\right)}{\alpha^2 (c + a_0 \tanh\left(\frac{t}{\alpha}\right))^2} \nonumber \,, \\
    K &=& \frac{6 a_0^2 \operatorname{sech}\left(\frac{t}{\alpha}\right)^8 \left(5 a_0^2 - c^2 - 4 a_0^2 \cosh\left(\frac{2t}{\alpha}\right) + (a_0^2 + c^2) \cosh\left(\frac{4t}{\alpha}\right) + 16 a_0 c \cosh\left(\frac{t}{\alpha}\right) \sinh\left(\frac{t}{\alpha}\right)^3\right)}{\alpha^4 (c + a_0 \tanh\left(\frac{t}{\alpha}\right))^4} \nonumber \,.
\end{eqnarray}

The absence of curvature singularities is shown from the curvature invariants plotted in 
\fig{cstanh}

\begin{figure}[H]
\centering
%\begin{minipage}[t]{.5\textwidth}
 % \centering
  \includegraphics[width=0.8\linewidth]{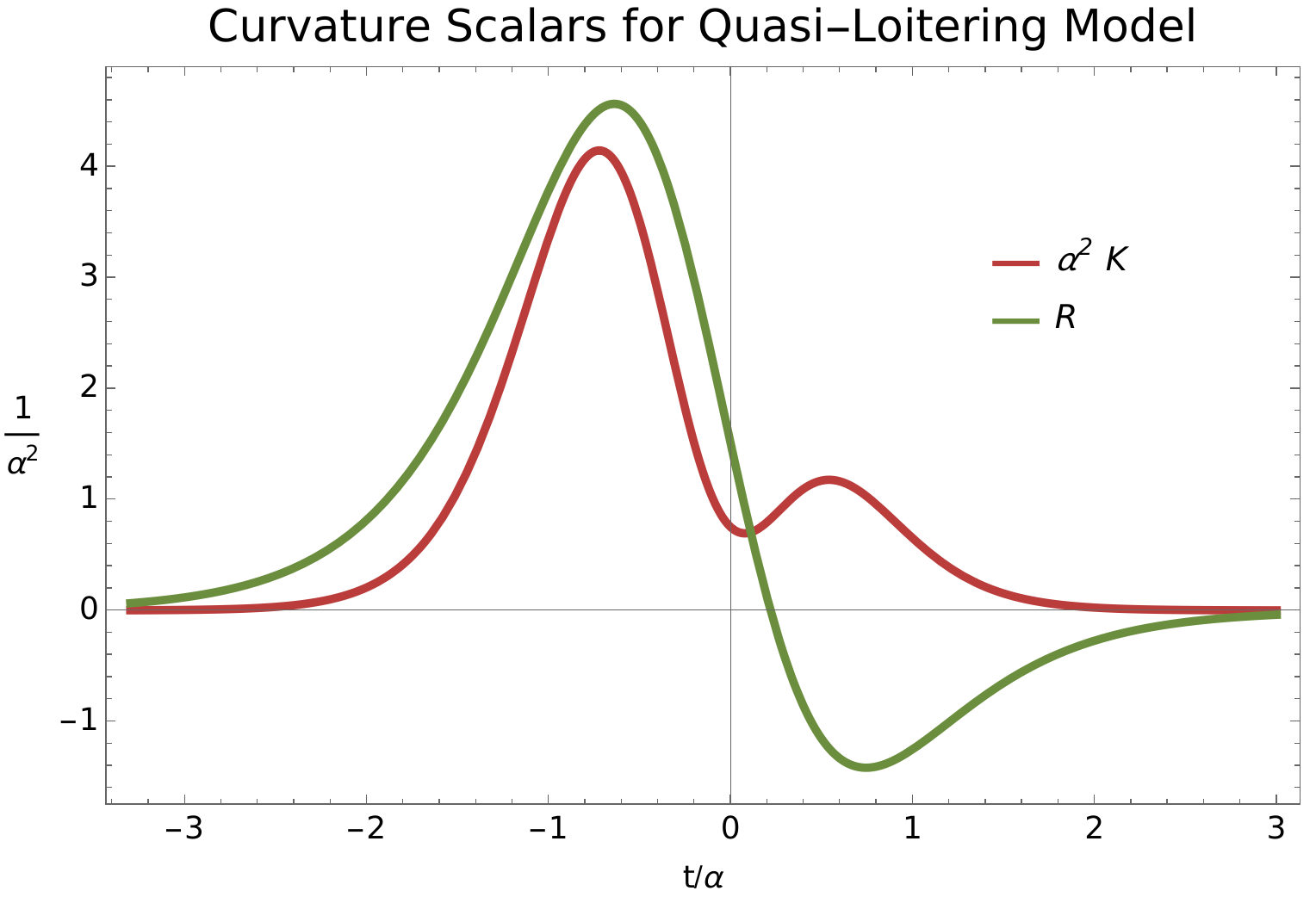}
  \captionof{figure}{Curvature invariants computed from \eq{atanh}. The Ricci scalar (green) and the Kretchmann scalar (red). Model parameters: $a_0=\alpha=1$, $c=2$.}
  \label{cstanh}
%\end{minipage}%
%\begin{minipage}[t]{.5\textwidth}
%  \centering
%  \includegraphics[width=.9\linewidth]{mTilldefor111.jpg}
%  \captionof{figure}{$\tilde{m}_{4}^{2}$ for the case $b=d=\Lambda=1.$}
%  \label{mtilde2}
%\end{minipage}
\end{figure}

Calculation of the Einstein tensor yields non-vanishing components:
\begin{eqnarray}\label{ectanheq}
   G_{tt} &=&-\frac{3 a_0^2 \operatorname{sech}^4\left(\frac{t}{\alpha}\right)}{\alpha^2 (c + a_0 \tanh \left(\frac{t}{\alpha}\right))^2}
 \,, \nonumber \\
   G_{ii} &=& \frac{a_0 \operatorname{sech}^4\left(\frac{t}{\alpha}\right) \left(-3 a_0 + 2 a_0 \cosh\left(\frac{2t}{\alpha}\right) + 2 c \sinh\left(\frac{2t}{\alpha}\right)\right)}{\alpha^2 } \,.
\end{eqnarray}

The energy density is given by $\rho = - G^t{}_t$ and the pressure is $p = G^i{}_i$. A plot elucidating the energy conditions is given in \fig{ectanh}.

\begin{figure}[H]
\centering
%\begin{minipage}[t]{.5\textwidth}
 % \centering
  \includegraphics[width=0.8 \linewidth]{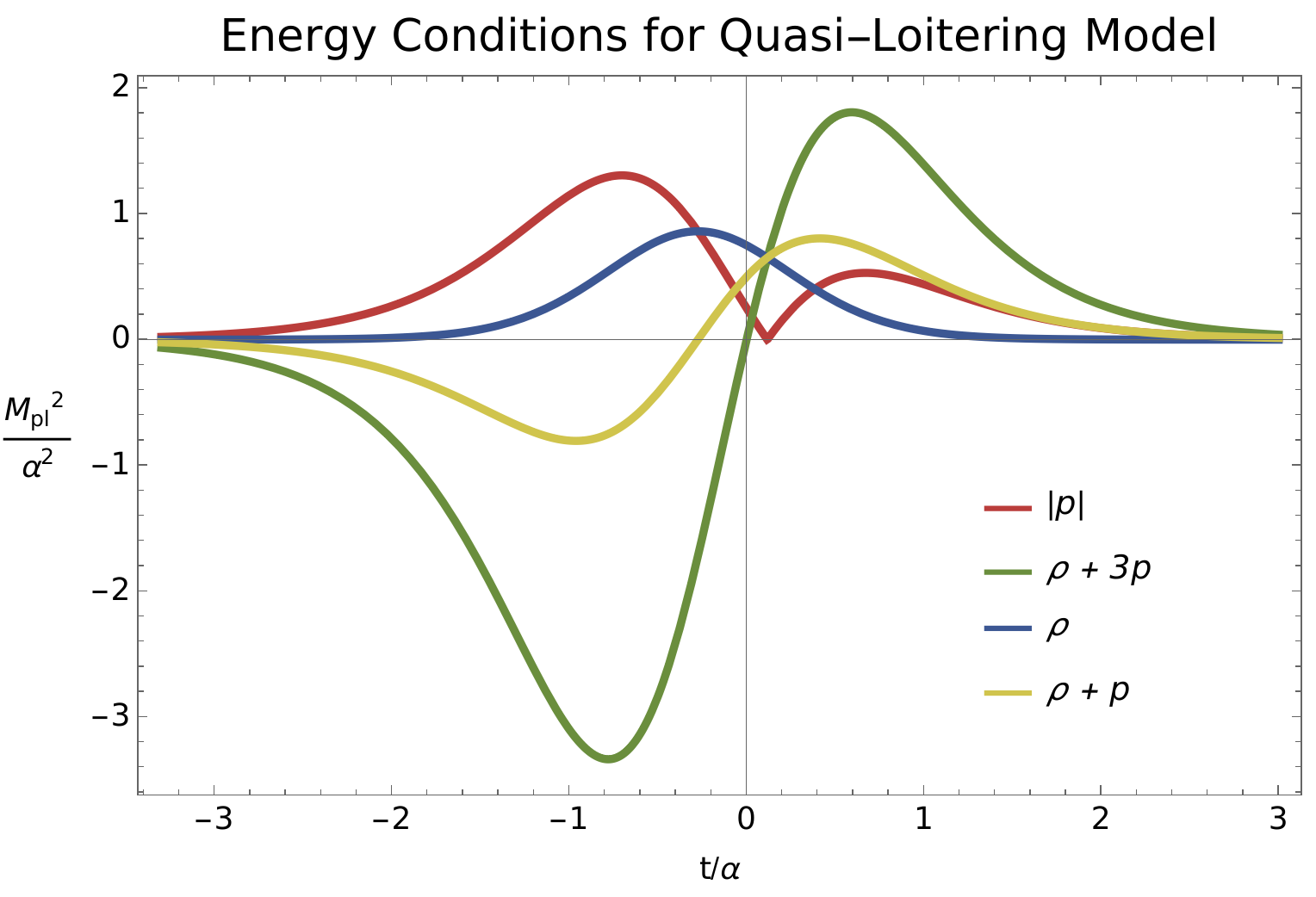}
  \captionof{figure}{Energy conditions from \eq{ectanheq}. Plot of energy density $\rho$ (blue), $\rho + p$  (yellow), $|p|$ (red) and $\rho + 3p$ (green). Legend quantities are rendered dimensionless by the factor $\alpha^2/M^2_{pl}$. Model parameters: $a_0=\alpha=1$, $c=2$.}
  \label{ectanh}
%\end{minipage}%
%\begin{minipage}[t]{.5\textwidth}
%  \centering
%  \includegraphics[width=.9\linewidth]{mTilldefor111.jpg}
%  \captionof{figure}{$\tilde{m}_{4}^{2}$ for the case $b=d=\Lambda=1.$}
%  \label{mtilde2}
%\end{minipage}
\end{figure}

At early times, all of the classical energy conditions are violated. At late times the WEC, NEC and SEC are all satisfied. Only the DEC is mildly violated at late times; however, there exists a brief period for small positive $t$, during which \emph{all} of the energy conditions are satisfied. This period lasts from $t=0$ when the SEC (green curve) starts to be satisfied until the moment the red curve grows larger than the blue curve where the DEC begins to be violated. This healthy period can be past extended to accommodate an inflationary phase where only the SEC is violated (starting where NEC is valid and yellow curve intersects the horizontal line). The inflationary period naturally ends when the SEC starts to be obeyed (green curve intersects the horizontal line). For the ANEC, we evaluate \eq{anec}
\begin{comment}
\begin{equation}
     \int_{- \infty}^\infty (\rho + p ) d\lambda=
\frac{2 \left(\frac{c}{a_0 \tanh \left(\frac{t}{\alpha }\right)+c}+\log \left(a_0 \sinh \left(\frac{t}{\alpha }\right)+c \cosh \left(\frac{t}{\alpha }\right)\right)-\log \left(\cosh \left(\frac{t}{\alpha }\right)\right)-1\right)-\frac{a_0^2}{\left(a_0 \sinh \left(\frac{t}{\alpha }\right)+c \cosh \left(\frac{t}{\alpha }\right)\right){}^2}}{\alpha  a_0}
Big|_{-\infty}^{\infty}
\leq 0 \,.
\end{equation}
\end{comment}
\begin{align}
    \int_{- \infty}^\infty \left(\rho + p \right) d\lambda &=  \frac{2}{\alpha a_0} \left(  \frac{c}{a_0 \tanh \left(\frac{t}{\alpha }\right)+c}+\log \left(a_0 \sinh \left(\frac{t}{\alpha }\right)+c \cosh \left(\frac{t}{\alpha }\right)\right) \right. \nonumber \\
    & \qquad \qquad  \left. \left. -\log \left(\cosh \left(\frac{t}{\alpha }\right)\right) -\frac{a_0^2}{2\left(a_0 \sinh \left(\frac{t}{\alpha }\right)+c \cosh \left(\frac{t}{\alpha }\right)\right)^2} -1 \right) \right|_{-\infty}^\infty \nonumber \\
    & \leq 0 \,.
\end{align}

% \begin{equation}
%      \int_{- \infty}^\infty (\rho + p ) d\lambda= \frac{2 a_0 \tanh\left(\frac{t}{\alpha}\right) \left(a_0 + c \tanh\left(\frac{t}{\alpha}\right)\right)}{c \alpha \left(c + a_0 \tanh\left(\frac{t}{\alpha}\right)\right)}
%  \Big|_{-\infty}^{\infty}
% = 0 \,.
% \end{equation}
Hence, the ANEC is violated.

\section{Discussion}\label{conclude}
Here we summarize our key findings. We see from the energy condition plots above that all of the models we have scrutinized embrace at least some period of null energy condition violation:
\begin{conj}\label{conjnec}
    In General Relativity, every smooth, non-constant scale factor a(t) of a geodesically complete, flat FRW spacetime must violate the NEC during at least some period of time.
\end{conj}
%The above is obviously true for all bouncing FRW cosmological models. 
For example, it is well known in a flat $k=0$ (or hyperbolic $k=-1$) FRW spacetime, the NEC must be violated in order to achieve a cosmological bounce \cite{Hawking:1970zqf}. During a bounce through a local minimum, the universe transitions from a contracting phase, $H<0$, to an expanding phase, $H>0$. This transition inherently requires that at the point where the contraction halts and expansion begins, the derivative of the Hubble parameter, $\dot H$, must be positive. Since $\dot H = - 4 \pi G (\rho + p)$, we must have $\rho + p<0$, signalling violation of the NEC. Thus, any such spacetime which exhibits a bounce, or bounces, for any part of it's history must violate NEC.  Likewise any super-inflationary models (by definition) have $\dot H>0$ and violate NEC. With non-zero curvature $k$, one can produce a bounce without violating NEC since, $\dot H = - 4 \pi G (\rho + p) + k/a^2$, can become positive at the bounce in the case of positive $k$; however, such universes are closed if the matter sector obeys NEC and SEC.
While the NEC holds, $\ddot a/a = H^2 + \dot H = - 4 \pi G (\rho/3 + p)$, and $H=0$ at the bounce, $\ddot a/a>0$ and the SEC is violated. Curvature bounces are discussed in detail in \cite{1978SvAL....4...82S,Graham:2011nb,Burwig2025}. 

Conjecture \ref{conjnec} clearly holds in the case of the general polynomial inflationary models given by \eq{apoly}. Any such model with non-constant scale factor that are complete will be inflationary by Corollary~\ref{cor: complete are inflationary} and clearly contain (at least) one bounce (hence, violating the NEC).  

It remains only to show that all models which are geodesically complete (in principle, do not bounce) and involve an inflationary phase of arbitrary length violate NEC; although, as one can see \emph{all} of the above models satisfy our definition of a bouncing cosmology Defn.~\ref{ELbouncedefn} (essentially, a place where $\dot a=0$). While most of the models demonstrate the bounce phase clearly, it is worth revisiting 
the eternal inflationary model \eq{aplusc}, in this context. 

Our exponential model \eq{aplusc}, can be expanded in a Taylor expansion. The Maclaurin series of the relevant portion of the scale factor is
\begin{equation}\label{polyexp}
   e^{\frac{2t}{\alpha}} = \sum_{n=0}^{\infty} \frac{\left(\frac{2t}{\alpha}\right)^n}{n!} = 1 + \frac{2t}{\alpha} + \frac{\left(\frac{2t}{\alpha}\right)^2}{2!} + \frac{\left(\frac{2t}{\alpha}\right)^3}{3!} + \cdots \,,
\end{equation}
and is compatible with \eq{gapoly}. Applied to the scale factor \eq{aplusc}, the odd powered polynomial terms are negative as $t \rightarrow - \infty$, and effectively cancel the positive even powered terms ultimately creating a ``bounce point at infinity" in accordance with Defns.~\ref{ELbouncedefn},\ref{defn: bounce at infinity} at $t \rightarrow - \infty$, with constant scale factor $a=c$ and with $H=0$. Thus, our eternal inflating model \eq{aplusc} is also a bouncing cosmology par Defn.\ref{ELbouncedefn}. To elucidate our discussion we remind the reader of the general polynomial bounce discussed in \S{\ref{sec:poly}} and \eqs{gapoly}{apoly}. In essence the bounce point of a finite series polynomial, such as \eq{apoly}, is in this case translated to $t \rightarrow -\infty$, where 
\begin{equation}
    \lim_{t \rightarrow -\infty} H =0 \,.
\end{equation}

In further support of this notion of the ``bounce", one sees that the bouncing cosh model \eq{acosh}, having $\dot a (t_0)=0$ at the bounce point,
is recovered by adding together two copies of the exponential model \eq{aplusc}, one with $\alpha>0$ and one with $\alpha<0$, each having its own ``bounce at infinity" at $t= \pm \infty$ where $H=0$:
\begin{equation}
a_0 \, e^\frac{2 t}{\alpha} + c + a_0 \, e^\frac{-2 t}{\alpha} + c = 2 \left(c + a_0 \cosh\left(\frac{2 t}{\alpha}\right)\right) \,.
\end{equation}

On the path to $t=+\infty$, we find all terms in the polynomial expansion \eq{polyexp} are positive, and sum to give a de Sitter inflating phase with constant $H$:
\begin{equation}
     \lim_{t \rightarrow \infty} H =2/\alpha \,.
\end{equation}

Indeed, \emph{all} of the geodesically complete spacetimes presented have at least some period of NEC violation, and a phase with $\dot a=0$, and may be classified as bouncing models in accordance with Defn.~\ref{ELbouncedefn}:
\begin{conj}
    In General Relativity, every smooth, non-constant scale factor a(t) of a geodesically complete FRW spacetime must experience at least one bounce, where said bounce may be at infinity.
\end{conj}

Combining this with Cor.\ref{cor: complete are inflationary} we formulate a provocative composite conjecture:
\begin{conj}\label{conj:both}
    Every geodesically complete eternal spacetime which admits a neighborhood isometric to a GFRW with a non-constant scale factor will inflate for some time and bounce at least once, where said bounce may be at infinity.
\end{conj}

Naturally, the question of whether our own observable universe started with an inflationary phase or cosmological bounce is of paramount importance; however, ultimately, in the context of an inflationary model, our access to primordial perturbations is limited to a small number of e-foldings, as discerned through observations of the cosmic microwave background (CMB). Observables from the early universe capable of probing and constraining models are plentiful and include primordial gravitational waves, non-gaussianities, isocurvature modes, various reheating predictions, of course the two-point correlation power spectrum and others \cite{Mather:1990tfx,Daly:1991uob,Mather:1993ij,Hu:1994bz, Linde:1996gt,Moroi:2001ct,Lyth:2001nq,Smith:2005mm, Smith:2006xf,WMAP:2012nax,Dent:2012ne,Cook:2015vqa,Planck:2018vyg}. Incredible progress has been made in detecting or attempting to detect these observables and we remain optimistic with respect to eventually identifying the underlying mechanisms responsible for generating our universe. But there are many technical challenges ahead, and limited data can make it difficult to definitively distinguish between bouncing and inflationary models due to degeneracy problems \cite{Easson:2010zy,Easson:2010uw,Easson:2012id,Vennin:2015vfa,Mishra:2021wkm,Ben-Dayan:2023rlj}, or other obstructions.~\footnote{For example censorship behind Cauchy horizons or, at a more subtle level, G{\"o}del's incompleteness theorems introduce exotic phenomena, rooted in logical and computational complexities  \cite{Godel1992}.} 

In our discussion of a fully geodesically complete spacetime that exists for an eternity, we encounter a tautological paradigm shift: the key focus moves from categorizing spacetimes as ``inflationary" or ``bouncing" to emphasizing their geodesic completeness. We discover from Conj. \ref{conj:both}, although both inflation and cosmological bounces play crucial roles in forming a geodesically complete spacetime, they are often transient stages in the vast timeline of an eternal cosmos.~\footnote{An observer inhabiting a decelerating universe, who holds completeness as a foundational principle, would predict the existence of a prior inflationary period by Cor. \ref{cor: complete are inflationary}, and even anticipate the likelihood of a future phase of acceleration.}

Developing a comprehensive model for an eternal universe goes far beyond the initial explorations presented here, requiring studies of anisotropic models and a move away from the simplistic perfect fluid description. One must incorporate our known observable universe with a proper reheating period, nucleosynthesis, and periods of radiation, matter and dark energy domination. This ambitious endeavor awaits those venturing along future-directed (hopefully complete) timelike geodesic rays.
%%
%\section*{Acknowledgements}
\acknowledgements
It is a pleasure to thank S.~Alexander, R.~Brandenberger, P.~Davies, G.~Ellis, G.~Geshnizjani, A.~Karch, W.~Kinney, B.~Kotschwar, M.~Tomasevic, M.~Parikh, T.~Vachaspati, E.~Verheijden, A.~Vikman and S.~Watson for useful discussions and correspondence. DAE is supported in part by the U.S. Department of Energy, Office of High Energy Physics, under Award Number DE-SC0019470.  JEL is supported by the gracious patronage of Ms. Deborah L Nelson, Dr. Edward J Lesnefsky, and Mrs. Laken S Lesnefsky.
%%%%%%%%%%%%%%%%%%%%%%%%%%%%%%%%%%%%%%%%%%%%%%%%%%%%%%%%%%%%%%%%%%%%%%%%%%%%%%%%%%%%%%%%%%%%%%%%%%%%%%%%%%%%%%%%%%%%
\newpage
\appendix
\renewcommand{\theequation}{\thesection.\arabic{equation}}
\addappheadtotoc
\appendixpage

\section{Energy conditions}\label{appa}

For this discussion we will assume General Relativity with matter described by a perfect fluid stress energy tensor
\begin{equation}
    T_{\mu\nu} = (\rho + p)U_\mu U_\nu + p g_{\mu\nu} \,,
\end{equation}
where $p$ and $\rho$ and the pressure and energy density of the fluid and $U^\mu$ is the fluid four-velocity.
\subsection{Weak Energy Condition (WEC)}
\noindent
The WEC states that $T_{\mu\nu} t^\mu t^\nu \ge 0$, $\forall$ timelike vectors $t^\mu$. This is equivalent to $\rho\ge 0$ and $\rho + p \ge 0$.
\subsection{Null Energy Condition (NEC)}
\noindent
The NEC states that $T_{\mu\nu} \ell^\mu \ell^\nu \ge 0$, $\forall$ null vectors $\ell^\mu$. This is equivalent to $\rho + p \ge 0$. The energy density may now be negative as long as it is compensated for by sufficient positive pressure.
\subsection{Dominant Energy Condition (DEC)}
\noindent
The DEC states that $T_{\mu\nu} t^\mu t^\nu \ge 0$, $\forall$ timelike vectors $t^\mu$ as well as $T_{\mu\nu}T^\nu{}_\lambda t^\mu t^\lambda \le 0$. This is equivalent to $\rho\ge |p|$. 
\subsection{Strong Energy Condition (SEC)}
\noindent
The SEC states that $T_{\mu\nu} t^\mu t^\nu \ge \frac{1}{2} T^\lambda{}_\lambda t^\sigma t_\sigma$, $\forall$ timelike vectors $t^\mu$. This is equivalent to $\rho + p \ge 0$ and $\rho + 3p \ge 0$.

\noindent
Note that the DEC$\implies$WEC, WEC$\implies$NEC, SEC$\implies$NEC; and SEC$\rlap{$\quad\not$}\implies$WEC.

\section{Theorem 1 Item 3 Proof}\label{proofof3}
The denominators of \eqs{geodrayf}{geodrayp} in the spacelike case of \(w_0 = -1\) evaluate to \(\sqrt{1 - a^2 \left( t \right)}\). We now invoke an (enthymeme) assumption that the domain of discourse of this proof exists for real valued differentiable manifolds, so \(\sup_{\mathbb{R}^1_1} a = 1\) if the denominator is to evaluate to a real number between \(\left[ 0 , 1 \right]\) and not become imaginary. We now have 3 cases;  unlike the timelike case where there are two disconnected lightcones (hence, two integral orientations $\int_{t_0}^\infty$ and $\int^{t_0}_{-\infty}$), the null case, despite the fact that the null cone is connected one remains either future pointing or past pointing, the unit spacelike hyperquadratic in $T_p \mathcal{M}$ under the exponential map $\exp_p$ is connected with potential (closed) periodic spacelike curves\footnote{It is a well-known (see \cite{Beem1996} Ch. 3 or \cite{Romero1994,Sanchez1998}) that GFRWs with complete Riemannian moieties are globally hyperbolic, so that closed causal curves cannot occur.}.  There are 3 subcases: $d \pi_{\mathbb{R}^1_1} \dot{\gamma} = \frac{d t}{d \lambda} \frac{\partial}{\partial t}$ is future directed, past directed, or purely spacelike (achronal).  In the achronal case, the geodesic is complete by the assumption that $\left( \Sigma , g_\Sigma \right)$ is complete.  In both the past and future cases,  the integrand of $\frac{a}{\sqrt{1 - a^2}} > a$ is pointwise $\forall \alpha \in \mathbb{R}^1_1$ because $\sqrt{1 - a^2} < 1$ by the spacelike assumption and the assumption that this is a $\mathbb{R}$ valued differentiable manifold.  Thus, if the null integral diverges so will the spacelike integrals. There is one more technicality to consider, that of periodic spacelike geodesics. In this case, the integrands are everywhere strictly positive, so with every period another strictly positive mass is added to the integral, and as the number of orbits becomes unbounded approaching cardinal $\aleph_0$, the integral diverges.  

Considering this, a periodic spacelike trajectory will have $d \pi_{\mathbb{R}^1_1} \dot{\gamma} = \frac{d t}{d \lambda} \frac{\partial}{\partial t}$, partially future pointing and partially past pointing so that the two distinct past (future) directed integrals coincide.  For the case of unbounded trajectories, if the null integral diverges so will the two past (future) pointing spacelike integrals, for the reason mentioned above.  Thus, for spacelike completeness, one must be both past (future) null complete - along with the achronal case which holds by assumption - so full null completeness must hold.  The null case of the integral of item (2) bounds the non-infinite $a$ spacelike case from below.  Thus, if the GFRW is past and / or future null complete then - in the case of a non-infinite $a$ - the GFRW is spacelike complete and the first portion of item (3) is proven.  For completeness, we include the case where $a$ potentially diverges and becomes infinite for a finite $t \in \mathbb{R}$.  In this case a ``big rip'' boundary will form and a geodesic will encounter it within finite parameter.

Finally, as a matter of definition, a Lorentzian spacetime is (fully) geodesically complete if every geodesic ray is both past and future, timelike, null, and spacelike complete; thus, item (5) is proven.

\section{Generalized Robertson Walker Spacetimes} \label{appendix: sec: gfrw}
The well-known Friedmann Robertson Walker (FRW) spacetime can be extended to a Generalized Friedmann Robertson Walker (GFRW) spacetime construction.  Many properties of FRW spacetimes do not utilize the spaceform homogeniety and isotropy of the spatial section and are predicated upon behavior of the scale factor only.  The constant sectional curvature assumptions of the FRW spatial sections can be weakened to any complete Riemannian manifold.

\subsection{Definition} \label{appendix: subsec: gfrw defn}
Building on the warped product formalism of \cite{Bishop1969} defined for Riemannian manifolds, it can be repurposed and melded with the FRW formalism to define a Generalized Friedmann Robertson Walker spacetime as in \cite{ONeill1983,Sanchez1998,Romero1994}.

\begin{definition} \label{appendix: defn: gfrw defn}
    Let $\mathbb{R}^1_1$ be the negative definite manifold $\left( \mathbb{R} , -dt^2 \right)$ and $\left( \Sigma , g_\Sigma \right)$ be any geodesically complete Riemannian manifold, and $a \in \mathcal{C}^
    \infty \left( \mathbb{R}^1_1 \right)$ a smooth, strictly positive function on $\mathbb{R}$.  A \emph{Generalized Robertson Walker spacetime}, denoted $\mathbb{R}^1_1 \times_a \Sigma$, is the (topological) Cartesian product $\mathbb{R}^1_1 \times \Sigma$ with metric $d \pi^*_{\mathbb{R}^1_1} \left( -dt^2 \right) + a^2 \left( t \right) d \pi^*_\Sigma g_\Sigma$ where $\pi_{\mathbb{R}^1_1}$, $\pi_\Sigma$ are the canonical projections onto $\mathbb{R}^1_1$, $\Sigma$ guaranteed by the Cartesian product construction.  From this point on, unless otherwise required, the metric pullbacks $d \pi^*_{\mathbb{R}^1_1} \left( -dt^2 \right)$, $d \pi^*_\Sigma g_\Sigma$ will have $d \pi^*_{\mathbb{R}^1_1}$, $d \pi^*_\Sigma g_\Sigma$ omitted.
\end{definition}

For enumerated derived properties of the connection, geodesics, and curvatures of GFRWs see \cite{ONeill1983} Chapters 7,12.  Here we now show how GFRW geodesic equation projections relate to the common geodesic equation
\begin{equation}\label{geodeq}
\frac{d^2 x^\sigma}{d \lambda^2} + \Gamma^\sigma_{\alpha \beta} \frac{d x^\alpha}{d \lambda} \frac{d x^\beta}{d \lambda}=0 \,.
\end{equation}

\subsection{Geodesics in GFRWs} \label{appendix: subsec gfrw geodesics}
Let $\mathcal{M} = \mathbb{R}^1_1 \times_a \Sigma$ be a GFRW and $x^\mu \left( \lambda \right) = t \left( \lambda \right) \oplus x^k \left( \lambda \right)$ be a geodesic having affine parameter $\lambda$ with $t = \pi_{\mathbb{R}^1_1 } x^\mu$ being the (timelike) projection on $\mathbb{R}^1_1$, and $x^k = \pi_\Sigma x^\mu$ being the (spacelike) projection on $\left( \Sigma , g_\Sigma \right)$. The geodesic equation as viewed in $\mathcal{M}$:
\begin{equation} \label{appendix: eq: gfrw geodesic eq geom}
\nabla_{\frac{d x^\nu}{d \lambda} \frac{\partial}{\partial x^\nu}} \frac{d x^\mu}{d \lambda} \frac{\partial}{\partial x^\mu} = \left( \frac{d^2 x^\mu}{d \lambda^2} +\Gamma^\mu_{\alpha \beta} \frac{d x^\alpha}{d \lambda} \frac{d x^\beta}{d \lambda} \right) \frac{\partial}{\partial x^\mu} = 0
\end{equation}
The full geodesic equation in $\mathcal{M}$ can be projected into both $\mathbb{R}^1_1 , \left( \Sigma , g_\Sigma \right)$ as differential equations:
\begin{equation} \label{appendix: eq: prop 7.38-t}
    \nabla_{\frac{\partial}{\partial t} } \frac{\partial}{\partial t} = - g_\Sigma \left( \frac{d x^k}{d \lambda} \frac{\partial}{\partial x^k} , \frac{d x^k}{d \lambda} \frac{\partial}{\partial x^k} \right) \cdot \left( a \left( t \left( \lambda \right) \right) \right) \cdot \frac{d a}{d t}
\end{equation}
\begin{equation} \label{appendix: eq: prop 7.38-x}
    \nabla_{\left( \frac{d x^k}{d \lambda} \frac{\partial}{\partial x^k} \right)} \left( \frac{d x^\ell}{d \lambda} \frac{\partial}{\partial x^\ell} \right) = - \frac{2}{a \left( t \left( \lambda \right) \right)} \frac{d a}{dt} \frac{dt}{d \lambda} \frac{d x^\ell}{d \lambda} \frac{\partial}{\partial x^\ell}
\end{equation}
respectively.  In particular, Eq. \ref{appendix: eq: prop 7.38-x} in $\Sigma$ is not quite totally geodesic, but is a pregeodesic because the acceleration is proportional to $\frac{d x^\ell}{d \lambda} \frac{\partial}{\partial x^\ell}$ and thus a geodesic reparameterization exists ``boosting'' away the acceleration the affine parameter $\lambda$.  The image of Eq. \ref{appendix: eq: prop 7.38-x} is the image of a geodesic.  We now invoke the completeness assumption of $\left( \Sigma , g_\Sigma \right)$, yielding $x^k \left( \lambda \right)$ is defined $\forall \, \mathbb{R}$.  Finally, if $x^\mu \left( \lambda \right)$ fails to be complete, it must be due to Eq. \ref{appendix: eq: prop 7.38-t}.

\subsection{Derivation of Warped Product Geodesics}
In fact, \eqs{appendix: eq: prop 7.38-t}{appendix: eq: prop 7.38-x} can be generalized to a warped product of any semi-Riemannian manifolds.  The familiar geodesic equation \eq{geodeq} can be distilled into the generalized form of \eqs{appendix: eq: prop 7.38-t}{appendix: eq: prop 7.38-x}.
In the following we use $f$ for the scale factor. As in \cite{ONeill1983} Chapter 7 in a warped product of any two semi-Riemannian manifolds $B \times_f F $ the first manifold is the leaf manifold $\left( B , \Tilde{g}_B \right)$ - corresponding to $\mathbb{R}^1_1$ - and the second manifold is the fiber $\left( F , \Tilde{g}_F \right)$ - corresponding to $\left( \Sigma , g_\Sigma \right)$.  Additionally, consider a chart $\left( U_B , \{ t^k \} \right)$ of $B$ and a chart $\left( U_F , \{ x^\ell \} \right)$.  The warped product metric is a block diagonal form $g_B + f^2 \left( t^k \right) g_F$.  All factors of $f\left( t^k \right)$ will be explicitly written and terms with a ``tilde'' are viewed to be elements of the constituent spaces and naked (no tilde) terms are elements of the warped product $B \times_f F$, for example the pullback $g_B = d \pi_B^* \Tilde{g}_B$.  Finally, the geodesic equation is calculated over curve $t^k \left( \lambda \right) \oplus x^\ell \left( \lambda \right)$.

For this construction, the Christoffel symbols can be calculated
\begin{align*}
    \Gamma^{t^k}_{t^\ell t^m} &= \Tilde{\Gamma}^{t^k}_{t^\ell t^m}     &     \Gamma^{t^k}_{x^\ell x^m} &= - \left(
    \Tilde{g}_F \right)_{x^\ell x^m} f \left( t^k \right) \left(\mathrm{grad}_{\Tilde{g}_B} f \right)^{t^k}     &     \Gamma^{t^k}_{t^\ell x^m} &= 0  \\
    \Gamma^{x^k}_{x^\ell x^m} &= \Tilde{\Gamma}^{x^k}_{x^\ell x^m}     &     \Gamma^{x^k}_{x^\ell t^m} &=  \frac{1}{f \left( t^m \right)} \frac{\partial f}{\partial t^m} \left( \Tilde{g}_F \right)^{x^k x^a} \left( \Tilde{g}_F \right)_{x^a x^\ell}     &     \Gamma^{x^k}_{t^\ell t^m} &= 0
\end{align*}
where $\left(\mathrm{grad}_{\Tilde{g}_B} f \right)^{t^k}$ is the $t^k$ component of the gradient, and the scale factor of $f : B \rightarrow \mathbb{R}$ is written $f \left( t^k \right)$.

The geodesic equation may be expanded
\begin{align} \label{appendix: eq gfrw prop 7.38 from geodesic eqn}
    0 &= \frac{d^2 x^\sigma}{d \lambda^2} + \Gamma^\sigma_{\alpha \beta} \frac{d x^\alpha}{d \lambda} \frac{d x^\beta}{d \lambda} \nonumber \\
      &= \frac{d^2 t^k}{d \lambda^2} + \Gamma^{t^k}_{\alpha \beta} \frac{d x^\alpha}{d \lambda} \frac{d x^\beta}{d \lambda} + \frac{d^2 x^\ell}{d \lambda^2} + \Gamma^{x^\ell}_{\alpha \beta} \frac{d x^\alpha}{d \lambda} \frac{d x^\beta}{d \lambda} \nonumber \\
      &= \frac{d^2 t^k}{d \lambda^2} + \Gamma^{t^k}_{t^a t^m} \frac{d t^a}{d \lambda} \frac{d t^m}{d \lambda} + \Gamma^{t^k}_{x^a x^m} \frac{d x^a}{d \lambda} \frac{d x^m}{d \lambda} \nonumber \\
      & \qquad + \frac{d^2 x^\ell}{d \lambda^2} + \Gamma^{x^\ell}_{x^a x^m} \frac{d x^a}{d \lambda} \frac{d x^m}{d \lambda} + \Gamma^{x^\ell}_{x^a t^m} \frac{d x^a}{d \lambda} \frac{d t^m}{d \lambda} + \Gamma^{x^\ell}_{t^a x^m} \frac{d t^a}{d \lambda} \frac{d x^m}{d \lambda} \nonumber \\
      &= \frac{d^2 t^k}{d \lambda^2} + \Tilde{\Gamma}^{t^k}_{t^a t^m} \frac{d t^a}{d \lambda} \frac{d t^m}{d \lambda} - \left(
    \left( \Tilde{g}_F \right)_{x^\ell x^m} \frac{d x^a}{d \lambda} \frac{d x^m}{d \lambda} \right) \cdot \left( f \left( t^k \right) \left( \mathrm{grad}_{\Tilde{g}_B} f \right)^{t^k} \right) \nonumber \\
    & \qquad + \frac{d^2 x^\ell}{d \lambda^2} + \Tilde{\Gamma}^{x^\ell}_{x^a x^m} \frac{d x^a}{d \lambda} \frac{d x^m}{d \lambda} + \frac{2}{f \left( t^m \right)} \frac{\partial f}{\partial t^m} \left( \Tilde{g}_F \right)^{x^\ell x^a} \left( \Tilde{g}_F \right)_{x^a x^k} \frac{d x^a}{d \lambda} \frac{d t^m}{d \lambda}
\end{align}
One then notes that the $t^k$ terms and $x^\ell$ terms in the final expression of Eq.~\ref{appendix: eq gfrw prop 7.38 from geodesic eqn} are linearly independent, so for the geodesic equation to be homogeneous each vector term must be homogeneous.  Thus
\begin{equation} \label{appendix: eq prop 7.38-t component}
\frac{d^2 t^k}{d \lambda^2} + \Tilde{\Gamma}^{t^k}_{t^a t^m} \frac{d t^a}{d \lambda} \frac{d t^m}{d \lambda} = \left(
    \left( \Tilde{g}_F \right)_{x^\ell x^m} \frac{d x^a}{d \lambda} \frac{d x^m}{d \lambda} \right) \cdot \left( f \left( t^k \right) \left(\mathrm{grad}_{\Tilde{g}_B} f \right)^{t^k} \right)    
\end{equation}
\begin{equation} \label{appendix: eq prop 7.38-x component}
    \frac{d^2 x^\ell}{d \lambda^2} + \Tilde{\Gamma}^{x^\ell}_{x^a x^m} \frac{d x^a}{d \lambda} \frac{d x^m}{d \lambda} = - \frac{2}{f \left( t^m \right)} \frac{\partial f}{\partial t^m} \delta^{x^\ell}_{x^k} \frac{d x^k}{d \lambda} \frac{d t^m}{d \lambda}   
\end{equation}
are recovered.  In order to eventually compare Eqs. \ref{appendix: eq prop 7.38-t component}, \ref{appendix: eq prop 7.38-x component} to Eqs. \ref{appendix: eq: prop 7.38-t}, \ref{appendix: eq: prop 7.38-x} in the case of a GFRW one can consolidate using formal notation~\footnote{A pedantic reader may have noticed the ``free'' index of $t^\ell$ on the right hand side of Eq. \ref{appendix: eq prop 7.38-t geometer}.  This is an unfortunate thorn of mixing ``component'' and ``geometer'' notation.  Strictly -- and in most geometric texts such as \cite{ONeill1983,Beem1996,Peterson2006} the projection of the geodesic onto $B$ is written $\tau : J \rightarrow B$ and the function written as the composition $\left(f \circ \tau \right)\left( \lambda \right)$.  While being precise, this minimizes the ``component'' moiety of the expression which we are trying to draw attention to.  A similar occurrence appears in Eq. \ref{appendix: eq prop 7.38-x geometer}.  We have been a bit caviler with the notation in order to show the relationship between the two notations - we hope the reader will forgive this indulgence of an unbalanced ``free'' index.}
\begin{equation} \label{appendix: eq prop 7.38-t geometer}
    \nabla_{\left( \frac{d t^\ell}{d \lambda} \frac{\partial}{\partial t^\ell} \right)} \left( \frac{d t^\ell}{d \lambda} \frac{\partial}{\partial t^\ell} \right) = 
    \Tilde{g}_F \left( \frac{d x^k}{d \lambda} \frac{\partial}{\partial x^k} , \frac{d x^m}{d \lambda} \frac{\partial}{\partial x^m} \right) \cdot \left( f \left( t^\ell \right) \mathrm{grad}_{\Tilde{g}_B} f \right)
\end{equation}
\begin{equation} \label{appendix: eq prop 7.38-x geometer}
    \nabla_{\left( \frac{d x^\ell}{d \lambda} \frac{\partial}{\partial x^\ell} \right)} \left( \frac{d x^\ell}{d \lambda} \frac{\partial}{\partial x^\ell} \right) = - \frac{2}{f \left( t^m \right)} \frac{\partial f}{\partial t^m} \frac{d t^m}{d \lambda} \frac{d x^k}{d \lambda} \frac{\partial}{\partial x^k}
\end{equation}

In the case of a GFRW, $\left( B , \Tilde{g}_B \right) = \mathbb{R}^1_1$ and $\left( F , \Tilde{g}_F \right) = \left( \Sigma , g_\Sigma \right)$; with respect to Eqs. \ref{appendix: eq prop 7.38-t geometer}, \ref{appendix: eq prop 7.38-x geometer} the aforementioned Eqs. \ref{appendix: eq: prop 7.38-t}, \ref{appendix: eq: prop 7.38-x} are recovered, as desired.

%\subsection{Derivation of Warped Product Geodesics}
%In fact, equations can be generalized to a warped product of any semi-Riemannian manifolds. 
%The simplified geodesic equation may be transformed into the generalized form. In the context of a warped product of two semi-Riemannian manifolds $B \times_a F$, with $B$ being the base manifold and $F$ the fiber, and using a chart for each. The warped product metric takes the form $g_B + a^2(t^k) g_F$, explicitly incorporating all factors of $a(t^k)$. The Christoffel symbols are recalculated, reflecting the function $a$ instead of $f$. The geodesic equation, when broken down, addresses the components related to both the base and the fiber, accounting for the influence of $a$ throughout the derivation. This adaptation ensures the completeness assumption of the fiber manifold $\Sigma$, allowing the examination of the geodesic's completeness to hinge on the equation associated with the base manifold $B$. The complete derivation maintains the rigor of the original analysis, tailoring it specifically to the dynamics introduced by the warping function $a$.

\subsection{A Useful Constant} \label{appendix: subsection constant}
The question of notation nonwithstanding, Eqs. \ref{appendix: eq prop 7.38-t component}, \ref{appendix: eq prop 7.38-x component}, Eqs. \ref{appendix: eq prop 7.38-t geometer}, \ref{appendix: eq prop 7.38-x geometer}, and / or Eqs. \ref{appendix: eq: prop 7.38-t}, \ref{appendix: eq: prop 7.38-x} form a system of coupled ODEs which, at least upon a cursory inspection, appear formidable.  However, they can be simplified with the introduction of a geodesic dependent constant.

Given a geodesic $t^\ell \left( \lambda \right) \oplus x^k \left( \lambda \right)$ a constant\footnote{We include a reference to the geodesic $\gamma$ as a subscript in order to remind ourselves that constant $\xi^2_\gamma$ is geodesic dependent.}
\begin{equation} \label{appendix: eq xi const}
    \xi^2_{t^\ell \oplus x^k} = f^4 \left( t^\ell \right) \Tilde{g}_F \left( \frac{d x^k}{d \lambda} \frac{\partial}{\partial x^k} , \frac{d x^k}{d \lambda} \frac{\partial}{\partial x^k} \right)
\end{equation}
can be defined along the image of the geodesic.  The constant $\xi^2_{t^\ell \oplus x^k}$ is not a Noether conserved quantity a propos any symmetry per se, but represents the initial spatial ``speed'' of a particle in ``free fall'' upon a geodesic.
\begin{proof}
Starting with geodesic $t^\ell \left( \lambda \right) \oplus x^k \left( \lambda \right)$ with velocity vector $\frac{d}{d \lambda}$ one calculates along the image of $t^\ell \left( \lambda \right) \oplus x^k \left( \lambda \right)$:
\begin{align*}
    \frac{d \xi^2_{t^\ell \oplus x^k}}{d \lambda} & \overset{?}{=} \frac{d}{d \lambda} \left[ f^4 \left( t^\ell \right) \Tilde{g}_F \left( \frac{d x^k}{d \lambda} \frac{\partial}{\partial x^k} , \frac{d x^k}{d \lambda} \frac{\partial}{\partial x^k} \right) \right] \\
    &= 4 f^3 \left( t^\ell \right) \frac{d f}{d t^\ell} \frac{d t^\ell}{d \lambda} \Tilde{g}_F \left( \frac{d x^k}{d \lambda} \frac{\partial}{\partial x^k} , \frac{d x^k}{d \lambda} \frac{\partial}{\partial x^k} \right) + 2 f^4 \left( t^\ell \right) \Tilde{g}_F \left( \nabla_{\frac{d}{d \lambda}} \frac{d x^k}{d \lambda} \frac{\partial}{\partial x^k} , \frac{d x^k}{d \lambda} \frac{\partial}{\partial x^k} \right) \\
    &= 4 f^3 \left( t^\ell \right) \frac{d f}{d t^\ell} \frac{d t^\ell}{d \lambda} \Tilde{g}_F \left( \frac{d x^k}{d \lambda} \frac{\partial}{\partial x^k} , \frac{d x^k}{d \lambda} \frac{\partial}{\partial x^k} \right) - 4 \frac{f^4 \left( t^\ell \right)}{f \left( t^\ell \right)} \Tilde{g}_F \left( \frac{d x^k}{d \lambda} \frac{\partial}{\partial x^k} , \frac{d x^k}{d \lambda} \frac{\partial}{\partial x^k} \right) \\
    &= 0
\end{align*}
where Eq. \ref{appendix: eq: prop 7.38-x} was inserted to evaluate $\nabla_{\frac{d}{d \lambda}} \frac{d x^k}{d \lambda} \frac{\partial}{\partial x^k}$.
\end{proof}

Utilizing $\xi^2_{t^\ell \oplus x^k}$ one can revisit Eq. \ref{appendix: eq: prop 7.38-t}.  The potentially complex - even before integration! - term $ f \left( t^\ell \right) \Tilde{g}_F \left( \frac{d x^k}{d \lambda} \frac{\partial}{\partial x^k} , \frac{d x^k}{d \lambda} \frac{\partial}{\partial x^k} \right)$ can be eliminated for $\frac{\xi^2_{t^\ell \oplus x^k}}{f^3 \left( t^\ell \right)}$.  The $\nabla_\frac{\partial}{\partial t} \frac{\partial}{\partial t}$ equation now reads
\begin{equation} \label{appendix: eq t xi}
    \nabla_\frac{\partial}{\partial t} \frac{\partial}{\partial t} = \frac{\xi^2_{t^\ell \oplus x^k}}{f^3 \left( t^\ell \right)} \mathrm{grad}_{\Tilde{g}_B} f
\end{equation}
which has a more tractable integral solution: Thm 1 follows.

\section{Geodesics of Selected Models} \label{appendix: quadratic geodesics}
The geodesics completeness of GFRWs is completely characterized by Thm. \ref{ledthm}; and in particular this is done without solving the geodesic equation in the usual fashion - namely the end state of a real valued function $\gamma : J \subset \mathbb{R} \rightarrow \mathcal{M}$ satisfying initial conditions with $\nabla_{\dot \gamma} \dot \gamma = 0$.  Instead one arrives at the integral of Eq. \ref{eq:gfrw soltn timelike integral} which reaps the domain of $\lambda \left( t \right) \in J$ from which geodesic completeness can be deduced.  This feature gives Thm. \ref{ledthm} much of its power, however, it can be useful to examine geodesics in the usual sense.  
%In addition to pedagogy, perhaps dialectic proffered here will suffice to convince any readers who are still skeptical of the main result.

While the geodesics for any model discussed in this opus can be solved for at least numerically for given initial conditions in a computer algebra system of your choice\footnote{Or a standard reference such as \cite{Polyanin2003,Abramowitz1965-ho}}, we have elected to solve geodesics for the ``$+ c$'' model of Section \ref{inflate} and the quadratic model of Section \ref{sec: quadratic model}.  Both of these models have the advantage of closed form ODE solutions in addition to their numeric solutions.  Examining the raw geodesics, one verifies that, in fact, the (timelike) lengths diverge with unbounded affine parameter.

\subsection{A Comment on Initial Conditions}
Instead of considering a geodesic of $\gamma \left( \lambda \right)$ for affine parameter $\lambda$, one can often solve for closed form solutions for $\gamma \left( t  \right)$ where $t$ is the time coordinate of GFRW $\mathbb{R}^1_1 \times_f \Sigma$.  In a GFRW coordinate chart $\left\{ t, x^k \right\}$ one can transform the usual geodesic equation into
\begin{equation} \label{eq: appendix example geodesics geodesic as t}
    \frac{d^2 x^k}{d t^2} + \Gamma^k_{a b} \frac{d x^a}{d t}\frac{d x^b}{d t} - \Gamma^t_{a b} \frac{d x^a}{d t} \frac{d x^b}{d t} \frac{d x^k}{d t} = 0
\end{equation}
where $x^k \left( t \right)$ is a (geodesic) solution utilizing the time coordinate $t$ as an independent variable.

However, one must relate the initial conditions of $\left. \frac{d \gamma}{d \lambda} \right|_{\lambda_0}$ to those of $\left. \frac{d \gamma}{d t} \right|_{t_0}$; in particular the two integration constants which we will denote $\left\{ x^k_0 , \omega^k_0 \right\}$.  Here $x^k_0 = x^k \left( t_0 \right)$ is the initial spatial condition at initial time $t_0$ and $\omega^k_0$ is related to the causal character of the geodesic\footnote{The glyph `$\omega_0$' was chosen to as an allusion to the notational conventions of \cite{LesnefskyDissertation_2024}.  In \cite{LesnefskyDissertation_2024} Chapter 2.5, Thm. \ref{ledthm} was proven, and the relationship of integration constant $w_0$ to causal character was discussed.  In the current opus, Eq. \ref{eq: appendix example geodesics geodesic as t}, the integration constant $\omega_0$ is related - but not equivalent - to the role of $w_0$ in Thm. \ref{ledthm}, hence the transliteration to Greek.}.  It is trivial to solve for $x^k_0$, however the solution for $\omega_0^k$ is more involved.  For timelike geodesics one can compute:
\begin{align} \label{eq: appendix example geodesics omega 0 long derivation timelike}
		-1 & = g \left( \frac{d x^\mu}{d \lambda} \frac{\partial}{\partial x^\mu} , \frac{d x^\mu}{d \lambda} \frac{\partial}{\partial x^\mu} \right) \nonumber \\
		   & = g \left( \frac{d t}{d \lambda} \frac{\partial}{\partial t} , \frac{d t}{d \lambda} \frac{\partial}{\partial t} \right) + g \left( \frac{d x^k}{d \lambda} \frac{\partial}{\partial x^k} , \frac{d x^k}{d \lambda} \frac{\partial}{\partial x^k} \right) \nonumber \\
		   & = \left( \frac{d t}{d \lambda} \right)^2 g \left( \frac{\partial}{\partial t} , \frac{\partial}{\partial t} \right) + \left( \frac{d x^k}{d \lambda} \right)^2 g \left( \frac{\partial}{\partial x^k} , \frac{\partial}{\partial x^k} \right) \nonumber \\
		   & = - \left( \frac{d t}{d \lambda} \right)^2 + f^2 \left( t \left(\lambda\right) \right) \left(\frac{d t}{d \lambda}\right)^2 \left(\frac{d x^k}{d \lambda} \frac{d \lambda}{d t}\right)^2 g_\Sigma \left( \frac{\partial}{\partial x^k} , \frac{\partial}{\partial x^k} \right) \nonumber \\
		   & = \left( \frac{d t}{d \lambda} \right)^2 \left( -1 + f^2 \left( t\left(\lambda\right) \right) \left( \frac{d x^k}{d t} \right)^2 g_\Sigma \left(\frac{\partial}{\partial x^k} , \frac{\partial}{\partial x^k}\right) \right)
\end{align}                                                                                                  Switching to $t$ as in independent variable and setting $t_0 = t \left( \lambda_0 \right)$, one can solve for\footnote{Per the assumptions enthymeme in this solutions of the geodesic equation - namely the solution exists, is unique, and is smooth per Thm. D.1 in \cite{Lee2012} - the Inverse Function Theorem of \cite{Lee2012} Thm. C.34 can be applied to invert $\frac{d x^k}{d t}$.} $ \frac{d x^k}{d t} \left( \omega_0 \right) $:
\begin{align} \label{eq: appendix example geodesics omega 0 general GFRW calculation}
			\left. \left(\frac{d x^k}{d t}\right)^2 \left( \omega_0^k \right) \right|_{t_0} & = \left. \frac{1 - \left(\frac{d \lambda}{d t}\right)^2}{f^2 \left( t_0 \right) g_\Sigma \left( \frac{\partial}{\partial x^k} , \frac{\partial}{\partial x^k} \right)} \right|_{t_0} \nonumber \\
			& = \left. \frac{f^2 \left( t_0\right)}{\xi^2_{x^k}} \left( 1 - \left(\frac{d \lambda}{d t}\right)^2 \right) \right|_{t_0}
\end{align}
where $\xi^2_{x^k}$ is the ``useful constant'' of Eq. \ref{appendix: eq xi const} for geodesic solution $x^k$ and $\frac{d t}{d \lambda}$ is the velocity in the frame $ \left[ \frac{d t}{d \lambda} , \left\{ \frac{d x^k}{d \lambda} \right\}_{k=1}^{n-1} \right]^T$.  Unfortunately, $\frac{d t}{d \lambda}$ cannot be eliminated, so the initial conditions for $\frac{d x^k}{d t}$ are subordinate to frame $ \left[ \frac{d t}{d \lambda} , \left\{ \frac{d x^k}{d \lambda} \right\}_{k=1}^{n-1} \right]^T$.  Until further information about the problem is specified, simplification of Eq. \ref{eq: appendix example geodesics omega 0 general GFRW calculation} cannot proceed.

\subsection{Geodesics of the $+c$ Model}
The geodesics for the model of Sec. \ref{inflate} given in Eq. \ref{aplusc} will be calculated in frames $ \left[ \frac{d t}{d \lambda} ,  \frac{d x}{d \lambda} , 0 , 0 \right]^T$ and $\left[1 , \frac{d x}{d t} , 0 , 0 \right]^T$ over GFRW\footnote{A careful reader will note that the scale factor of $\mathbb{R}^1_1 \times_{\exp{\nicefrac{t}{\alpha}} + c} \mathbb{R}^{n-1}$ differs from that of Eq. \ref{aplusc}.  Utilizing this form simplifies the geodesic calculations in the warped product formalism a propos Eqs. \ref{appendix: eq: prop 7.38-t}, \ref{appendix: eq: prop 7.38-x} without changing the crux of the function behavior.} $\mathbb{R}^1_1 \times_{\exp{\nicefrac{t}{\alpha}} + c} \mathbb{R}^{n-1}$ with $n=4$.  The former admits a numeric solution but the latter admits a closed form solution!  Additionally, the frame utilized for timelike geodesics is that of
\begin{equation}
	\frac{d x^\mu}{d\lambda} \frac{\partial}{\partial x^\mu} = \left. \begin{bmatrix}
			\sqrt{2} \\
			\frac{1}{f\left( t \right)} \\
			0 \\
			0
		\end{bmatrix} \right|_{\lambda=0} = \begin{bmatrix}
			\sqrt{2} \\
			\frac{1}{2} \\
			0 \\
			0 
		\end{bmatrix}
		\label{eq:appendix example geodesics ndsolve initial condition timelike}
\end{equation}
and the frame for null geodesics is that of
\begin{equation}
	\frac{d x^\mu}{d\lambda} \frac{\partial}{\partial x^\mu} = \left. \begin{bmatrix}
			1 \\
			\frac{1}{f\left( t \right)} \\
			0 \\
			0
		\end{bmatrix} \right|_{\lambda=0} = \begin{bmatrix}
			1 \\
			\frac{1}{2} \\
			0 \\
			0 
		\end{bmatrix} \, .
		\label{eq:appendix example geodesics ndsolve initial condition null}
\end{equation}
for parameterization of $\alpha = c = 1$ and $x\left(0 \right) = t \left( 0 \right) = 0$.

\subsubsection{Numeric Geodesics of the $+c$ Model} \label{sec:appendix Num Geodesics +c}
Utilizing the \texttt{ogre.m} library of \emph{Mathematica} - see \cite{Shoshany2021_OGRe} - the geodesic equation is calculated in the usual way - not using the trick of Eq. \ref{eq: appendix example geodesics geodesic as t} - as
\begin{eqnarray}
    \label{eq:appendix example geodesics lambda geodsic equations}
			0 & = & \frac{e^{\nicefrac{t}{\alpha}}}{\alpha}  \left(c+e^{\nicefrac{t}{\alpha}}\right) \dot{x}^2+\ddot{t} \\
			0 & = & \frac{1}{\alpha}\left(\frac{2 \dot{t} \dot{x} e^{\nicefrac{t}{\alpha}}}{c+e^{\nicefrac{t}{\alpha}}}\right)+\ddot{x}
\end{eqnarray}
Utilizing the above initial conditions one finds the timelike geodesic solutions of Figs. \ref{fig:appendix example numeric geodesics +c},\ref{fig:appendix example numeric future geodesics +c},\ref{fig:appendix example numeric past geodesics +c},\ref{fig:appendix example numeric geodesics +c 1},\ref{fig:appendix example numeric future geodesics +c 1},\ref{fig:appendix example numeric past geodesics +c 1},\ref{fig:appendix example numeric parametric geodesics +c},\ref{fig:appendix example numeric future parametric geodesics +c},\ref{fig:appendix example numeric past parametric geodesics +c}.
	\begin{figure}
		\centering
		\includegraphics[width=\linewidth]{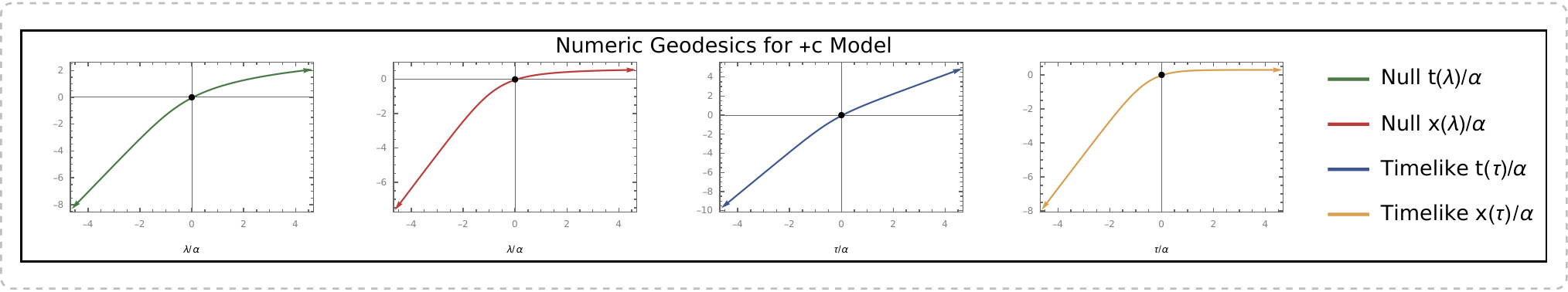}
		\caption{Plot of null and timelike geodesics of $\mathbb{R}^1_1 \times_{\exp{\nicefrac{t}{\alpha}} + c} \mathbb{R}^3$ with $\alpha = c = 1$.  Null affine parameter is given by $\lambda$ and timelike proper time is given by $\tau$.  All length units have been renormalized to be dimensionless dividing by canonical length scale $\alpha$.  Initial null velocity is $\left. \left[ 1 , \frac{1}{2} , 0 , 0 \right]^T \right|_{\lambda = 0}$ and initial timelike velocity is $\left. \left[ \sqrt{2} , \frac{1}{2} , 0 , 0 \right]^T \right|_{\tau = 0}$.  The shown plot is a cylinder over constant coordinates $y,z$, hence two spacetime dimensions have been suppressed.  Initial conditions of $x\left(\lambda_0 = 0 \right) = t \left( \lambda_0 = 0 \right) = 0 $ are shown by black point.  Future and past geodesic rays shown by arrows points towards $x,t \rightarrow + \infty $ and $ - \infty \leftarrow x,t$, respectively.  Appropriate causal arclength was calculated to diverge as $\lambda, \tau \rightarrow \infty$; thus this model is geodesically complete.  However, $H^-_{avg} = 0$ for computation over first limit asymptotic past $t \rightarrow - \infty$, $H^+_{avg} = \frac{1}{\alpha} = 1$ for first limit asymptotic future $t \rightarrow + \infty$, and $H_{avg} = \frac{1}{2}$ for a symmetric interval limit computation - see \cite{Lesnefsky:2022fen} for discussion of $H_{avg}$ limit order concerns. }
		\label{fig:appendix example numeric geodesics +c}
	\end{figure}
    \begin{figure}
		\centering
		\includegraphics[width=\linewidth]{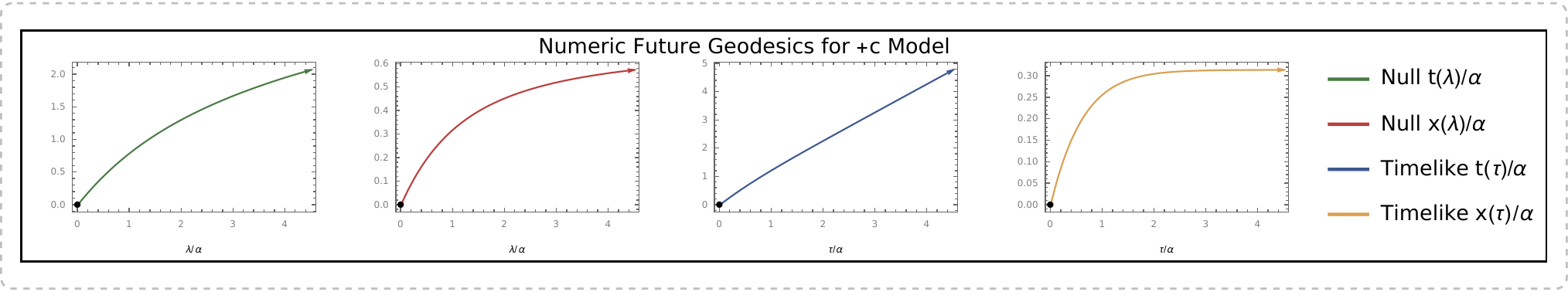}
		\caption{An additional plot highlighting future geodesic rays of the model and computations of Fig. \ref{fig:appendix example numeric geodesics +c}.}
		\label{fig:appendix example numeric future geodesics +c}
	\end{figure}
    \begin{figure}
		\centering
		\includegraphics[width=\linewidth]{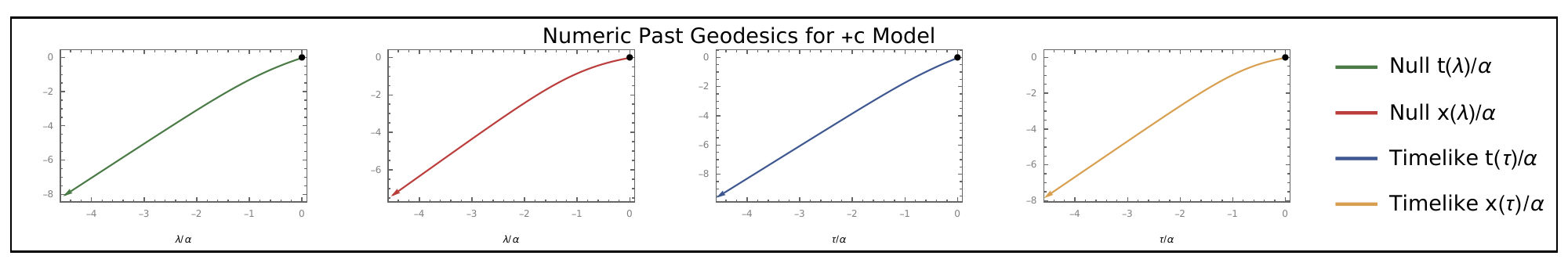}
		\caption{An additional plot highlighting past geodesic rays of the model and computations of Fig. \ref{fig:appendix example numeric geodesics +c}.}
		\label{fig:appendix example numeric past geodesics +c}
	\end{figure}
    \begin{figure}
		\centering
		\includegraphics[width=\linewidth]{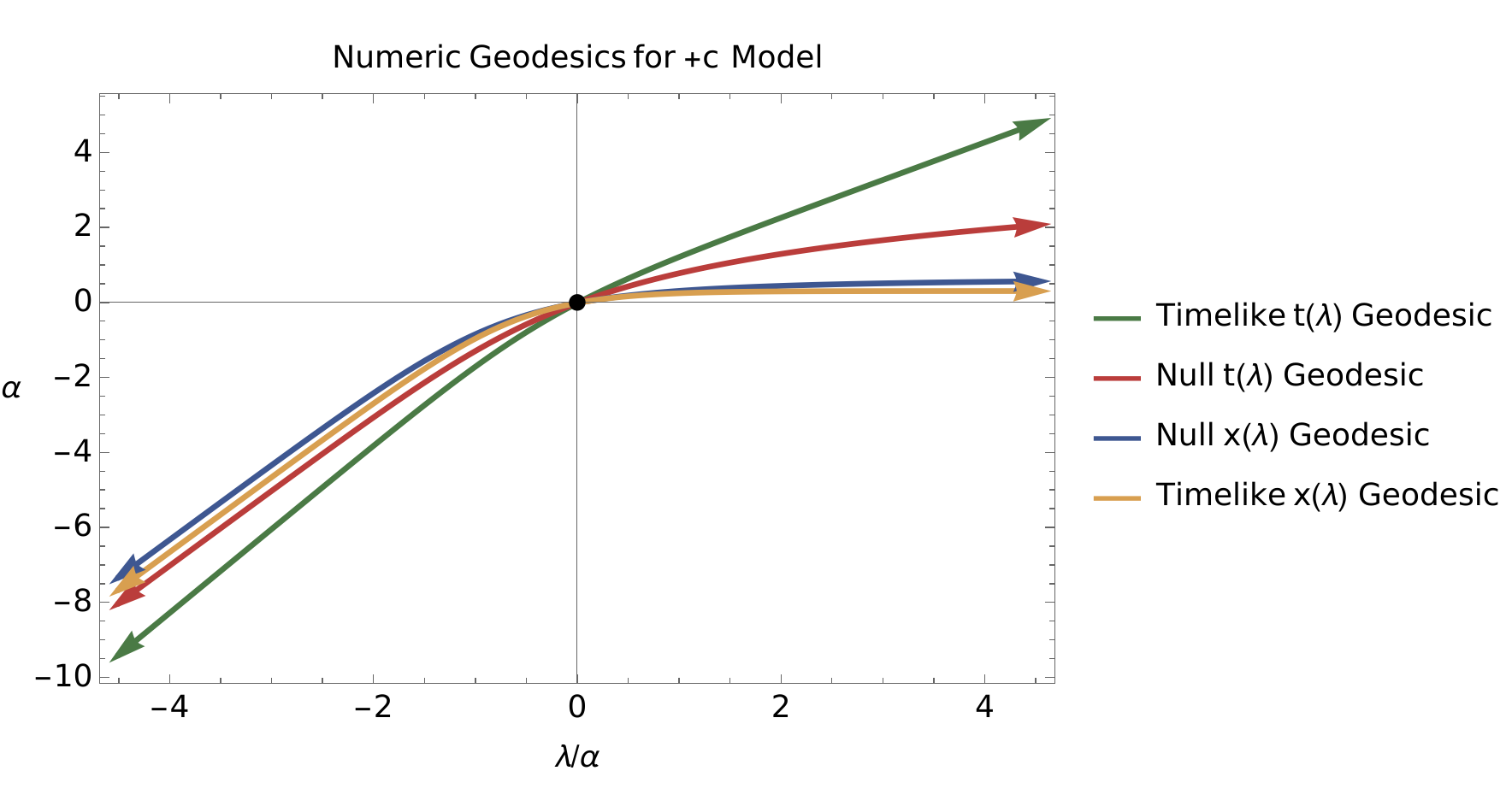}
		\caption{An additional plot of null and timelike geodesics utilizing the model and computations of Fig. \ref{fig:appendix example numeric geodesics +c}.}
		\label{fig:appendix example numeric geodesics +c 1}
	\end{figure}
    \begin{figure}
		\centering
		\includegraphics[width=\linewidth]{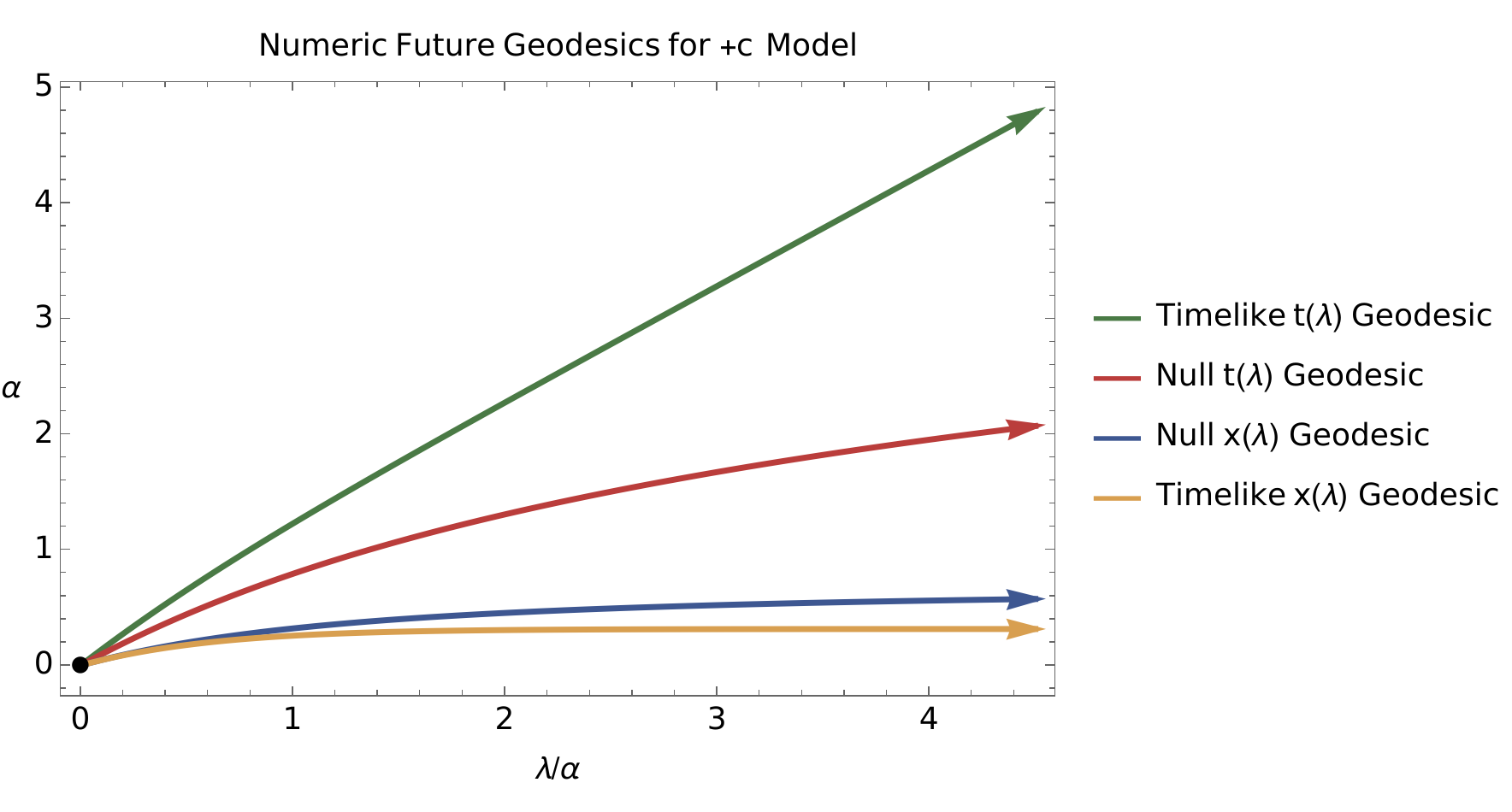}
		\caption{An additional plot of future geodesic rays utilizing the model and computations of Fig. \ref{fig:appendix example numeric geodesics +c}.}
		\label{fig:appendix example numeric future geodesics +c 1}
	\end{figure}
    \begin{figure}
		\centering
		\includegraphics[width=\linewidth]{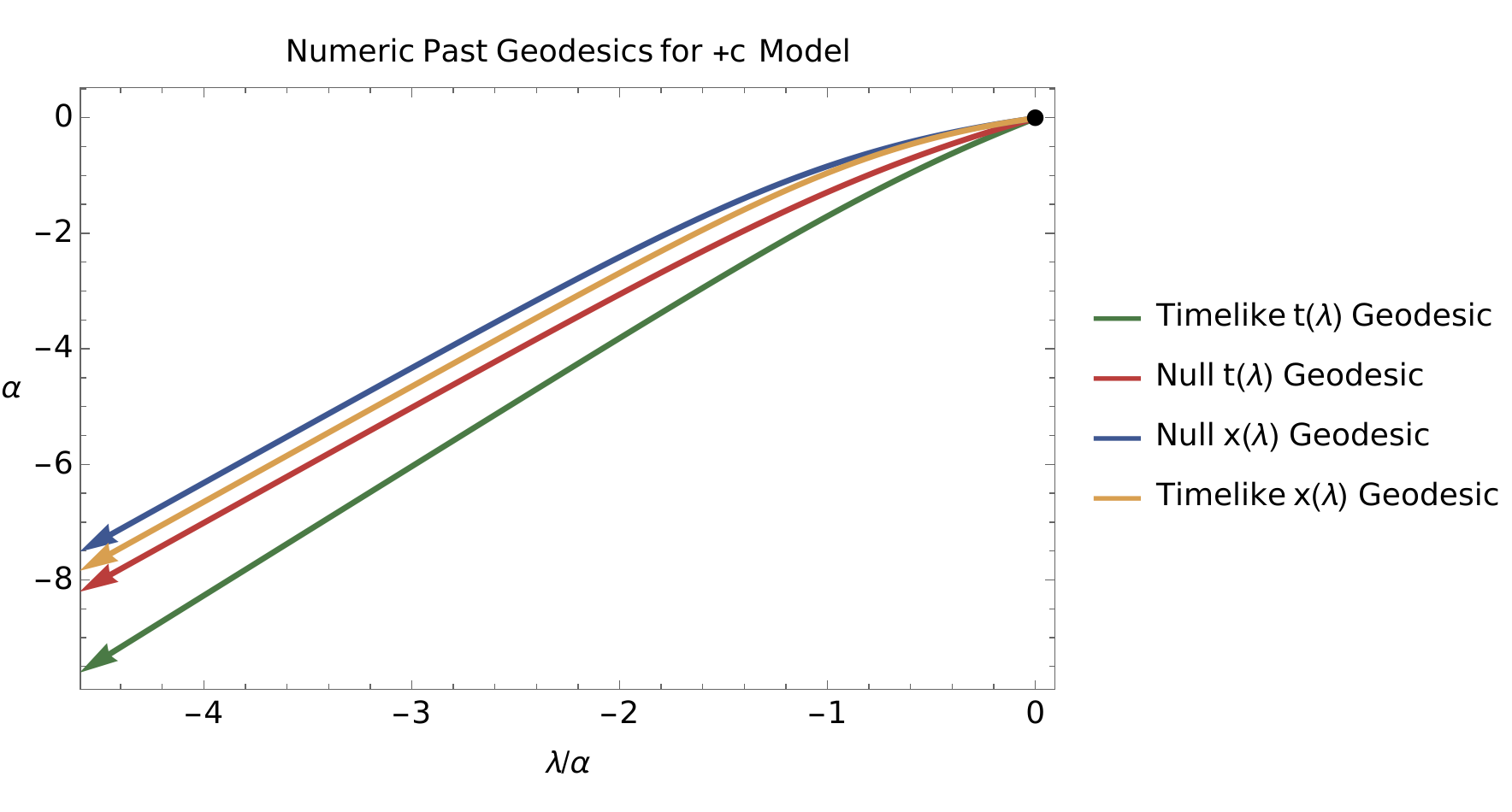}
		\caption{An additional plot of past geodesic rays utilizing the model and computations of Fig. \ref{fig:appendix example numeric geodesics +c}.}
		\label{fig:appendix example numeric past geodesics +c 1}
	\end{figure}
	\begin{figure}
		\centering
		\includegraphics[width=\linewidth]{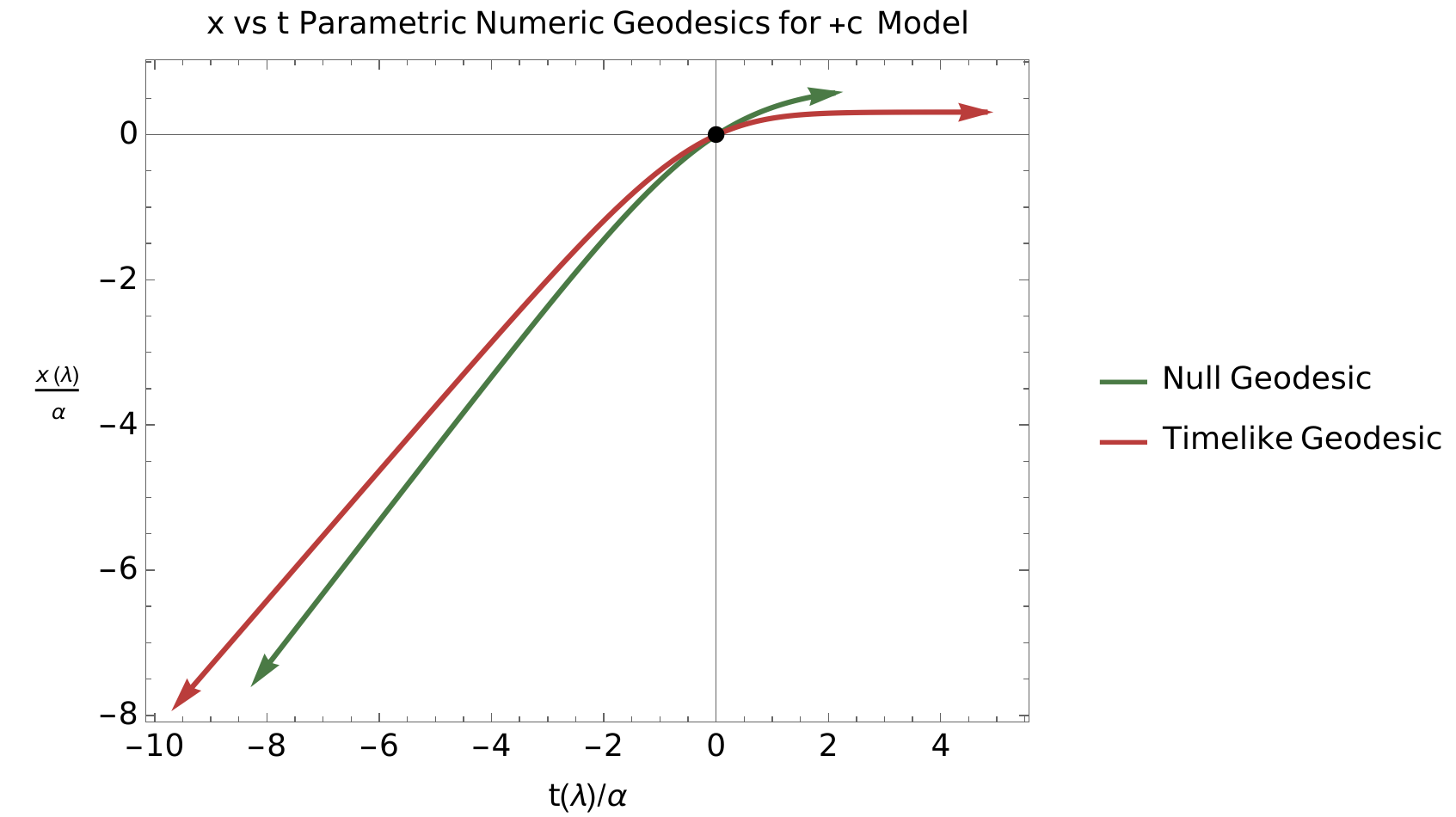}
		\caption{Parametric plot $x$ vs $t$ of null and timelike geodesics of $\mathbb{R}^1_1 \times_{\exp{\nicefrac{t}{\alpha}} + c} \mathbb{R}^3$ with $\alpha = c = 1$.  Null affine parameter is given by $\lambda$ and timelike proper time is given by $\tau$.  All length units have been renormalized to be dimensionless dividing by canonical length scale $\alpha$.  The appropriate affine parameter was evolved for $\nicefrac{\lambda}{\alpha}, \nicefrac{\tau}{\alpha} \in \left[ -\frac{9}{2} , \frac{9}{2} \right]$.  Initial null velocity is $\left. \left[ 1 , \frac{1}{2} , 0 , 0 \right]^T \right|_{\lambda = 0}$ and initial timelike velocity is $\left. \left[ \sqrt{2} , \frac{1}{2} , 0 , 0 \right]^T \right|_{\tau = 0}$.  The shown plot is a cylinder over constant coordinates $y,z$, hence two spacetime dimensions have been suppressed.  Initial conditions of $x\left(\lambda_0 = 0 \right) = t \left( \lambda_0 = 0 \right) = 0 $ are shown by black point.  Future and past geodesic rays shown by arrows points towards $x,t \rightarrow + \infty $ and $ - \infty \leftarrow x,t$, respectively.  Appropriate causal arclength was calculated to diverge as $\lambda, \tau \rightarrow \infty$; thus this model is geodesically complete.  However, $H^-_{avg} = 0$ for computation over first limit asymptotic past $t \rightarrow - \infty$, $H^+_{avg} = \frac{1}{\alpha} = 1$ for first limit asymptotic future $t \rightarrow + \infty$, and $H_{avg} = \frac{1}{2}$ for a symmetric interval limit computation - see \cite{Lesnefsky:2022fen} for discussion of $H_{avg}$ limit order concerns. }
		\label{fig:appendix example numeric parametric geodesics +c}
	\end{figure}
    \begin{figure}
		\centering
		\includegraphics[width=\linewidth]{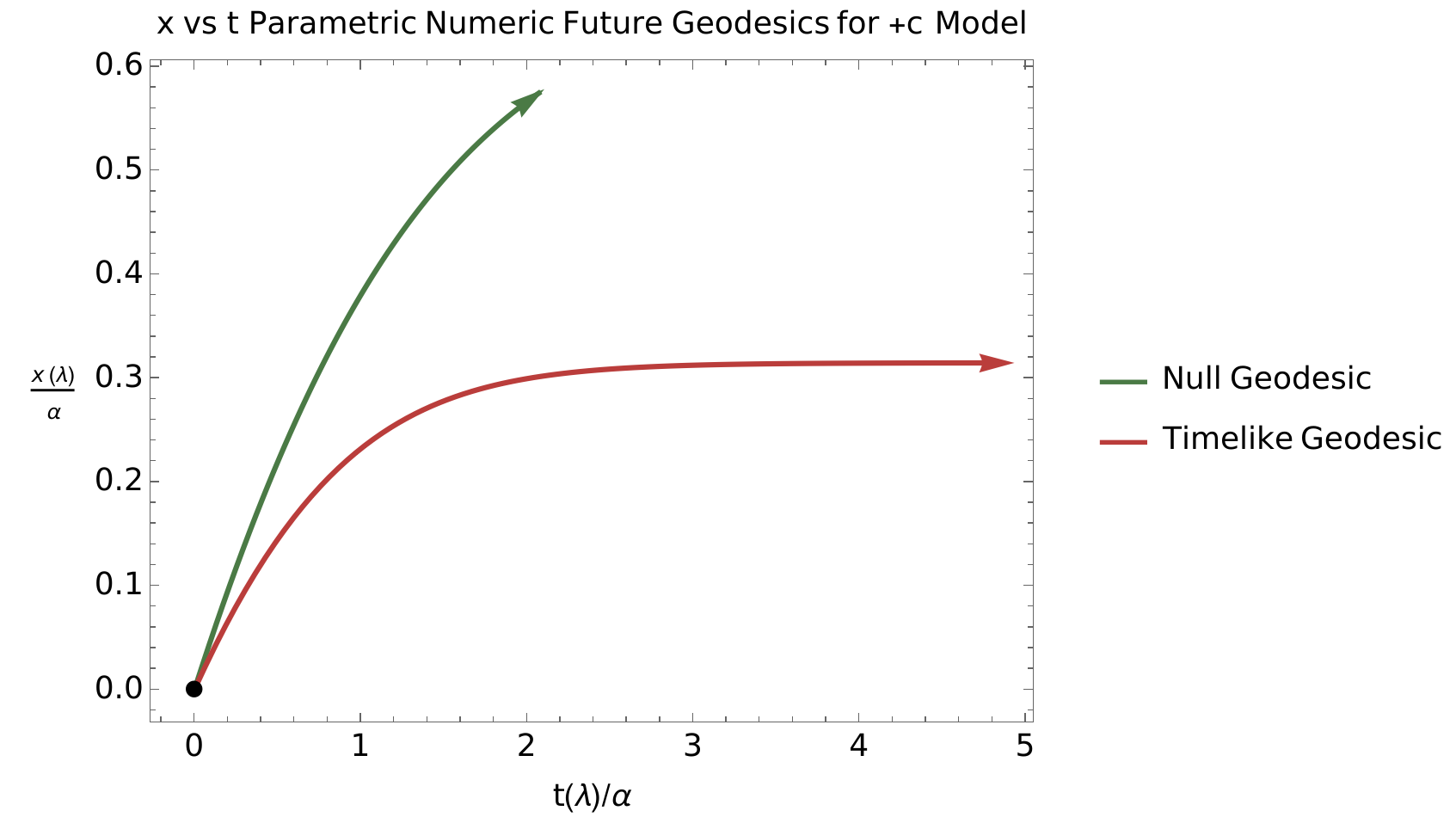}
		\caption{Additional parametric plots $x$ vs $t$ of future null and timelike geodesics of the model and calculations Fig. \ref{fig:appendix example numeric parametric geodesics +c} with $\nicefrac{\lambda}{\alpha}, \nicefrac{\tau}{\alpha} \in \left[ 0 , \frac{9}{2} \right]$. }
		\label{fig:appendix example numeric future parametric geodesics +c}
	\end{figure}
    \begin{figure}
		\centering
		\includegraphics[width=\linewidth]{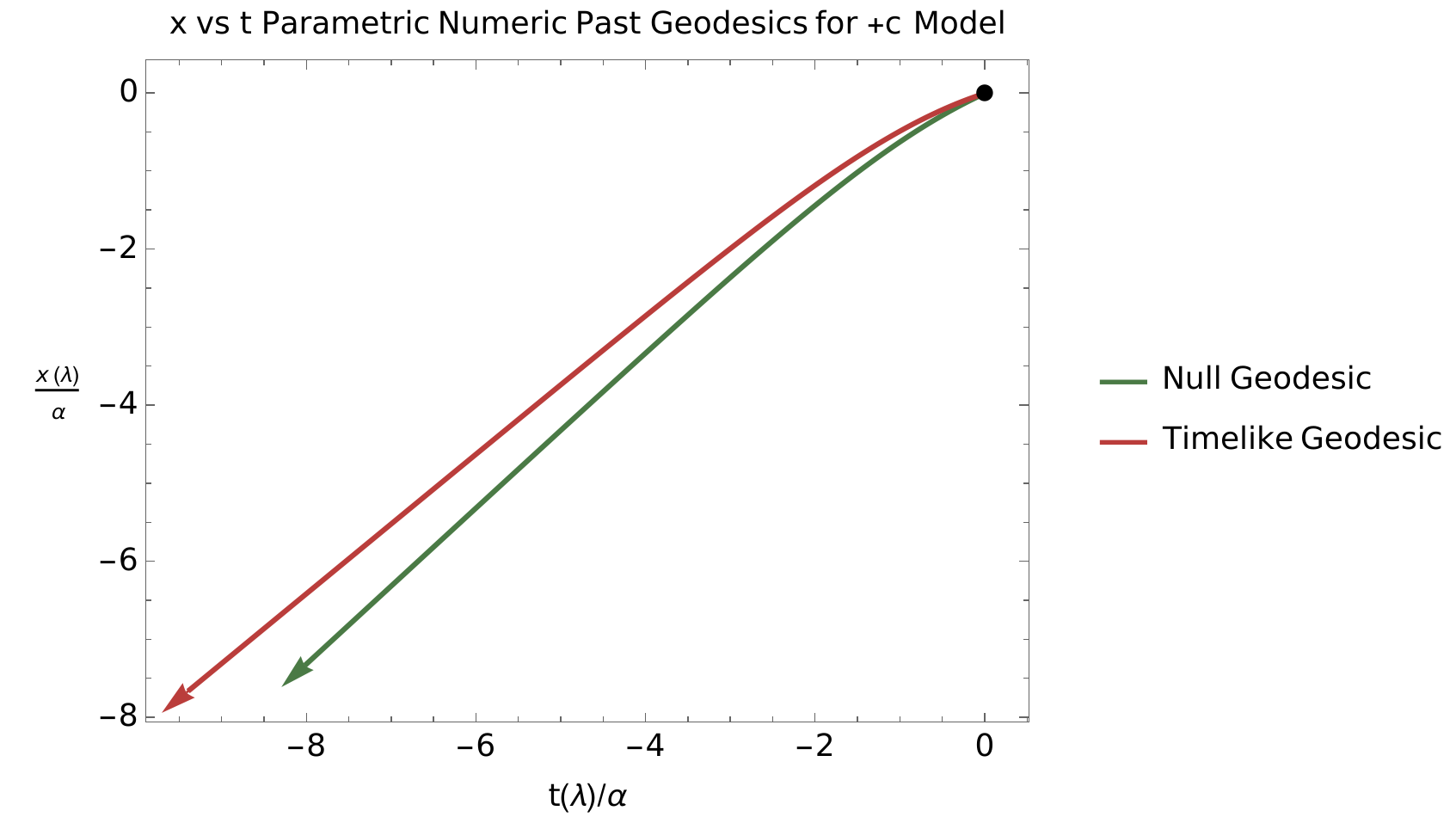}
		\caption{Additional parametric plots $x$ vs $t$ of future null and timelike geodesics of the model and calculations Fig. \ref{fig:appendix example numeric parametric geodesics +c} with $\nicefrac{\lambda}{\alpha}, \nicefrac{\tau}{\alpha} \in \left[ -\frac{9}{2} , 0 \right]$. }
		\label{fig:appendix example numeric past parametric geodesics +c}
	\end{figure}
    
A cursory examination shows that the aforementioned geodesic rays are defined for all $\mathbb{R}^+$ with diverging timelike length; thus they have been verified to be complete corroborating the results of Thm. \ref{ledthm}.

\subsubsection{Analytic Geodesics of the $+c$ Model} \label{sec: appendix analytic geodesic +c}
Utilizing the calculation of Eq. \ref{eq: appendix example geodesics geodesic as t} one derives geodesic equations of
	\begin{equation} \label{eq:appendix t frame geodesic eqn}
       0 =  x''(t)-\frac{1}{\alpha}\frac{e^{\nicefrac{t}{\alpha} } x'(t) \left(\left(c+e^{\nicefrac{t}{\alpha}}\right)^2 x'(t)^2-2\right)}{c+e^{\nicefrac{t}{\alpha}}}
	\end{equation}
WLOG initial conditions of $y \left( t \right) = z \left( t \right) = 0$ are assumed to simplify the computation.  Using anti-derivative methods one finds dimensionless
	\begin{align} \label{eq:appendix geodesic soltn closed form raw}	
		\frac{x(\nicefrac{t}{\alpha})}{\alpha} & = \frac{x_0}{\alpha} + \left( \arctanh\left(\sqrt{1 + \omega_0 \left( c + e^{\nicefrac{t}{\alpha}}\right)^2}\right)-\frac{\arctanh\left(\frac{c^2 \omega_0 + c \omega_0 e^{\nicefrac{t}{\alpha}} + 1}{\sqrt{1+c^2 \omega_0} \sqrt{c^2 \omega_0 + 2 c \omega_0 e^{\nicefrac{t}{\alpha}} + \omega_0 e^{2 \nicefrac{t}{\alpha}} + 1}}\right)}{\sqrt{1+c^2 \omega_0}} \right)
		 \nonumber \\
			 &  \qquad \qquad \times \left( \frac{ \left(c+e^{\nicefrac{t}{\alpha}}\right) \sqrt{c^2 \omega_0 + 2 c \omega_0 e^{\nicefrac{t}{\alpha}} + \omega_0 e^{2 \nicefrac{t}{\alpha}} + 1} }{ c \sqrt{\left(c+e^{\nicefrac{t}{\alpha}}\right)^2 \left(c^2 \omega_0 + 2 c \omega_0 e^{\nicefrac{t}{\alpha}} + \omega_0 e^{2 \nicefrac{t}{\alpha}} + 1\right)} } \right)
	\end{align}
where $x_0$ is the first integration constant solving for initial spatial position and $\omega_0$ is the second integration constant relating to initial velocity and causal character.  In particular, we refer the reader to \cite{LesnefskyDissertation_2024} Section 2.5 where we solve the warped product geodesic equation projection of Eq. 2.14 in the aforementioned opus.  The result is Eqs. 2.22, 2.23 in said opus, where, being forced to utilize various ODE tricks to solve the equation, one is forced to adopt a particular frame.  In this frame the constant $w_0$ in the aforementioned equations of \cite{LesnefskyDissertation_2024} encapsulate the causal character, with $w_0 \in \left\{ -1 , 0 , +1 \right\}$ for timelike, null, spacelike geodesics respectively.  Note that in this prescription $w_0$ is constant for all time.  In this case, we will show that integration constant $\omega_0$ can be related to causal character, although not quite as simple as $w_0$ discussed above.

From the heuristic of solving an ODE, one sees an integration constant of $\omega_0$ appearing in Eq. \ref{eq:appendix geodesic soltn closed form raw}: one must solve for this and hopefully divine a physical interpretation.  From this solution one calculates
	\begin{equation}
		\frac{d x^k}{d \lambda} \frac{\partial}{\partial x^k} = \pm \frac{1}{\sqrt{f \left( t \right)\left( 1 + \omega_0 f^2 \left( t \right) \right)}} \frac{\partial}{\partial x^k}
	\end{equation}
which looks suspiciously like Eq. 2.21 in \cite{LesnefskyDissertation_2024}!  Does $\omega_0$ relate to causal character?    We will show it does.  One computes
	\begin{equation} \label{eq:appendix anal geodesic velocity}
		g \left( \frac{d x}{d t} , \frac{d x}{d t} \right) = f^2 \left( t \right) \left( \frac{d x}{d t} \right)^2 = \frac{1}{1 + f^2 \left( t \right) \omega_0}
	\end{equation}
With the above equation, one can continue simplifying Eq. \ref{eq: appendix example geodesics omega 0 general GFRW calculation}.  One then reaps
    \begin{align} \label{eq:appendix omega 0 particular solution calculation}
\left.\frac{1}{f^2 \left( t \left(\lambda\right)\right)\left(1 + \omega_0 f^2 \left( t \left(\lambda\right)\right)\right)}\right|_{t_0} & = \left.\frac{f^2 \left( t \left(\lambda\right)\right)}{\xi^2_{x^k}} \left( 1 - \left(\frac{d \lambda}{d t}\right)^2 \right)\right|_{t_0} \nonumber \\
    	1 + \omega_0 f^2 \left( t_0\right) & = \left.\frac{\xi^2_{x^k}}{f^4 \left( t \left(\lambda\right)\right)}\frac{1}{1 - \left(\frac{d \lambda}{d t}\right)^2}\right|_{t_0} \nonumber \\
    	\omega_0 & = \left.\frac{g_\Sigma \left( \frac{\partial}{\partial x^k} , \frac{\partial}{\partial x^k} \right)}{f^2 \left( t \left(\lambda\right)\right)} \left( \frac{1}{1 - \left(\frac{d \lambda}{d t}\right)^2} - 1 \right)\right|_{t_0} \nonumber \\
    	\omega_0 & = \left.\frac{1}{f^2 \left( t \left(\lambda\right)\right)} \left( \frac{1}{1 - \left(\frac{d \lambda}{d t}\right)^2} - 1 \right)\right|_{t_0}
    \end{align}
Finally utilizing the frame of Eq. \ref{eq:appendix example geodesics ndsolve initial condition timelike} the timelike geodesics parameterized as $x\left( t\right)$ can be computed; the results are shown in Figs. \ref{fig:appendix example anal geodesics +c},\ref{fig:appendix example anal future geodesics +c},\ref{fig:appendix example anal past geodesics +c},\ref{fig:appendix example anal lightcone geodesics +c}.
	\begin{figure}
		\centering
		\includegraphics[width=\linewidth]{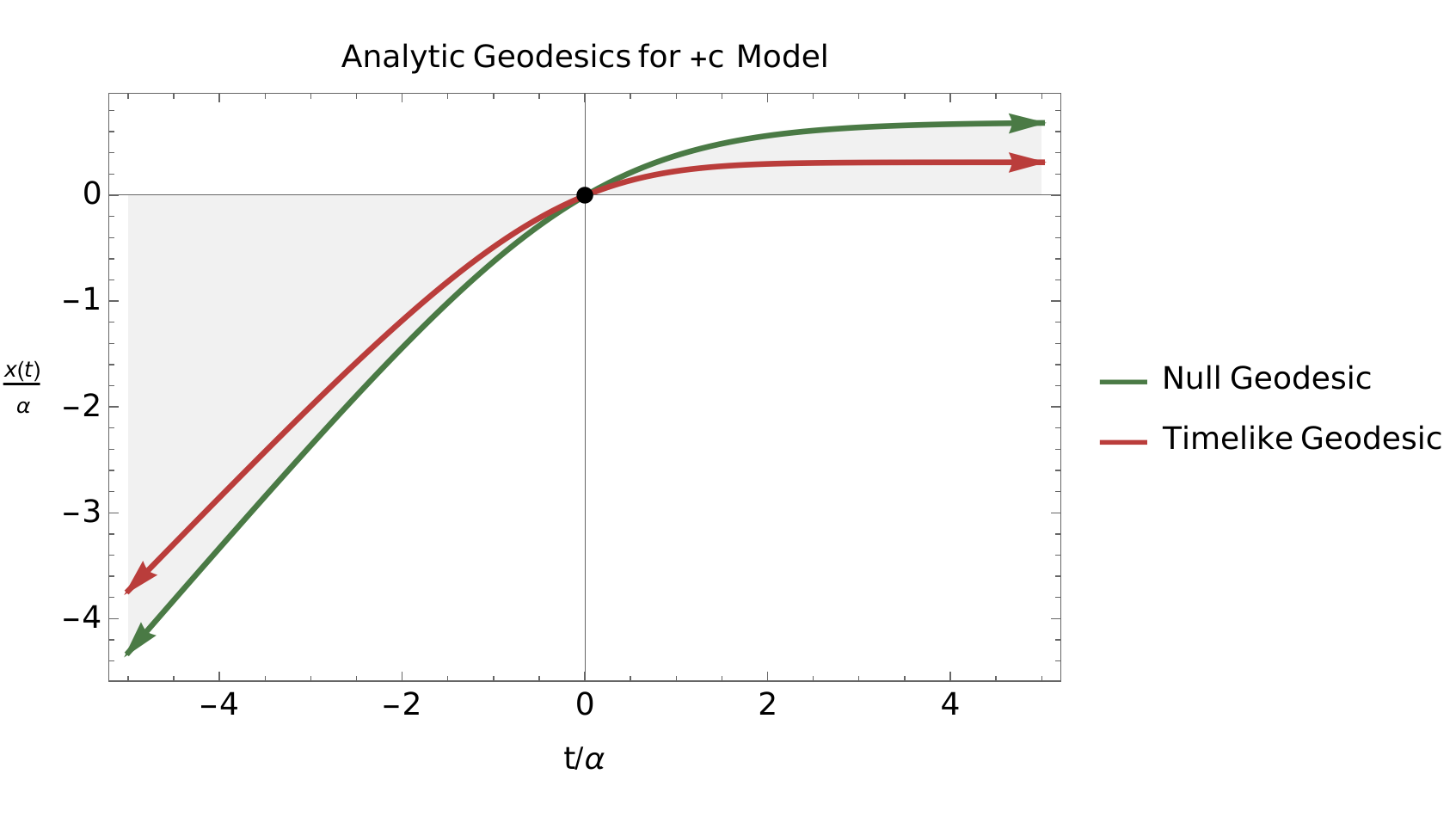}
		\caption{Plot $x\left( t\right)$ vs $t$ of null and timelike geodesics of $\mathbb{R}^1_1 \times_{\exp{\nicefrac{t}{\alpha}} + c} \mathbb{R}^3$ with $\alpha = c = 1$.  Independent variable is the time coordinate.  All length units have been renormalized to be dimensionless dividing by canonical length scale $\alpha$.  Initial null velocity is derived from frame $\left. \left[ 1 , \frac{1}{2} , 0 , 0 \right]^T \right|_{t = 0}$ and initial timelike velocity is derived from frame $\left. \left[ \sqrt{2} , \frac{1}{2} , 0 , 0 \right]^T \right|_{t = 0}$ a propos Eq. \ref{eq:appendix omega 0 particular solution calculation}.  The shown plot is a cylinder over constant coordinates $y,z$, hence two spacetime dimensions have been suppressed.  Partial light cones can be seen as the area under the null curve shaded light gray.  Appropriate causal arclength was calculated to diverge as $t \rightarrow +\infty$; thus this model is geodesically complete.  However, $H^-_{avg} = 0$ for computation over first limit asymptotic past $t \rightarrow - \infty$, $H^+_{avg} = \frac{1}{\alpha} = 1$ for first limit asymptotic future $t \rightarrow + \infty$, and $H_{avg} = \frac{1}{2}$ for a symmetric interval limit computation - see \cite{Lesnefsky:2022fen} for discussion of $H_{avg}$ limit order concerns. }
		\label{fig:appendix example anal geodesics +c}
	\end{figure} 
    \begin{figure}
		\centering
		\includegraphics[width=\linewidth]{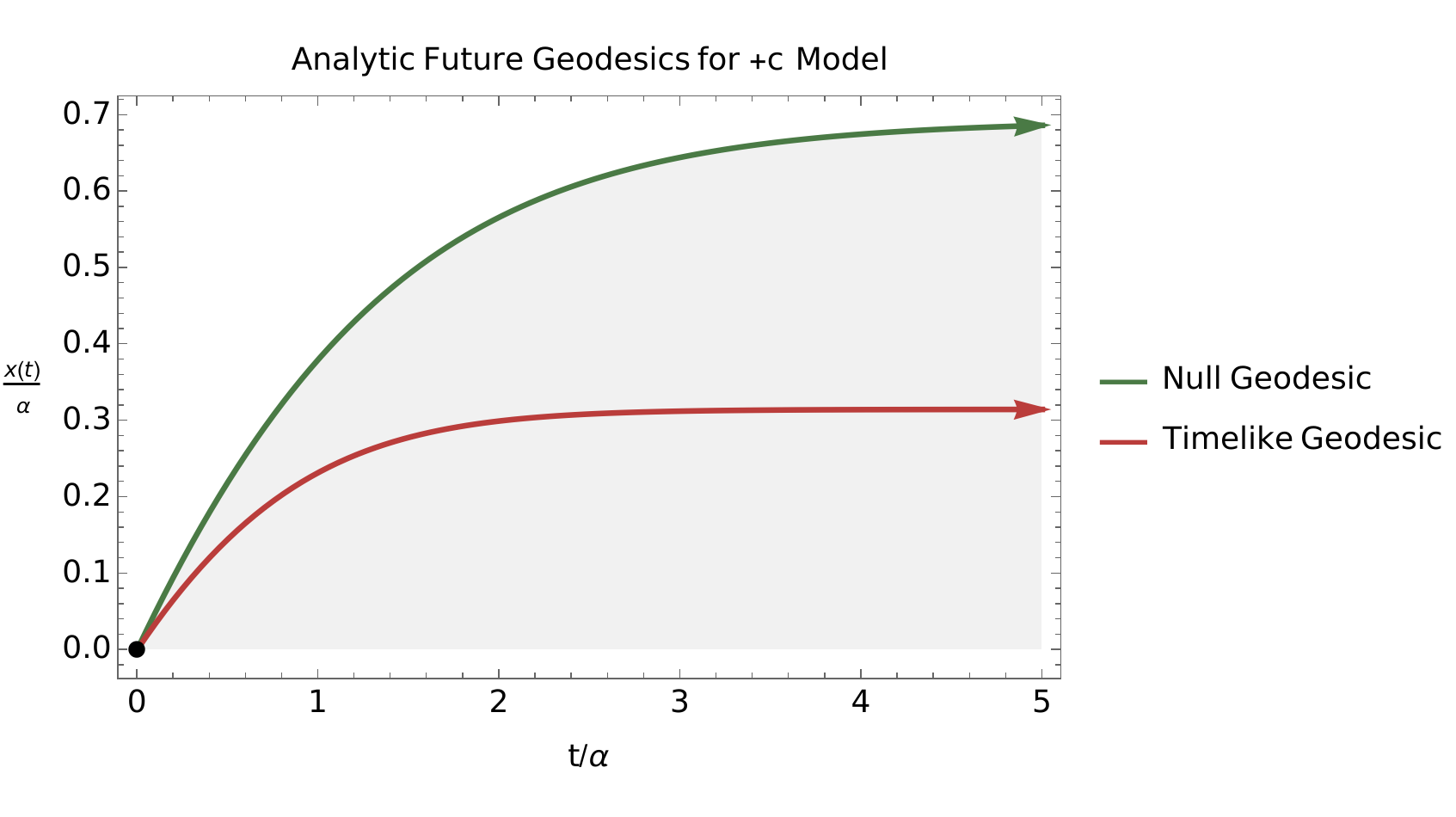}
		\caption{Additional future geodesic rays of the model and calculations of Fig. \ref{fig:appendix example anal geodesics +c}.}
		\label{fig:appendix example anal future geodesics +c}
	\end{figure} 
    \begin{figure}
		\centering
		\includegraphics[width=\linewidth]{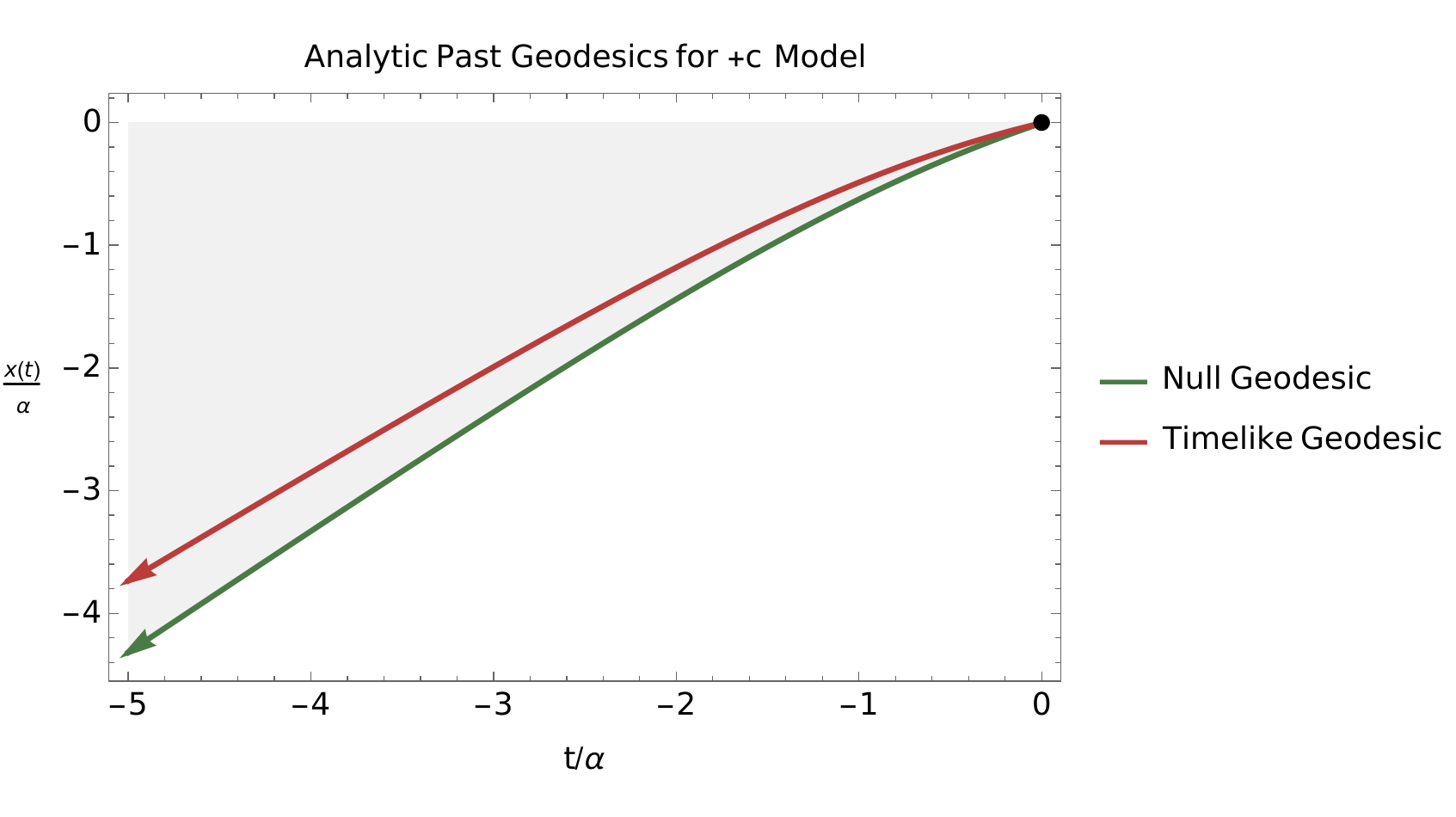}
		\caption{Additional past geodesic rays of the model and calculations of Fig. \ref{fig:appendix example anal geodesics +c}.}
		\label{fig:appendix example anal past geodesics +c}
	\end{figure} 
    \begin{figure}
		\centering
		\includegraphics[width=\linewidth]{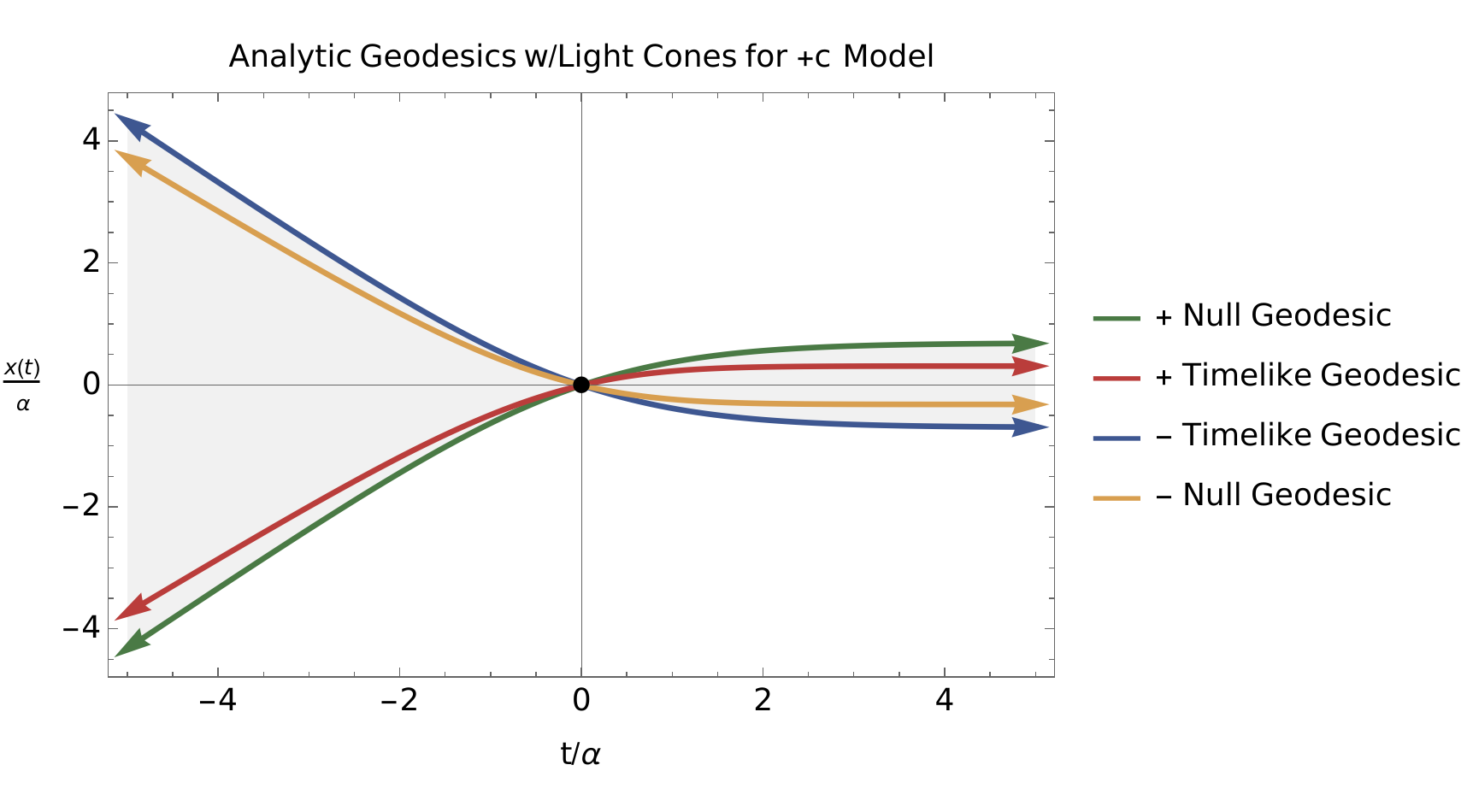}
		\caption{Additional geodesic plots of the model and calculations of Fig. \ref{fig:appendix example anal geodesics +c}.  Additional null and timelike geodesics were plotted as to show the light cone cylinder $\Lambda_{\left( 0 , 0,y,z \right)} \mathbb{R}^1_1 \times_{\exp{\nicefrac{t}{\alpha}} + c} \mathbb{R}^3$ over cylinder $t=0 , x=0 , y , z$.}
		\label{fig:appendix example anal lightcone geodesics +c}
	\end{figure}
    
Analytic null geodesics are most easily calculated from the length calculation of $g\left( \dot{\gamma}, \dot{\gamma} \right) = 0 = -dt^2 + f^2 \left( t \right) \overline{g}\left(  \dot{\gamma}, \dot{\gamma} \right)$ where $\dot{\gamma}$ is a null geodesic of $\mathbb{R}^1_1 \times_{\exp{\nicefrac{t}{\alpha}} + c} \mathbb{R}^3$ and $\overline{g}$ is the (flat) metric of $\mathbb{R}^3$.  From this calculation one reaps
\begin{equation} \label{eq:appendix anal null geodesic +c model}
    \dot{x}\left( t\right) = \frac{1}{f\left( t\right)} = \frac{1}{\exp{\nicefrac{t}{\alpha}} + c}
\end{equation}
with solution
\begin{equation} \label{eq:appendix anal null geodesic +c model soltn}
    \frac{x \left( \nicefrac{t}{\alpha}  \right)}{\alpha} = \frac{\chi_0}{\alpha} + \frac{1}{c}\frac{t}{\alpha} - \frac{1}{c}\left( \ln{c} - \ln{\alpha} + \ln{\left(\exp{\nicefrac{t}{\alpha}} + c\right)} \right)
\end{equation}
with integration constant $\chi_0 = x_0 - \frac{t_0}{c} + \frac{\alpha}{c}\left( \ln{c} - \ln{\alpha} + \ln{\left(\exp{\nicefrac{t_0}{\alpha}} + c\right)} \right)$ where $x_0$ is the initial $x$ coordinate.  Applied to the model at hand with $\alpha=c=1$, the results are plotted in Fig. \ref{fig:appendix example anal geodesics +c}.

\subsection{Geodesics of the Quadratic Model}
Geodesics of Eq. \ref{apoly} will now be plotted.  To emphasize the polynomial nature of the scale factor we adopt the standard mathematical nomenclature
\begin{equation} \label{eq:appendix quadratic scale factor}
    f \left( t \right) = a_2 t^2 + a_1 t + a_0 \,.
\end{equation}
The GFRW used for calculation is that of $\mathbb{R}^1_1 \times_{a_2 t^2 + a_1 t + a_0} \mathbb{R}^3$.

\subsubsection{Numeric Geodesics of the Quadratic Model}
In a similar vein to the calculations of Sec. \ref{sec:appendix Num Geodesics +c}, the geodesics of the quadratic model will be calculated numerically utilizing \texttt{ogre.m} in Mathematica.  The geodesic equation gives
\begin{eqnarray} \label{eq:appendix quadratic model numeric geodesic equation}
 0 & = & \left(2 a_2 t+a_1\right) \left(t \left(a_2 t+a_1\right)+a_0\right) \left(\dot{x}^2+\dot{y}^2+\dot{z}^2\right)+\ddot{t} \\
 0 & = & \frac{2 \dot{t} \dot{x} \left(2 a_2 t+a_1\right)}{t \left(a_2 t+a_1\right)+a_0}+\ddot{x} \\
 0 & = & \frac{2 \dot{t} \dot{y} \left(2 a_2 t+a_1\right)}{t \left(a_2 t+a_1\right)+a_0}+\ddot{y} \\
 0 & = & \frac{2 \dot{t} \dot{z} \left(2 a_2 t+a_1\right)}{t \left(a_2 t+a_1\right)+a_0}+\ddot{z}
\end{eqnarray}
Again using the initial conditions for the frames of Eqs. \ref{eq:appendix example geodesics ndsolve initial condition timelike},\ref{eq:appendix example geodesics ndsolve initial condition null} one reaps the solutions\footnote{In Sec. \ref{sec:appendix Num Geodesics +c} there was only a single choice of parameter to renormalize length: that of $\alpha$.  However, in a polynomial, there are multiple parameters with units, namely $\left\{ a_\ell \right\}_{\ell=1}^n$, with $a_\ell$ having units $\mathrm{length}^{-\ell}$.  The leading power in a polynomial contributes most to canonical scale: we chose to renormalize by $\sqrt[n]{a_n}$ which has units of $\mathrm{length}^{-1}$.} plotted in Figs. \ref{fig:appendix example numeric geodesics quadratic}, \ref{fig:appendix example numeric future geodesics quadratic}, \ref{fig:appendix example numeric past geodesics quadratic}, \ref{fig:appendix example numeric geodesics quadratic1}, \ref{fig:appendix example numeric future geodesics quadratic1}, \ref{fig:appendix example numeric past geodesics quadratic1}, \ref{fig:appendix example numeric parametric geodesics quadratic}, \ref{fig:appendix example numeric parametric future geodesics quadratic}, \ref{fig:appendix example numeric parametric past geodesics quadratic}.
	\begin{figure}
		\centering
		\includegraphics[width=\linewidth]{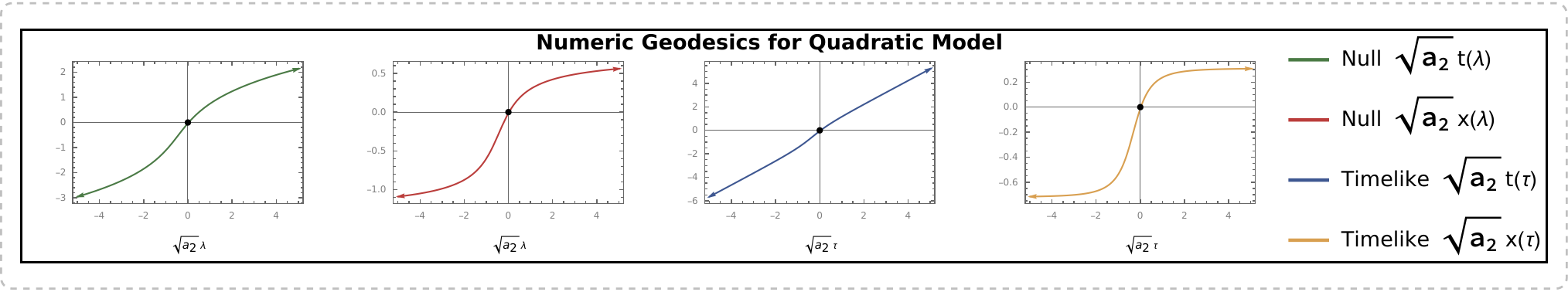}
		\caption{Plot of null and timelike geodesics of $\mathbb{R}^1_1 \times_{a_2 t^2 + a_1 t + a_0} \mathbb{R}^3$ with $a_2 = a_1 = 1$ and $a_0 = 2$.  Null affine parameter is given by $\lambda$ and timelike proper time is given by $\tau$.  All length units have been renormalized to be dimensionless multiplying by canonical length scale $\sqrt{a_2}$.  Initial null velocity is $\left. \left[ 1 , \frac{1}{2} , 0 , 0 \right]^T \right|_{\lambda = 0}$ and initial timelike velocity is $\left. \left[ \sqrt{2} , \frac{1}{2} , 0 , 0 \right]^T \right|_{\tau = 0}$.  The shown plot is a cylinder over constant coordinates $y,z$, hence two spacetime dimensions have been suppressed.  Appropriate causal arclength was calculated to diverge as $\lambda, \tau \rightarrow \infty$; thus this model is geodesically complete.  However, $H^-_{avg} = 0$ for computation over first limit asymptotic past $t \rightarrow - \infty$, $H^+_{avg} = 0$ for first limit asymptotic future $t \rightarrow + \infty$, and $H_{avg} > 0$ for a symmetric compact interval limit computations - see \cite{Lesnefsky:2022fen} for discussion of $H_{avg}$ limit order concerns. 
		\label{fig:appendix example numeric geodesics quadratic}.  For the interval $\left[ -10 , 10 \right]$, one computes $H_{avg}^{\left[ -10 , 10 \right]} = \frac{\ln 28 - \ln 23}{20} \simeq 0.009835$.}
	\end{figure}
    \begin{figure}
		\centering
		\includegraphics[width=\linewidth]{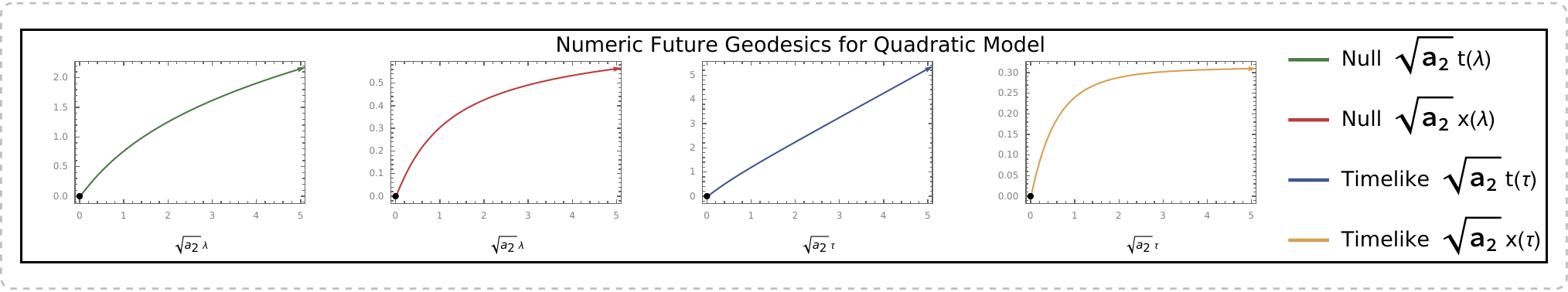}
		\caption{An additional plot highlighting future geodesic rays of the model and computations of Fig. \ref{fig:appendix example numeric geodesics quadratic}.}
		\label{fig:appendix example numeric future geodesics quadratic}
	\end{figure}
    \begin{figure}
		\centering
		\includegraphics[width=\linewidth]{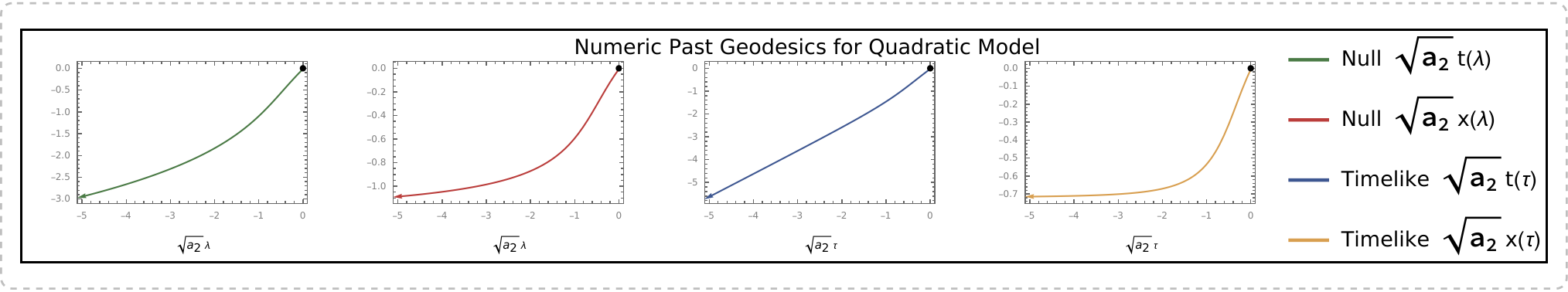}
		\caption{An additional plot highlighting past geodesic rays of the model and computations of Fig. \ref{fig:appendix example numeric geodesics quadratic}.}
		\label{fig:appendix example numeric past geodesics quadratic}
	\end{figure}
	\begin{figure}
		\centering
		\includegraphics[width=\linewidth]{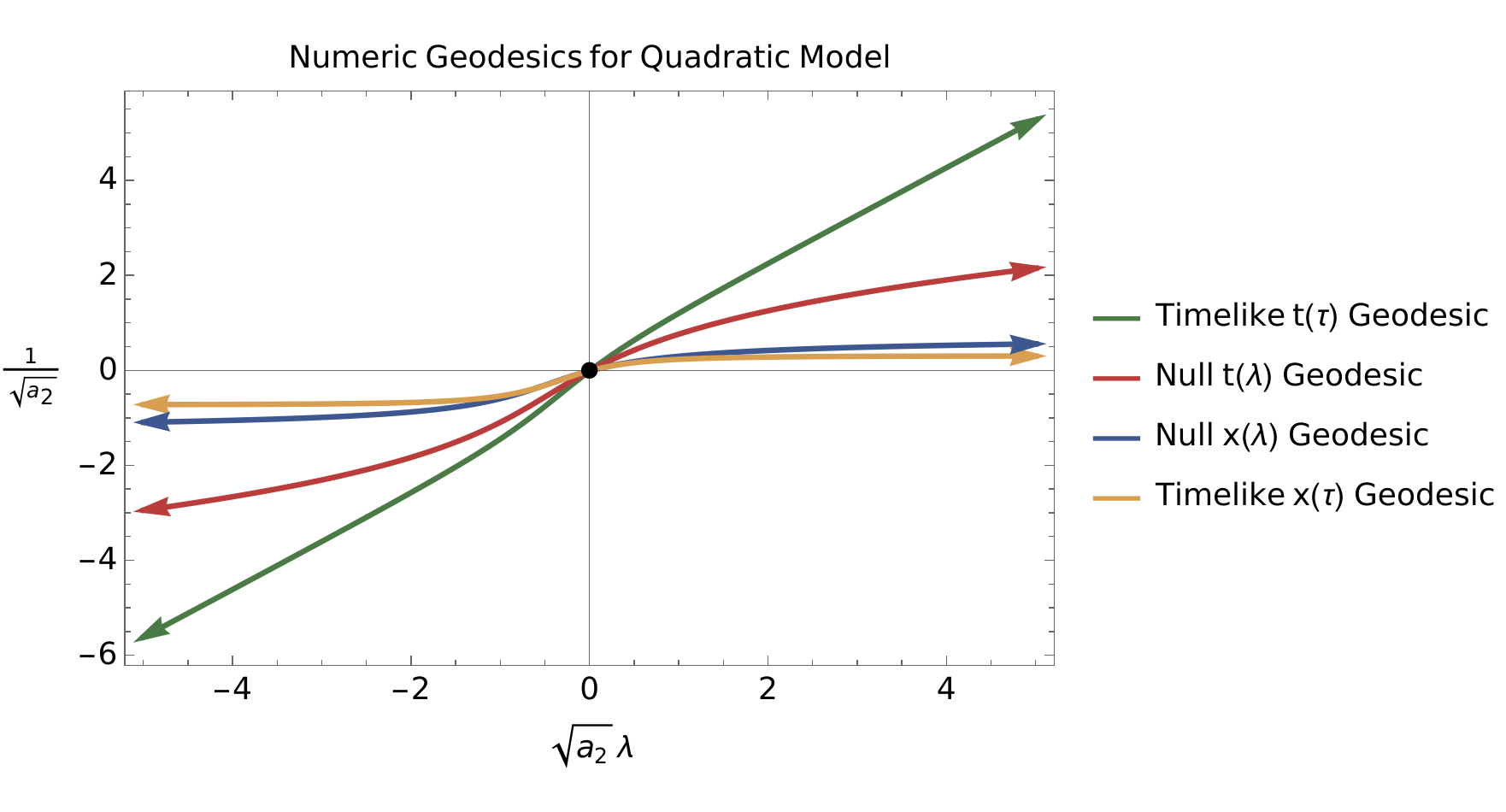}
		\caption{An additional plot of null and timelike geodesics of the model and computations of Fig. \ref{fig:appendix example numeric geodesics quadratic}.}
        \label{fig:appendix example numeric geodesics quadratic1}
	\end{figure}	
    \begin{figure}
		\centering
		\includegraphics[width=\linewidth]{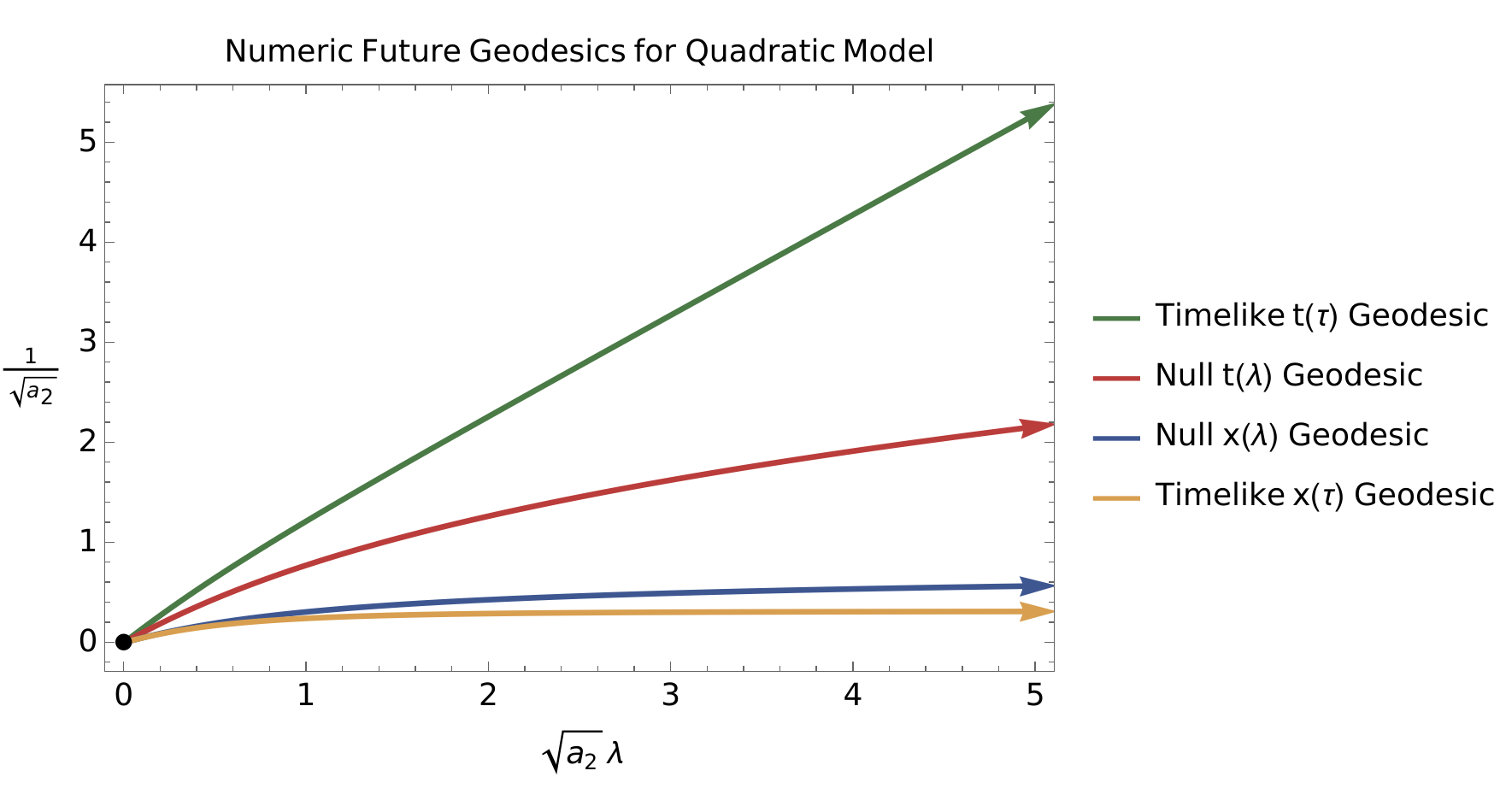}
		\caption{An additional plot of future null and timelike geodesic rays of the model and computations of Fig. \ref{fig:appendix example numeric geodesics quadratic}.}
        \label{fig:appendix example numeric future geodesics quadratic1}
	\end{figure}
    \begin{figure}
		\centering
		\includegraphics[width=\linewidth]{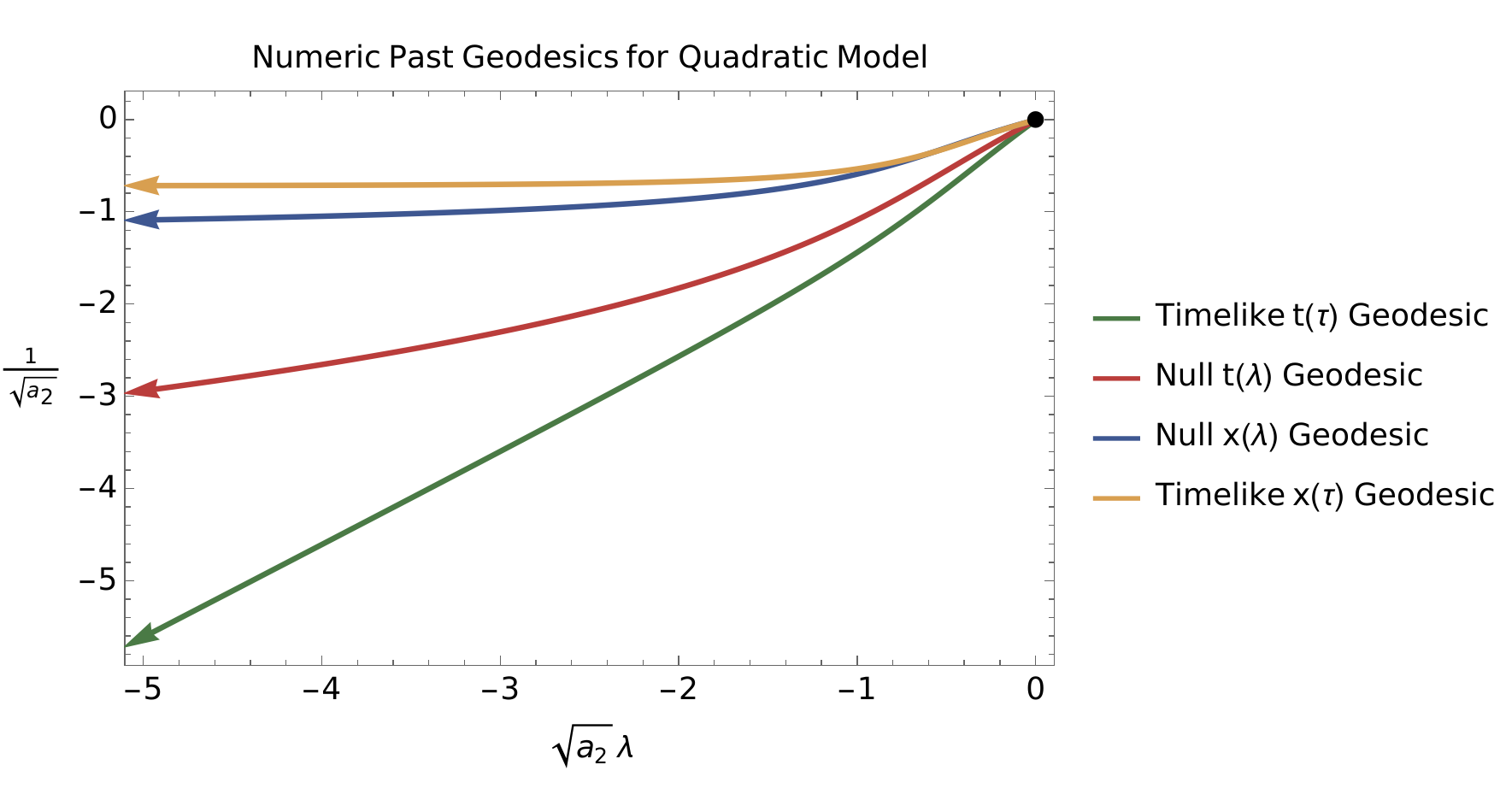}
		\caption{An additional plot of past null and timelike geodesic rays of the model and computations of Fig. \ref{fig:appendix example numeric geodesics quadratic}.}
        \label{fig:appendix example numeric past geodesics quadratic1}
	\end{figure}	
    \begin{figure}
		\centering
		\includegraphics[width=\linewidth]{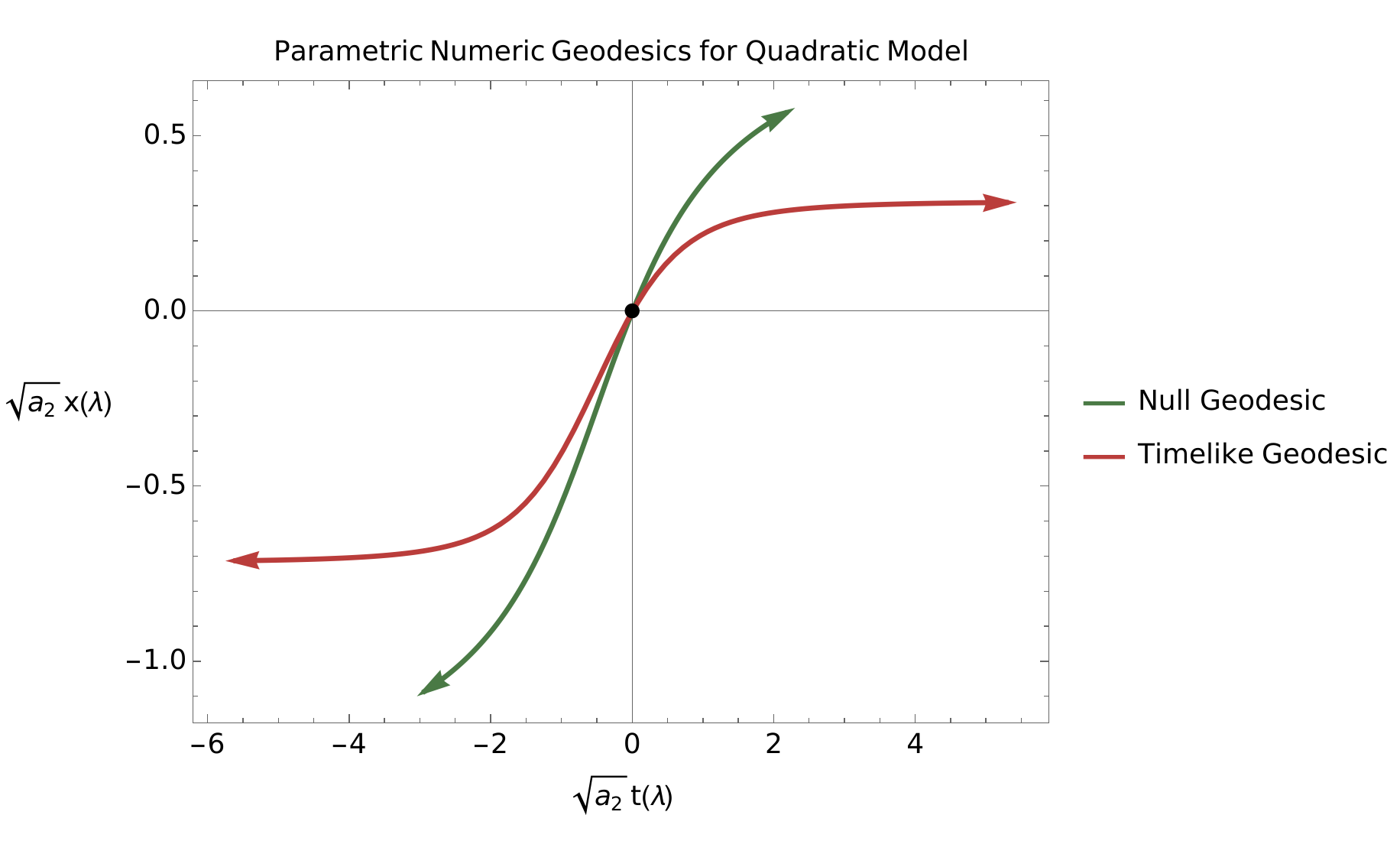}
		\caption{Parametric plot $x$ vs $t$ of null and timelike geodesics of $\mathbb{R}^1_1 \times_{a_2 t^2 + a_1 t + a_0} \mathbb{R}^3$ with $a_2 = a_1 = 1$ and $a_0 = 2$.  Null affine parameter is given by $\lambda$ and timelike proper time is given by $\tau$.  All length units have been renormalized to be dimensionless multiplying by canonical length scale $\sqrt{a_2}$.  The appropriate affine parameter was evolved for $\sqrt{a_2}\lambda, \sqrt{a_2}\tau \in \left[ -5 , 5 \right]$.  Initial null velocity is $\left. \left[ 1 , \frac{1}{2} , 0 , 0 \right]^T \right|_{\lambda = 0}$ and initial timelike velocity is $\left. \left[ \sqrt{2} , \frac{1}{2} , 0 , 0 \right]^T \right|_{\tau = 0}$.  The shown plot is a cylinder over constant coordinates $y,z$, hence two spacetime dimensions have been suppressed.  Light cones can be seen as the area under the null curve.  Appropriate causal arclength was calculated to diverge as $\lambda, \tau \rightarrow \infty$; thus this model is geodesically complete.  However, $H^-_{avg} = 0$ for computation over first limit asymptotic past $t \rightarrow - \infty$, $H^+_{avg} = 0$ for first limit asymptotic future $t \rightarrow + \infty$, and $H_{avg} > 0$ for a symmetric compact interval limit computations - see \cite{Lesnefsky:2022fen} for discussion of $H_{avg}$ limit order concerns. 
		Utilizing the interval $\left[ -10 , 10 \right]$, one computes $H_{avg}^{\left[ -10 , 10 \right]} = \frac{\ln 28 - \ln 23}{20} \simeq 0.009835$.}
		\label{fig:appendix example numeric parametric geodesics quadratic}
	\end{figure}
    \begin{figure}
		\centering
		\includegraphics[width=\linewidth]{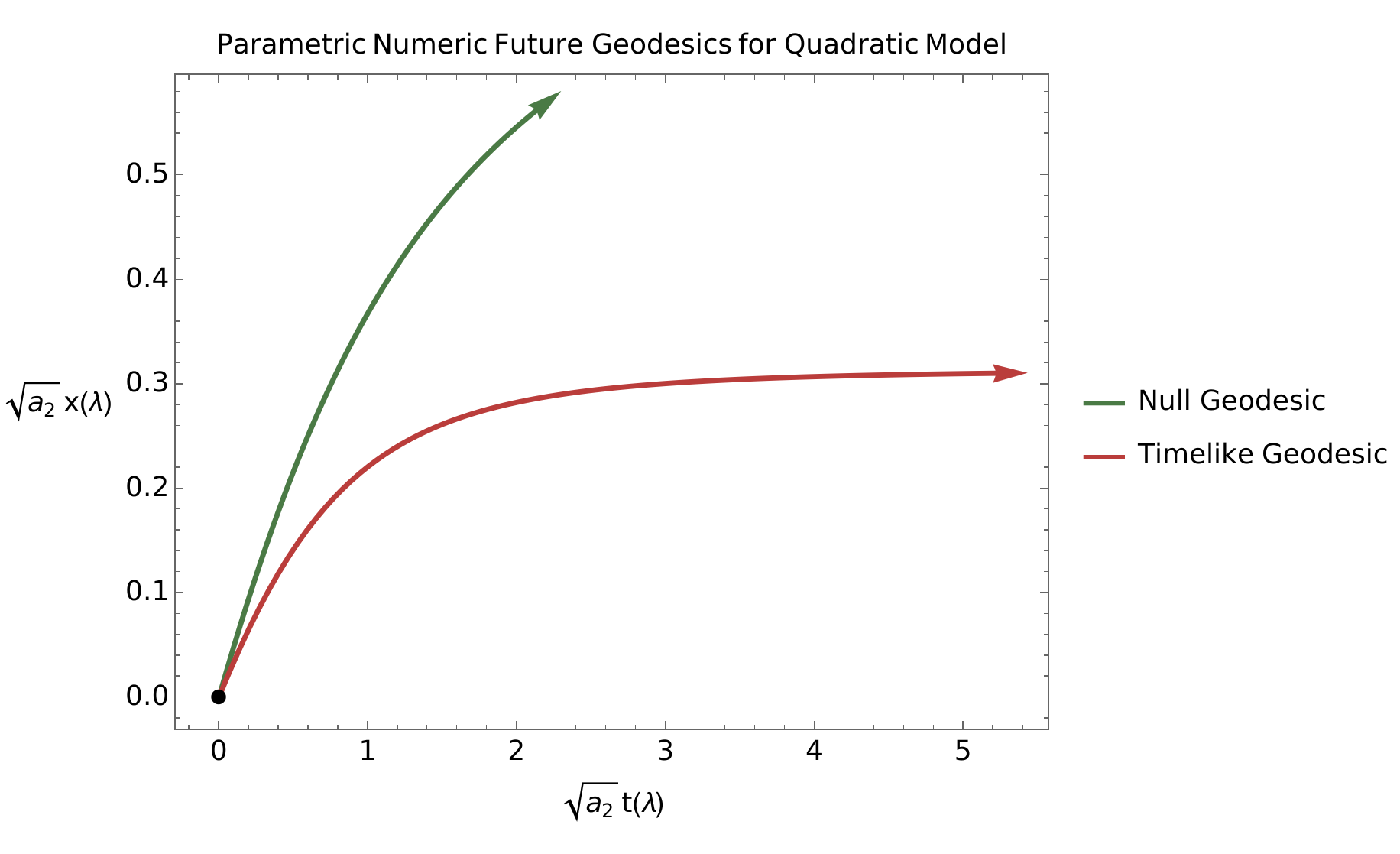}
		\caption{An additional plot of future null and timelike geodesic rays of the model and computations of Fig. \ref{fig:appendix example numeric parametric geodesics quadratic} with $\sqrt{a_2}\lambda, \sqrt{a_2}\tau \in \left[ 0 , 5 \right]$.}
        \label{fig:appendix example numeric parametric future geodesics quadratic}
	\end{figure}
    \begin{figure}
		\centering
		\includegraphics[width=\linewidth]{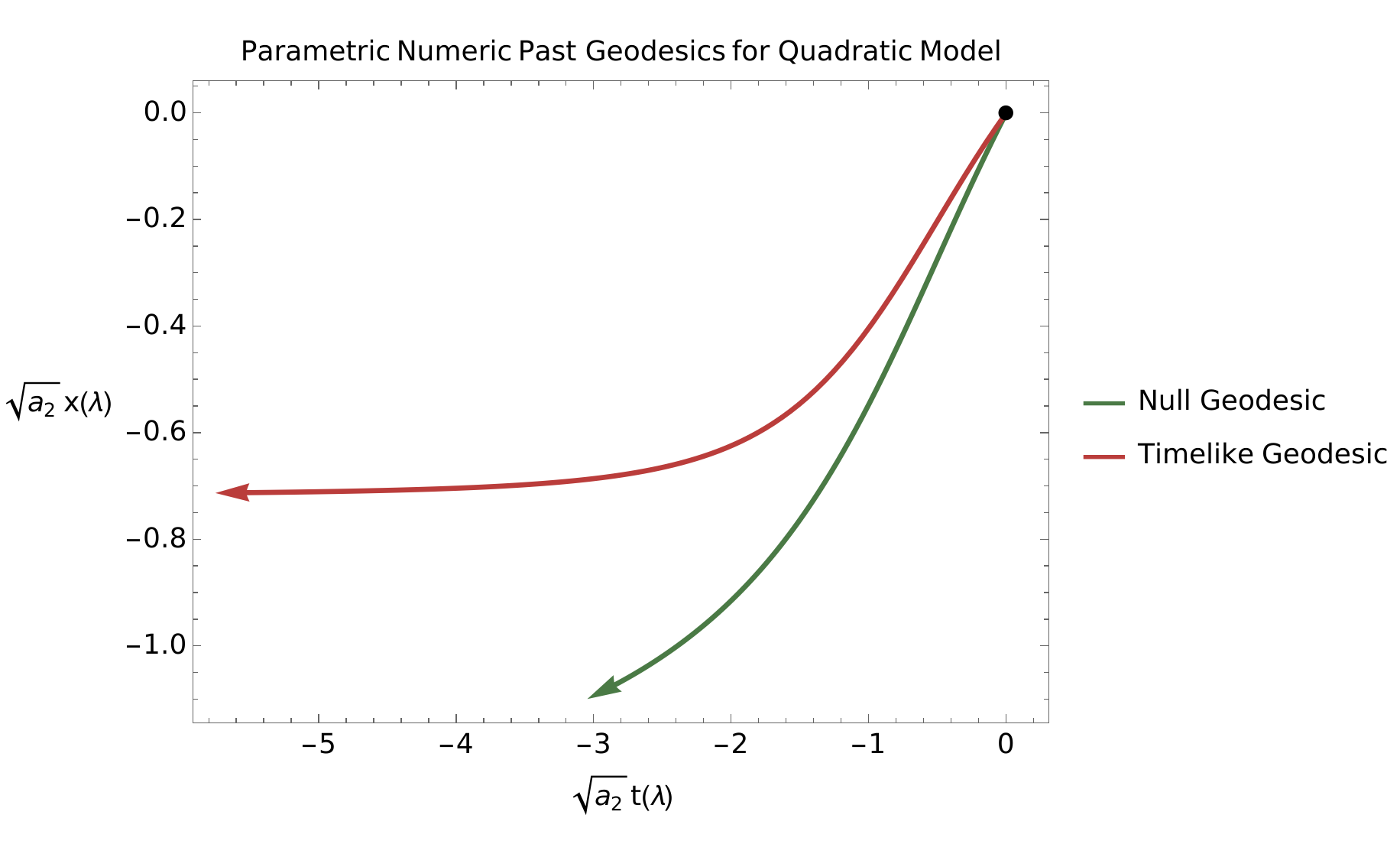}
		\caption{An additional plot of past null and timelike geodesic rays of the model and computations of Fig. \ref{fig:appendix example numeric parametric geodesics quadratic} with $\sqrt{a_2}\lambda, \sqrt{a_2}\tau \in \left[ -5 , 0 \right]$.}
        \label{fig:appendix example numeric parametric past geodesics quadratic}
	\end{figure}
    
We would like to remind the reader of some $H_{avg}$ considerations a propos this model discussed in Sec. \ref{sec: quadratic model}.  Investigated at length in \cite{Lesnefsky:2022fen}, the interval and order in which the limits are taken in the calculation of generalized Hubble parameter
\begin{equation}
    H^{\gamma} \left( \dot{\gamma} , U \right) = - \frac{g \left( \dot{\gamma} , \nabla_{\dot{\gamma}} U \right)}{g \left( U , \dot{\gamma} \right) + \kappa} \label{eq:appendix H BGV geometer}
\end{equation}
over geodesic $\gamma$ measuring a particle on observer frame field $U$ with $\kappa = \left\{ -1 , 0 \right\}$ for \{timelike, null\} $\gamma$, respectively.  This can be averaged over $\gamma$ by
\begin{equation}
H_{avg}^{\gamma}=\frac{1}{\lambda_{f}-\lambda_{i}}\int_{\lambda_{i}}^{\lambda_{f}}H^{\gamma}\left(\alpha\right)d\alpha\label{eq:appendix havg bgv}
\end{equation}
for a finite interval $\left[ \lambda_i , \lambda_f \right]$.  However, if one would like to calculate $H_{avg}$ over all $\mathbb{R}$ limits must be taken:
\begin{equation}
    H^{\pm,\gamma}_{avg}=\lim_{-a,b\rightarrow\infty}\frac{1}{b-a}\int_{a}^{ b}H^{\gamma}\left(\alpha\right)d\alpha\label{eq:appendix Havg asym full}.
\end{equation}
The results of Eq. \ref{eq:appendix havg bgv}, \ref{eq:appendix Havg asym full} can be highly dependent on the interval $\left[ \lambda_i , \lambda_f \right]$ chosen or the order the limits are taken.  A prime example of this is the quadratic model at hand.  A cursory examination Fig. \ref{cepoly} reveals greater measure of $f$ over $\mathbb{R}^+$ compared to $\mathbb{R}^-$ due to the anisotropy of the linear term, so one would expect $H_{avg} > 0$ if computed over a symmetric interval $\left[ -\lambda_0 , \lambda_0 \right]$ and thus geodesic incomplete by the results of \cite{Borde:2001nh}.  This scale factor, however, yields a complete GFRW per Thm. \ref{ledthm}: a counter-example!  However, if one utilizes Eq. \ref{eq:appendix Havg asym full} one reaps $H_{avg}^\pm = 0$ because the (small) linear anisotropy is quenched by $\frac{1}{\lambda_f - \lambda_i}$ as either $\lambda_i \rightarrow - \infty$ or $\lambda_f \rightarrow + \infty$.  Thus in this calculation $H_{avg}^\pm = 0$ yet the GFRW is complete by Thm. \ref{ledthm}.
To make matters worse one could consider the scale factor $f\left( t \right) = t^2 + t + \frac{1}{4}$.  This scale factor has a single root at $t = - \frac{1}{2}$ which means a GFRW with this scale factor is (trivially) geodesically incomplete.  However, by inspection, calculation of $H_{avg} > 0$ over symmetric interval $\left[ -\lambda_0 , \lambda_0 \right]$ and $H_{avg}^\pm = 0$ by Eq. \ref{eq:appendix Havg asym full}.  Thus the results of \cite{Borde:2001nh} does not elucidate any meaningful information about geodesic completeness of inflationary spacetimes.  However, Thm. \ref{ledthm}, which is a direct solution of the geodesic equation, does!

\subsubsection{Analytic Geodesics of the Quadratic Model}
Again utilizing Eq. \ref{eq: appendix example geodesics geodesic as t} for GFRW model $\mathbb{R}^1_1 \times_{a_2 t^2 + a_1 t + a_0} \mathbb{R}^3$ one reaps geodesic ODE
\begin{equation} \label{eq: appendix quadratic analytic x t geodesic ode}
    0 = x''(t) + \frac{\left(2 a_2 t+a_1\right) x'(t) \left(2-\left(t \left(a_2 t+a_1\right)+a_0\right)^2 x'(t)^2\right)}{t \left(a_2 t+a_1\right)+a_0}
\end{equation}
assuming WLOG $y\left( t \right) = z \left( t \right) = 0$. An analytic solution exists\footnote{Mathematica 13.3.0.0 code utilized to solve ODE is \texttt{DSolve}$\bigl[ '\left( \mathrm{insert \; RHS \; Eq. \; \ref{eq: appendix quadratic analytic x t geodesic ode}} \right)' == 0 $ , \texttt{x[t]} , \texttt{t} , \texttt{Assumptions}$\rightarrow t \in \mathbb{R}$ ,  \texttt{Assumptions}$\rightarrow a_1^2 - 4a_2 a_0 < 0$ , \texttt{Assumptions}$\rightarrow a_2 \in \mathbb{R}$ , \texttt{Assumptions}$\rightarrow a_1 \in \mathbb{R}$ , \texttt{Assumptions}$\rightarrow a_0 \in \mathbb{R} \bigr]$ }.  Due to its length - 7 pages - only a small portion will be shown; please contact the authors if you are interested in the full solution.
\begin{align} \label{eq:appendix quadratic anal solution schematic}
    x \left( t \right) & = \left(t-\frac{-a_1+\frac{\sqrt{a_1^2 \omega_0-4 a_2 \sqrt{-\omega_0}-4 a_0 a_2 \omega_0}}{\sqrt{\omega_0}}}{2 a_2}\right)^2 \times  \cdots \times\left( \begin{matrix} \mathrm{complex \; solution \; of \; elliptic \;functions \;} \\ \mathrm{ of \; variables \;} a_2 , a_1 , a_0 , \omega_0 , t \end{matrix}\ \right) \nonumber \\
    & \qquad \;\times \cdots \times \left( \frac{-a_1+\frac{\sqrt{a_1^2 \omega_0-4 a_2 \sqrt{-\omega_0}-4 a_0 a_2 \omega_0}}{\sqrt{\omega_0}}}{2 a_2}-\frac{-a_1+\sqrt{a_1^2+\frac{4 a_2 \left(\sqrt{-\omega_0}-a_0 \omega_0\right)}{\omega_0}}}{2 a_2} \right) + x_0
\end{align}
with $\omega_0$ being the integration constant corresponding to causal character.  Utilizing Eq. \ref{eq: appendix example geodesics omega 0 general GFRW calculation} one reaps $\omega_0 = \nicefrac{1}{4}$ corresponding to a timelike geodesic for initial conditions of $t_0=0$, $a_2 = a_1 = 1$, $a_0=2$.  Additionally, $x\left( t_0 =0 \right) = 0$ demands $x_0 \simeq 1.322$.  Plot of this curve is shown in Figs. \ref{fig:appendix example anal geodesics quadratic}, \ref{fig:appendix example anal future geodesics quadratic}, \ref{fig:appendix example anal past geodesics quadratic}, \ref{fig:appendix example anal lightcone geodesics quadratic}.
	\begin{figure}
		\centering
		\includegraphics[width=\linewidth]{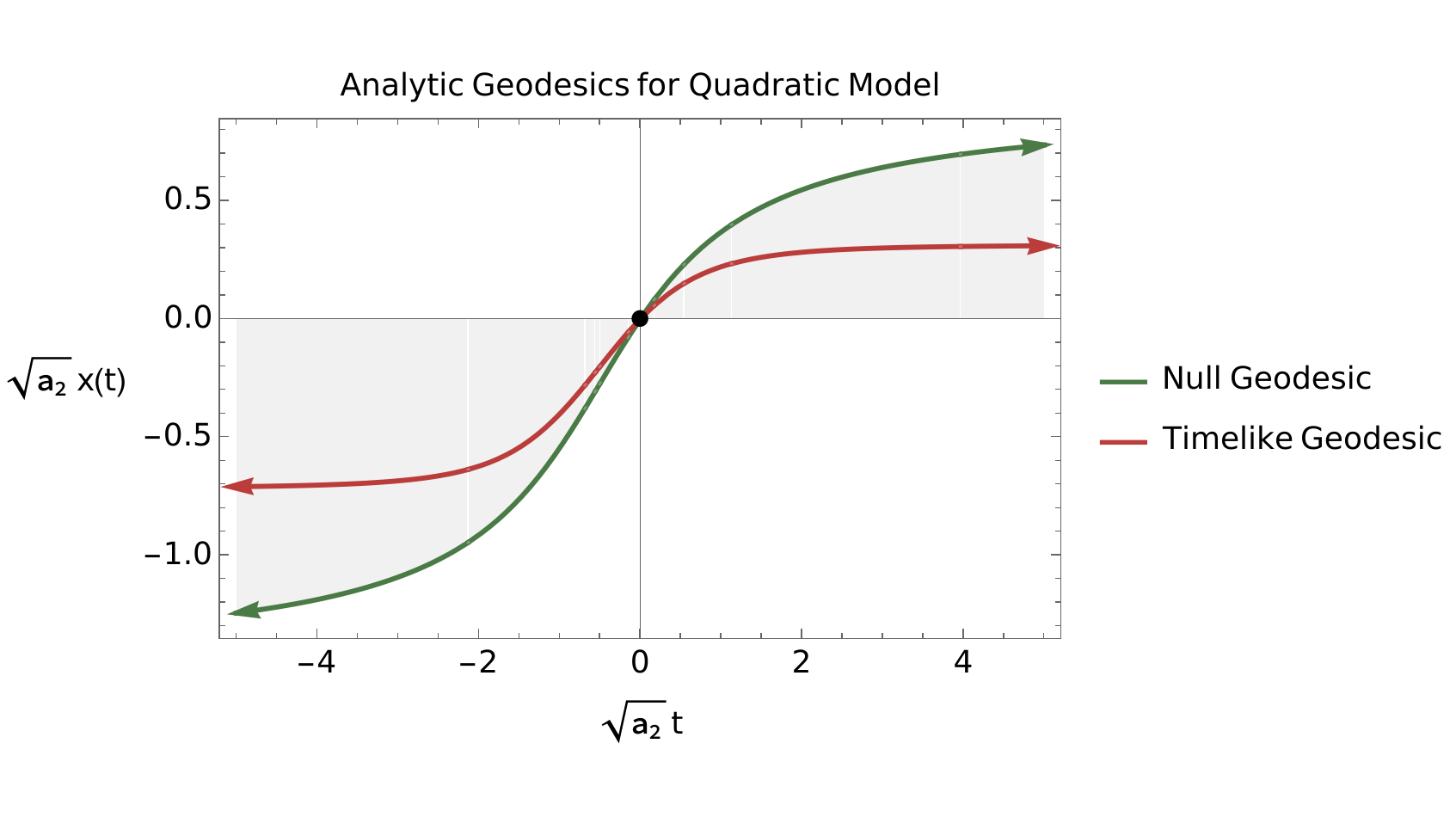}
		\caption{Plot $x\left( t\right)$ vs $t$ of null and timelike geodesics of $\mathbb{R}^1_1 \times_{a_2 t^2 + a_1 t + a_0 } \mathbb{R}^3$ with $a_2 = a_1 = 1$ and $a_0 = 2$.  Independent variable is the time coordinate.  All length units have been renormalized to be dimensionless rescaling by canonical length scale $\sqrt{a_2}$.  Initial null velocity is derived from frame $\left. \left[ 1 , \frac{1}{2} , 0 , 0 \right]^T \right|_{t = 0}$ and initial timelike velocity is derived from frame $\left. \left[ \sqrt{2} , \frac{1}{2} , 0 , 0 \right]^T \right|_{t = 0}$ a propos Eq. \ref{eq: appendix example geodesics omega 0 general GFRW calculation}.  The shown plot is a cylinder over constant coordinates $y,z$, hence two spacetime dimensions have been suppressed.  Partial light cones can be seen as the area under the null curve shaded light gray.  Small artifacts on plot are due to partial resolution of branch cuts in the complexified full analytic solution\footnote{Full resolution of branch cuts of the analytic solution of Eq. \ref{eq:appendix quadratic anal solution schematic} is beyond the scope of this paper: this is left as an exercise to the reader.  We would like to remind the reader, however, that Thm. \ref{ledthm} guarantees geodesic completeness of the solution, augmented by the complete numeric solutions of Figs. \ref{fig:appendix example numeric parametric geodesics quadratic}, \ref{fig:appendix example numeric parametric future geodesics quadratic}, \ref{fig:appendix example numeric parametric past geodesics quadratic}.  See \cite{Tupper2001} for more information.}.  Appropriate causal arclength was calculated to diverge as $t \rightarrow +\infty$; thus this model is geodesically complete. However, $H^-_{avg} = 0$ for computation over first limit asymptotic past $t \rightarrow - \infty$, $H^+_{avg} = 0$ for first limit asymptotic future $t \rightarrow + \infty$, and $H_{avg} > 0$ for a symmetric compact interval limit computations - see \cite{Lesnefsky:2022fen} for discussion of $H_{avg}$ limit order concerns. For the interval $\left[ -10 , 10 \right]$, one computes $H_{avg}^{\left[ -10 , 10 \right]} = \frac{\ln 28 - \ln 23}{20} \simeq 0.009835$.}
		\label{fig:appendix example anal geodesics quadratic}
	\end{figure}
    
    \begin{figure}
		\centering
		\includegraphics[width=\linewidth]{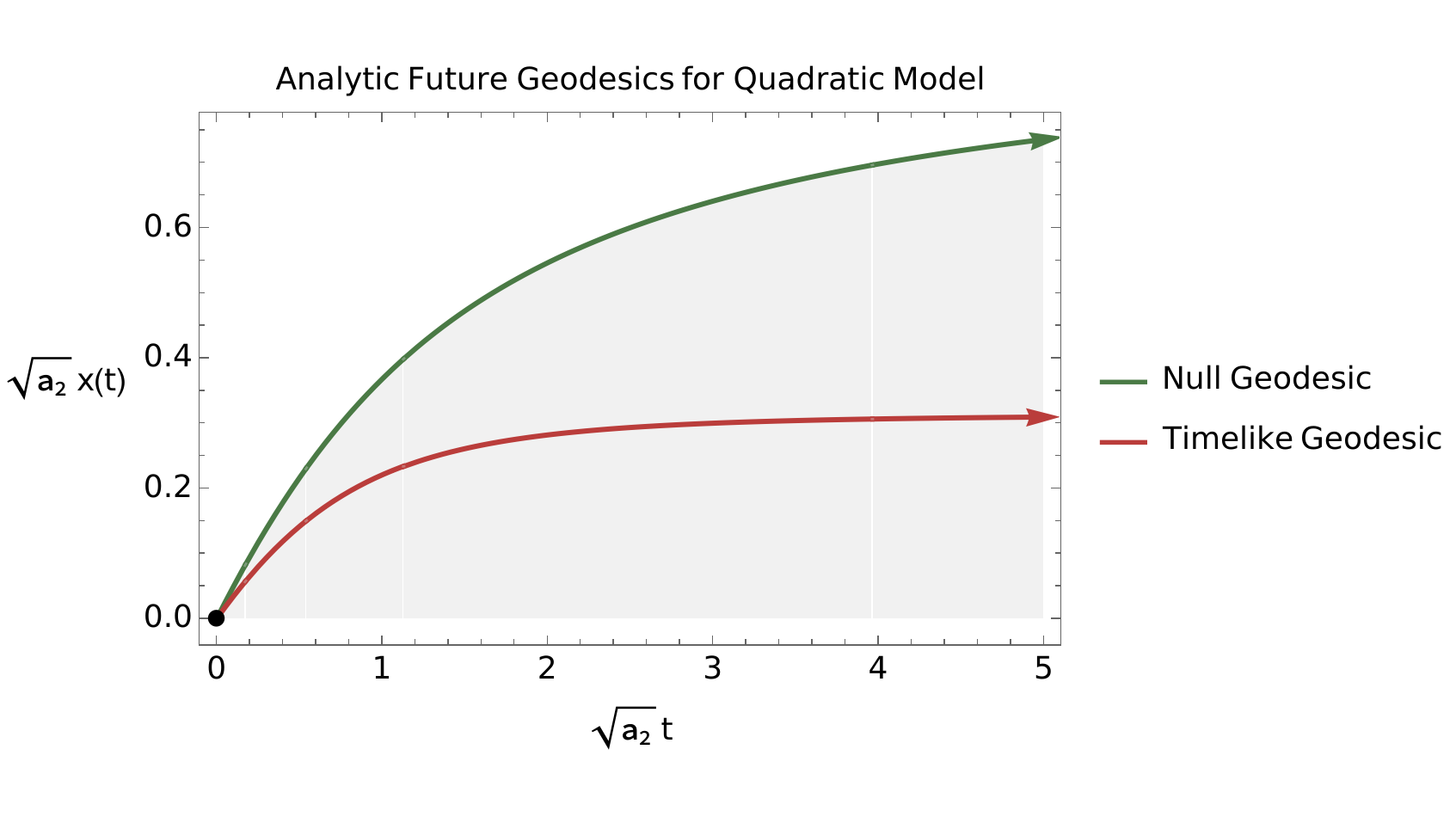}
		\caption{Additional future geodesic rays of the model and calculations of Fig. \ref{fig:appendix example anal geodesics quadratic}.}
		\label{fig:appendix example anal future geodesics quadratic}
	\end{figure}
    \begin{figure}
		\centering
		\includegraphics[width=\linewidth]{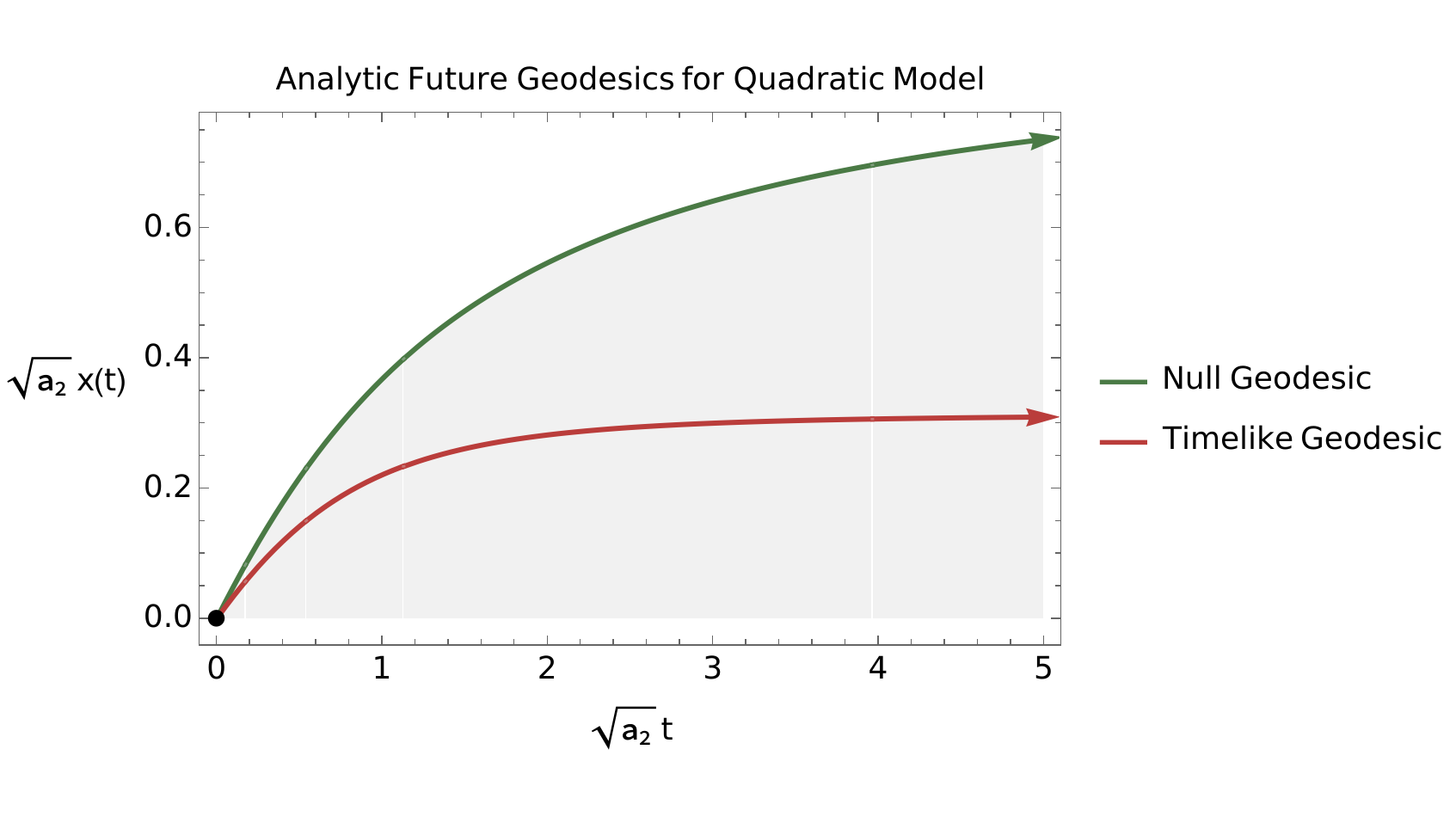}
		\caption{Additional past geodesic rays of the model and calculations of Fig. \ref{fig:appendix example anal geodesics quadratic}.}
		\label{fig:appendix example anal past geodesics quadratic}
	\end{figure}
    
    \begin{figure}
		\centering
		\includegraphics[width=\linewidth]{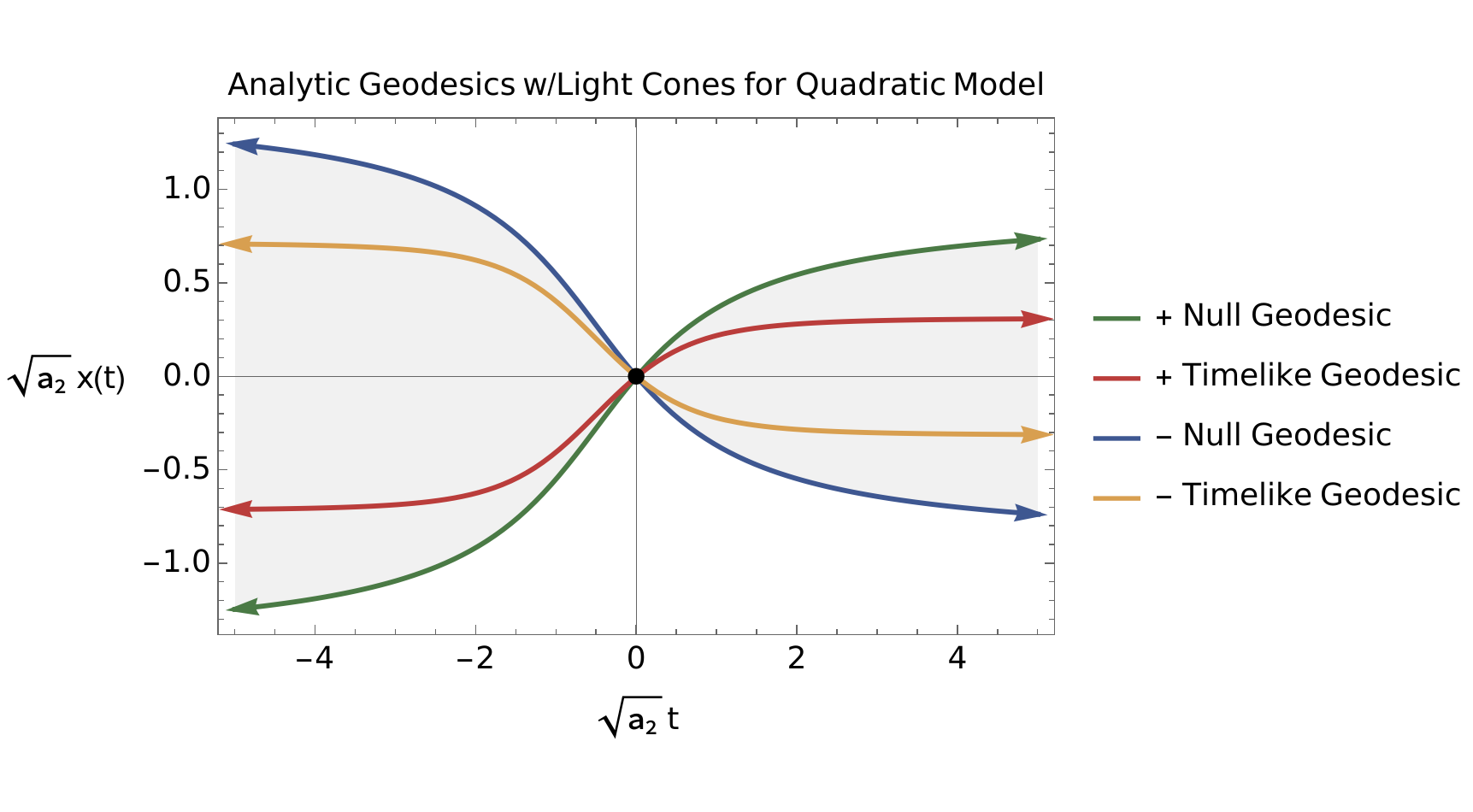}
		\caption{Additional geodesic plots of the model and calculations of Fig. \ref{fig:appendix example anal geodesics quadratic}.  Additional null and timelike geodesics were plotted as to show the light cone cylinder $\Lambda_{\left( 0 , 0,y,z \right)} \mathbb{R}^1_1 \times_{a_2 t^2 + a_1 t + a_0} \mathbb{R}^3$ over cylinder $t=0 , x=0 , y , z$.}
		\label{fig:appendix example anal lightcone geodesics quadratic}
	\end{figure}
    
Analytic null geodesics for the quadratic model can be calculated in a similar fashion to Sec. \ref{sec: appendix analytic geodesic +c}.  Geodesic null ODE can be calculated as
\begin{equation} \label{eq:appendix anal null geodesic quadratic model}
    \dot{x}\left( t\right) = \frac{1}{f\left( t\right)} = \frac{1}{a_2 t^2 + a_1 t +a_0}
\end{equation}
with dimensionless solution
\begin{equation} \label{eq:appendix anal null geodesic soltn quadratic model}
    \sqrt{a_2} x \left( \sqrt{a_2} t \right) = \frac{2 \arctan\left(\frac{2 \sqrt{a_2} t + \frac{a_1}{\sqrt{a_2}}}{\sqrt{4 a_0 - \frac{a_1^2}{a_2}}}\right)}{\sqrt{4 a_0 - \frac{a_1^2}{a_2}}} + \sqrt{a_2}\chi_0
\end{equation}
with integration constant $\chi_0$ corresponding to initial position $x_0$ vis a vis
\begin{equation}
    \chi_0 = x_0 -\frac{2 \arctan\left(\frac{2 \sqrt{a_2} t_0 + \frac{a_1}{\sqrt{a_2}}}{\sqrt{4 a_0 - \frac{a_1^2}{a_2}}}\right)}{\sqrt{a_2}\sqrt{4 a_0 - \frac{a_1^2}{a_2}}}
\end{equation}
An analytic null geodesic solution for the quadratic model with parameters $a_2 = a_1 = 1$ and $a_0 = 2$ is shown in Fig.  \ref{fig:appendix example anal geodesics quadratic}.

\newpage
\bibliography{physics}
\end{document}